\documentclass[3p,11pt,fleqn,sort&compress]{elsarticle}

\usepackage{multirow}
\usepackage{footmisc}
\usepackage{tikz}
\usetikzlibrary{calc}
\usepackage{subcaption}

\usepackage{amsmath,amssymb,amsthm}

\usepackage[T1]{fontenc}
\usepackage[utf8]{inputenc}

\usepackage{ifthen}
\newboolean{colorChanges}
\setboolean{colorChanges}{false}

\ifthenelse{\boolean{colorChanges}}{

\newcommand{\LinhModified}[1]{\textcolor{blue}{#1}\xspace}
\newcommand{\markLinhModified}{\color{blue}}

}{

\newcommand{\LinhModified}[1]{#1\xspace}
\newcommand{\markLinhModified}{}

}

\theoremstyle{plain}
\newtheorem{theorem}{Theorem}[section]
\newtheorem{lemma}[theorem]{Lemma}
\newtheorem{corollary}[theorem]{Corollary}

\newtheorem{proposition}[theorem]{Proposition}

\theoremstyle{definition}
\newtheorem{definition}[theorem]{Definition}
\newtheorem{example}[theorem]{Example}

\newtheorem{remark}[theorem]{Remark}

\usepackage{tabularx}
\usepackage{nicematrix}
\usepackage{booktabs}

\DeclareFontFamily{OT1}{pzc}{}
\DeclareFontShape{OT1}{pzc}{m}{it}{<-> s * [1.200] pzcmi7t}{}
\DeclareMathAlphabet{\mathpzc}{OT1}{pzc}{m}{it}
\renewcommand{\leq}{\leqslant}
\renewcommand{\geq}{\geqslant}
\renewcommand{\le}{\leqslant}

\def\eqref#1{(\ref{#1})}
\def\defeq{\stackrel{\mathrm{def}}{=}}
\def\tuple#1{\langle#1\rangle}

\newcommand{\mL}{\mathcal{L}}
\newcommand{\mLe}{\mathcal{L}_\varepsilon}

\newcommand{\bL}{\mathbf{L}}
\newcommand{\bLdb}{\bL^{\leq k}}

\newcommand{\myend}{\mbox{}\hfill{\small$\Box$}}

\newcommand{\cnv}[1]{{#1}^{-1}}

\newcommand{\mA}{\mathpzc{A}}
\newcommand{\mAp}{{\mathpzc{A}'}}

\newcommand{\mB}{\mathpzc{B}}

\newcommand{\mF}{\mathpzc{F}}
\newcommand{\dF}{d\!\mathpzc{F}}

\newcommand{\QA}{Q}
\newcommand{\TA}{\delta}
\newcommand{\IA}{I}
\newcommand{\FA}{F}
\newcommand{\QAp}{Q'}
\newcommand{\TAp}{\delta'}
\newcommand{\IAp}{I'}
\newcommand{\FAp}{F'}

\newcommand{\lee}{\le_\varepsilon}
\newcommand{\leep}{\le_{\varepsilon'}}
\newcommand{\eqe}{=_\varepsilon}
\newcommand{\eqep}{=_{\varepsilon'}}
\newcommand{\otimese}{\otimes_{\!\varepsilon}}
\newcommand{\toe}{\to_{\!\varepsilon}}
\newcommand{\circe}{\,\circ_{\!\varepsilon\,}}
\newcommand{\circep}{\,\circ_{\!\varepsilon'\,}}
\newcommand{\lande}{\land_\varepsilon}
\newcommand{\lore}{\lor_{\!\varepsilon}}
\newcommand{\bigveee}{\textstyle\bigvee_{\!\!\varepsilon}}
\newcommand{\bigwedgee}{\textstyle\bigwedge_\varepsilon}
\newcommand{\setminuse}{\!\setminus_{\varepsilon\!}}
\newcommand{\slashe}{/_{\!\!\varepsilon\,}}

\newcommand{\NN}{\mathbb{N}}
\newcommand{\FPO}{FPO\xspace}
\newcommand{\FPOs}{FPOs\xspace}

\newcommand{\FfA}{FfA\xspace}
\newcommand{\FfAs}{FfAs\xspace}
\newcommand{\DfA}{DfA\xspace}
\newcommand{\DfAs}{DfAs\xspace}
\newcommand{\NfA}{NfA\xspace}
\newcommand{\NfAs}{NfAs\xspace}

\newcommand{\True}{\textit{true}}
\newcommand{\lfdeg}{\textit{lfdeg}}

\makeatletter
\newcommand{\dbloverline}[1]{\overline{\dbl@overline{#1}}}
\newcommand{\dbl@overline}[1]{\mathpalette\dbl@@overline{#1}}
\newcommand{\dbl@@overline}[2]{%
  \begingroup
  \sbox\z@{$\m@th#1\overline{#2}$}%
  \ht\z@=\dimexpr\ht\z@-2\dbl@adjust{#1}\relax
  \box\z@
  \ifx#1\scriptstyle\kern-\scriptspace\else
  \ifx#1\scriptscriptstyle\kern-\scriptspace\fi\fi
  \endgroup
}
\newcommand{\dbl@adjust}[1]{%
  \fontdimen8
  \ifx#1\displaystyle\textfont\else
  \ifx#1\textstyle\textfont\else
  \ifx#1\scriptstyle\scriptfont\else
  \scriptscriptfont\fi\fi\fi 3
}
\makeatother

\usepackage{tikz}    % this is for graphics. e.g. rectangle on title page
\usetikzlibrary{shapes}
\usetikzlibrary{fit}

\usepackage{tikz-3dplot,pgfplots}
\pgfplotsset{compat=newest}
\usetikzlibrary{shapes.geometric,arrows,arrows.meta}
\usetikzlibrary{fadings}
\tikzfading %strangely gives bad bounding box when inside the tikzpicture
[
  name=fade out,
  inner color=transparent!0,
  outer color=transparent!100
]
\usetikzlibrary{arrows.meta,automata,positioning,shadows}

\tikzset{
    invisible/.style={opacity=0},
    visible on/.style={alt={#1{}{invisible}}},
    alt/.code args={<#1>#2#3}{%
      \alt<#1>{\pgfkeysalso{#2}}{\pgfkeysalso{#3}} % \pgfkeysalso doesn't change the path
    },
  }
\usetikzlibrary{calc,shapes}

\usepackage[ruled,vlined,linesnumbered]{algorithm2e}
\usepackage{multicol}
\RestyleAlgo{ruled}
\SetKwComment{Comment}{/* }{ */}
\SetKwInput{Input}{Input}
\SetKwInput{Output}{Output}
\SetKwInput{Purpose}{Purpose}
\SetKw{Break}{break}
\SetKw{Continue}{continue}
\SetKw{Return}{return}
\SetKw{New}{new\,}
\SetKwFor{RepTimes}{repeat}{times}{end}

\newcommand{\ReductionByRightInvariance}{\mbox{$\mathsf{ReductionByRightInvariance}$}\xspace}
\newcommand{\SoftStateReductionZ}{\mbox{$\mathsf{SoftStateReduction}_0$}\xspace}
\newcommand{\SoftStateReduction}{\mbox{$\mathsf{SoftStateReduction}$}\xspace}

\journal{arXiv}

\begin{document}
\sloppy

\begin{frontmatter}

\title{Soft state reduction of fuzzy automata over residuated lattices}

\author[1,2]{Linh Anh Nguyen}
\ead{nguyen@mimuw.edu.pl}
\ead{nalinh@ntt.edu.vn}
\address[1]{Faculty of Mathematics, Informatics and Mechanics, University of Warsaw}
\address[2]{Faculty of Information Technology, Nguyen Tat Thanh University}

\author[3]{Son Thanh Cao}
\ead{sonct@vinhuni.edu.vn}
\address[3]{Faculty of Information Technology, School of Engineering and Technology, Vinh University}

\author[4]{Stefan Stanimirovi\'{c}}
\ead{stefan.stanimirovic@pmf.edu.rs}
\address[4]{University of Ni\v s, Faculty of Sciences and Mathematics}

\begin{abstract}
State reduction of finite automata plays a significant role in improving efficiency in formal verification, pattern recognition, and machine learning, where automata-based models are widely used. While deterministic automata have well-defined minimization procedures, reducing states in nondeterministic fuzzy finite automata (\FfAs) remains challenging, especially for FfAs over non-locally finite residuated lattices like the product and Hamacher structures. This work introduces soft state reduction, an approximate method that leverages a small threshold~$\varepsilon$ possibly combined with a word length bound~$k$ to balance reduction accuracy and computational feasibility. By omitting fuzzy values smaller than~$\varepsilon$, the underlying residuated lattice usually becomes locally finite, making computations more tractable. We introduce \LinhModified{and study} approximate invariances\LinhModified{, which are fuzzy relations that allow merging of ``almost equivalent'' states of an FfA up to a tolerance level $\varepsilon$ and, optionally, to words of bounded length $k$. We further} present an algorithm which iteratively applies \LinhModified{such} invariances to achieve reduction while preserving approximate language equivalence. Our method effectively reduces \FfAs where existing techniques fail.
\end{abstract}


\end{frontmatter}

\section{Introduction}

The problem of reducing the number of states in a finite automaton is a fundamental topic in automata theory, with applications spanning formal verification~\cite{DBLP:books/sp/Milner80,DBLP:books/el/01/Glabbeek01}, pattern recognition~\cite{DBLP:journals/pieee/Rabiner89,DBLP:books/cu/MBCT2015}, and text processing~\cite{DBLP:books/aw/HopcroftU79,DBLP:journals/coling/Mohri97}. In the case of deterministic finite automata (\DfAs), a well-established minimization procedure ensures that every \DfA has a unique  smallest equivalent representation, up to isomorphism. For nondeterministic finite automata (\NfAs), the situation is considerably more challenging. Unlike \DfAs, \NfAs do not necessarily have a unique minimal form, and determining the smallest equivalent \NfA~\cite{KamedaW70} is a computationally hard problem. In fact, finding a minimal \NfA is PSPACE-complete~\cite{Jiang1993}. Consequently, researchers often focus on reducing the state space of \NfAs without necessarily achieving strict minimality (see, e.g., \cite{IlieY03,BIANCHINI2024114621} and therein references). 

Fuzzy automata generalize classical automata by allowing fuzzy values for transitions, initial states and final states. This makes them particularly useful in domains requiring graded or uncertain decision-making, such as pattern recognition~\cite{DBLP:journals/pieee/Rabiner89,DBLP:books/cu/MBCT2015} and machine learning~\cite{mohri2002weighted,de2010grammatical}. As in the classical setting, state reduction of fuzzy finite automata (\FfAs) is an important research topic that has been studied extensively. 
Most works focus either on finding a minimal deterministic \FfA \cite{LI20071423,Belohlvek2009OnAM,Halamish2015,YANG202172,DEMENDIVILGRAU2024109108,ShamsizadehZG24} or on reducing the state space of a nondeterministic \FfA \cite{MALIK1999323,BASAK2002223,CHENG2004439,LEI20071413,PEEVA20084152,WU20101635,CSIP.10,SCI.14,SCB.18,YangL.19,StanimirovicMC.22,EPTCS386.6,FSS-D-25-00029}, with few exceptions focusing on finding a minimal nondeterministic \FfA~\cite{LITFS2015}. 
Notably, in many studies on state reduction of nondeterministic \FfAs, the term ``minimal'' does not imply a minimal number of states. Instead, the employed techniques typically involve merging equivalent or indistinguishable states based on specific criteria. 
Some works~\cite{Belohlvek2009OnAM,YangL.19,EPTCS386.6,StanimirovicMC.22,FSS-D-25-00029} address a relaxed problem of finding a minimal or reduced \FfA that approximates a given one. 
A considerable number of studies~\cite{LEI20071413,PEEVA20084152,WU20101635,CSIP.10,SCI.14,SCB.18,EPTCS386.6,DEMENDIVILGRAU2024109108,FSS-D-25-00029} focus on \FfAs over residuated lattices~\cite{Qiu01}, while others \cite{MALIK1999323,BASAK2002223,CHENG2004439,LI20071423,Belohlvek2009OnAM,LITFS2015,Halamish2015,YangL.19,YANG202172,StanimirovicMC.22,ShamsizadehZG24} consider only \FfAs over the G\"odel structure or Heyting algebras (of fuzzy values). 
The main results on state reduction of the works~\cite{LEI20071413,WU20101635} are formulated only for \FfAs over finite residuated lattices. 
\FfAs considered in the works~\cite{MALIK1999323,CHENG2004439,LEI20071413,WU20101635} are Mealy machines, which associate inputs with outputs but lack initial and final states. Table~\ref{table: JHJSM} summarizes the characteristics of notable works on state reduction of~\FfAs.

\begin{table}
\footnotesize
\centering
\begin{NiceTabular}{@{}p{16em}|p{2.6em}|p{1.2em}|p{2.6em}|p{2.8em}|p{2.4em}|p{1.2em}|p{2.8em}|p{1.2em}|p{1.2em}|p{1.7em}@{}}[]
\toprule
Articles & \!\!\cite{LEI20071413,WU20101635}\!\!\!\!\! & \!\!\cite{DEMENDIVILGRAU2024109108} & \!\!\cite{EPTCS386.6,FSS-D-25-00029,StanimirovicMC.22}\!\!\!\!\! & \!\!\cite{PEEVA20084152,CSIP.10,SCI.14,SCB.18}\!\!\!\!\! & \!\!\cite{MALIK1999323,CHENG2004439}\!\!\!\!\! & \!\!\cite{Belohlvek2009OnAM} & \!\!\cite{LI20071423,Halamish2015,YANG202172,ShamsizadehZG24}\!\!\!\!\! & \!\!\cite{LITFS2015} & \!\!\cite{YangL.19} & \!\!\cite{BASAK2002223} \\
\bottomrule
\FfAs over various residuated lattices & \checkmark & \checkmark & \checkmark$^{(**)}$ & \checkmark & & & & & & \\
\midrule
\FfAs over only finite residuated lattices & \checkmark$^{(*)}$ & & & & & & & & & \\
\midrule
\FfAs over only the G\"odel structure & & & & & \checkmark & \checkmark & \checkmark & \checkmark & \checkmark & \checkmark \\
\midrule
\FfAs of the Mealy type & \checkmark & & & & \checkmark & & & & & \\
\bottomrule
Finding minimal deterministic \FfAs & & \checkmark & & & & \checkmark & \checkmark & & & \\
\midrule
Finding minimal nondeterministic \FfAs & & & & & & & & \checkmark & & \\
\midrule
Reduction of nondeterministic \FfAs & \checkmark & & \checkmark & \checkmark & \checkmark & & & & \checkmark & \checkmark \\
\midrule
With approximation & & & \checkmark & & & \checkmark & & & \checkmark & \\
\bottomrule
\end{NiceTabular}
\caption{Characteristics of notable works on state reduction of \FfAs. (*):~The work~\cite{WU20101635} also studies \FfAs over infinite residuated lattices, but its main results on state reduction are provided only for \FfAs over finite residuated lattices. (**):~The work~\cite{StanimirovicMC.22} deals only with Heyting algebras.}\label{table: JHJSM}
\end{table}

This work focuses on approximate state reduction of nondeterministic \FfAs over linear and complete residuated lattices, including the product and Hamacher structures, which are not locally finite. To motivate this study, consider the state reduction problem for the \FfA $\mA$ over the alphabet $\Sigma = \{\sigma\}$ and the product structure, illustrated in Figure~\ref{fig: HEJSA} (and detailed in Example~\ref{example: HEJSA}). 
Existing methods from \cite{LEI20071413,PEEVA20084152,WU20101635,StanimirovicMC.22} are inapplicable in this case because they do not accommodate the product structure.\footnote{The main results of~\cite{PEEVA20084152} do not apply to ``max–product machines''.} As shown later (in Section~\ref{section: JHRJS}), the methods in~\cite{CSIP.10,SCI.14,SCB.18} fail to produce a reduced \FfA equivalent to~$\mA$. 
Moreover, the approximation approach of~\cite{EPTCS386.6,FSS-D-25-00029}, which relaxes the problem by considering equivalence only for words of a length bounded by a constant~$k$, is also ineffective. Specifically, for any $k \geq 4$, this approach does not reduce the state space of~$\mA$. Given these limitations, the central question arises: What new state reduction approaches can be applied to~$\mA$?

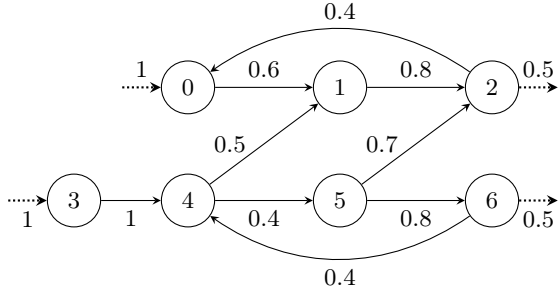
\begin{figure}[!h]
\begin{center}
\begin{tikzpicture}[->,>=stealth,auto,black]
\tikzset{every state/.style={inner sep=0.06cm,minimum size=0.7cm}}
\tikzstyle{every node}=[font=\footnotesize]
\node[state] (u) {$0$};
\node[state] (v) [node distance=2.0cm, right of=u] {$1$};
\node[state] (w) [node distance=2.0cm, right of=v] {$2$};
\node[state] (up) [node distance=1.5cm, below of=u] {$4$};
\node[state] (vp) [node distance=2.0cm, right of=up] {$5$};
\node[state] (wp) [node distance=2.0cm, right of=vp] {$6$};
\node[state] (tp) [node distance=1.5cm, left of=up] {$3$};
\node[left of=u] (in) {};
\node[right of=w] (out) {};
\node[left of=tp] (inp) {};
\node[right of=wp] (outp) {};
\draw[densely dotted, thick] (in) to node[above]{1} (u);
\draw[densely dotted, thick] (inp) to node[below]{1} (tp);
\draw[densely dotted, thick] (w) to node[above]{0.5} (out);
\draw[densely dotted, thick] (wp) to node[below]{0.5} (outp);
\draw (u) to node[above]{0.6} (v);
\draw (v) to node[above]{0.8} (w);
\draw (w) edge[above,out=145,in=35] node{0.4} (u);
\draw (up) to node[left,xshift=-3pt]{0.5} (v);
\draw (up) to node[below]{0.4} (vp);
\draw (vp) to node[left,xshift=-3pt]{0.7} (w);
\draw (vp) to node[below]{0.8} (wp);
\draw (wp) edge[below,out=-145,in=-35] node{0.4} (up);
\draw (tp) to node[below]{1} (up);
\end{tikzpicture}
\caption{An illustration of the fuzzy automaton $\mA$ discussed in the introduction section and used in Example~\ref{example: HEJSA}.\label{fig: HEJSA}}
\end{center}
\end{figure}

In this work, we introduce and study a novel form of approximate state reduction for nondeterministic \FfAs over linear and complete residuated lattices, which we call {\em soft state reduction}. Given an \FfA $\mA$, a small threshold $\varepsilon$ from the underlying residuated lattice and a natural number $k$, our goal is to construct a reduced \FfA $\mAp$ that is {\em $\varepsilon$-equivalent} (respectively, {\em $(\varepsilon,k)$-equivalent}) to $\mA$. This means that, for every word $w$ (respectively, every word $w$ of length at most~$k$) over the considered alphabet, either $\mAp$ and $\mA$ accept $w$ with the same degree, or both accept $w$ with degrees at most~$\varepsilon$. The key idea is that, when $\varepsilon$ is very small, fuzzy values below $\varepsilon$ may result from noise or numerical errors and can be  considered negligible. A significant technical advantage is that restricting the underlying residuated lattice to values greater than or equal to $\varepsilon$ can make it $\otimes$-locally finite, even if the original lattice is not, as in the case of the product and Hamacher structures. This approach is inspired by~\cite{MicicCMSN.24}, which explores approximate weak simulations and bisimulations for fuzzy automata over the product structure. Under this restriction, certain fundamental notions become computable in a finite number of steps.

Technically, we introduce and study the corresponding approximation of linear and complete residuated lattices in Section~\ref{sec:arl}. \LinhModified{The conditions of linearity and completeness are necessary not only to properly define such approximations, but also to establish the concepts introduced in the subsequent sections.} In Section~\ref{section: KJRKS}, we define and analyze {\em approximate invariances} on a fuzzy automaton, namely {\em right/left \mbox{$\varepsilon$-invariances}} and {\em right/left \mbox{$(\varepsilon,k)$-invariances}}. These generalize weakly right/left invariant fuzzy pre-orders (\FPOs)~\cite{SCI.14} and serve as fundamental tools for soft state reduction of \FfAs. Roughly speaking, the greatest right $\varepsilon$-invariance on a fuzzy automaton $\mA$ in the case $\varepsilon = 0$ (i.e., without approximation) corresponds to the greatest fuzzy weak auto-simulation of~$\mA$~\cite{BFFB}, which is greater than or equal to the greatest fuzzy auto-simulation of $\mA$~\cite{ijar/Nguyen.23}, which in turn is greater than or equal to the greatest fuzzy auto-bisimulation of $\mA$~\cite{ijar/Nguyen.23}. 
In Section~\ref{section: JHRJS}, we present how to use approximate invariances for soft state reduction of \FfAs. Given a fuzzy $\varepsilon$-pre-order ($\varepsilon$-\FPO) $Z$ on the set of states of $\mA$, where $\varepsilon$-\FPOs are a modified version of \FPOs, we define the {\em $(Z,\varepsilon)$-afterset fuzzy automaton} $\mA_{Z,\varepsilon}$ and prove that, if $Z$ is a right $\varepsilon$-invariance (respectively, right $(\varepsilon,k)$-invariance) on $\mA$, then $\mA_{Z,\varepsilon}$ is $\varepsilon$-equivalent (respectively, $(\varepsilon,k)$-equivalent) to $\mA$. The notion of a $(Z,\varepsilon)$-afterset fuzzy automaton in the case $\varepsilon = 0$ (i.e., without approximation) coincides with the notion of a $Z$-afterset fuzzy automaton introduced in~\cite{SCI.14}, which can be viewed as a form of ``quotient fuzzy automaton''. In Section~\ref{sec: alg}, we present our \SoftStateReduction algorithm, which iteratively applies soft state reduction based on right and left approximate invariances. The algorithm is guaranteed to terminate when either $k \in \NN$ is specified or the approximate residuated lattice (restricted by $\varepsilon$) is $\otimes$-locally finite. 

\LinhModified{The effectiveness of the proposed approach is demonstrated on several nontrivial examples. In particular, our algorithm achieves substantial reductions in cases where existing methods fail to produce any simplification. For instance,} applying the \SoftStateReduction algorithm to the \FfA shown in Figure~\ref{fig: HEJSA} with $\varepsilon = 0.1$ reduces the number of states from seven to five. As another example, consider the \FfA $\mA$ depicted in Figure~\ref{fig: KJRHW} (and detailed in Example~\ref{example: KJRHW}), which initially has 28 states. Applying our algorithm with the product structure and $\varepsilon = 0.01$ reduces $\mA$ to an \FfA with 19 states. Notably, for this \FfA, the approximation methods from~\cite{EPTCS386.6,FSS-D-25-00029} fail to produce a reduction for any $k \geq 14$, and the methods in~\cite{CSIP.10,SCI.14,SCB.18} do not yield a smaller \FfA (i.e., one with fewer than 28 states). Furthermore, the techniques from~\cite{MALIK1999323,BASAK2002223,CHENG2004439,LI20071423,Belohlvek2009OnAM,LITFS2015,Halamish2015,YangL.19,YANG202172,ShamsizadehZG24}, as well as those from~\cite{LEI20071413,PEEVA20084152,WU20101635,StanimirovicMC.22}, are inapplicable because they do not support \FfAs over the product structure. 
\LinhModified{These results illustrate that soft state reduction not only advances the theoretical framework but also leads to practically significant simplifications of fuzzy automata.}

The structure of this work is as follows: In addition to Sections~\ref{sec:arl}--\ref{sec: alg} discussed earlier, Section~\ref{sec:Preliminaries} introduces the basic definitions and notation, covering residuated lattices, fuzzy sets and relations, and fuzzy automata. Finally, Section~\ref{sec: conc} provides concluding remarks.

\section{Preliminaries} \label{sec:Preliminaries}

\subsection{Residuated lattices}

A {\em residuated lattice} \cite{B.02a,BV.05} is an algebra $\mL = \tuple{L, \leq, \land, \lor, \otimes, \to, 0, 1}$ such that:
\begin{itemize}
\item $\tuple{L, \leq, \land, \lor, 0, 1}$ is a lattice with the smallest element $0$ and the greatest element $1$,
\item $\tuple{L, \otimes, 1}$ is a commutative monoid with the unit $1$, 
\item the following (adjunction) property holds for all $x, y, z \in L$: 
\begin{equation}
x \otimes y \leq z \ \ \textrm{iff}\ \ x \leq (y \to z). 
\end{equation}
\end{itemize}
Such a residuated lattice is said to be {\em complete} (respectively, {\em linear}) if the lattice $\tuple{L, \leq, \land, \lor, 0, 1}$ is complete (respectively, linear). 
It is {\em $\otimes$-locally finite} if the monoid $\tuple{L, \otimes, 1}$ is locally finite (i.e., all of its finitely generated submonoids are finite). 
The operations $\otimes$ and $\to$ are called {\em multiplication} and {\em residuum}, respectively. 
Note that the operations $\land$, $\lor$ and $\otimes$ are associative, commutative and monotonically increasing with respect to both the arguments. On the other hand, the operation $\to$ is monotonically decreasing with respect to its first argument and monotonically increasing with respect to its second argument.
Given $A \subseteq L$, we write $\bigwedge\! A$ and $\bigvee\! A$ to denote the infimum and supremum of $A$, respectively. 
We also use the following notation:
\[ 
    \bigwedge_{i \in I} x_i \defeq \inf\{x_i \mid i \in I\},\qquad\qquad 
    \bigvee_{i \in I} x_i \defeq \sup \{x_i \mid i \in I\}. 
\]

The {\em product structure} and the {\em Hamacher structure} are the residuated lattices $\tuple{[0,1], \leq, \land, \lor, \otimes_{\!P}, \to_{\!P}, 0, 1}$ and $\tuple{[0,1], \leq, \land, \lor, \otimes_{\!H}, \to_{\!H}, 0, 1}$, respectively, which are the unit interval equipped with the usual order and the multiplication and residuum defined as follows:
\[ 
\begin{array}{ll}
x \otimes_{\!P} y \defeq xy, & 
(x \to_{\!P} y) \defeq \left\{\!\!\!
		\begin{array}{ll}
		1 & \!\textrm{if } x \leq y, \\ 
		y/x & \!\textrm{otherwise},
		\end{array}
        \right.
\\
\\
x \otimes_{\!H} y \defeq \left\{\!\!\!
		\begin{array}{ll}
		\frac{xy}{x+y-xy} & \!\textrm{if $x \neq 0$ or $y \neq 0$}, \\ 
		0 & \!\textrm{otherwise},
		\end{array}
        \right.
&
(x \to_{\!H} y) \defeq \left\{\!\!\!
		\begin{array}{ll}
		1 & \!\textrm{if } x \leq y, \\ 
		\frac{xy}{x-y+xy} & \!\textrm{otherwise}.
		\end{array}
        \right.
\end{array}
\]
These residuated lattices are linear and complete, but not $\otimes$-locally finite. 

From now on, if not stated otherwise, let $\mL = \tuple{L, \leq, \land, \lor, \otimes, \to, 0, 1}$ denote an arbitrary linear and complete residuated lattice. 

\subsection{Fuzzy sets and relations}

Let $X$ denote a non-empty set. A {\em fuzzy subset} of $X$ is a function $f: X \to L$ that assigns to each element of $X$ a degree of membership in the fuzzy subset. We refer to such a function as a {\em fuzzy set}. Given fuzzy sets $f, g: X \to L$, we write $f \leq g$ to denote that $f(x) \leq g(x)$ for all $x \in X$. 

A {\em fuzzy relation} on $X$ is a fuzzy subset of $X \times X$. The {\em inverse} of a fuzzy relation $r$ on $X$ is the fuzzy relation $\cnv{r}$ on $X$ defined by: $\cnv{r}(x,y) = r(y,x)$, for $x,y \in X$.
Given fuzzy subsets $f$ and $g$ of $X$ and fuzzy relations $r$ and $s$ on $X$, we define the {\em composition} operator $\circ$ between them as follows:
\[
\begin{array}{lrcl}
f \circ g \in L, & f \circ g & \!\!\defeq\!\! & \sup \{f(x) \otimes g(x) \mid x \in X\}, \\
f \circ r : X \to L, & (f \circ r)(x) & \!\!\defeq\!\!  & \sup \{f(y) \otimes r(y,x) \mid y \in X\}, \\
r \circ f : X \to L, & (r \circ f)(x) & \!\!\defeq\!\! & \sup \{r(x,y) \otimes f(y) \mid y \in X\}, \\
r \circ s : X \times X \to L\qquad & (r \circ s)(x,y) & \!\!\defeq\!\! & \sup \{r(x,z) \otimes s(z,y) \mid z \in X\}.
\end{array}
\]
Note that the operation $\circ$ is associative. 

The {\em identity relation} on $X$ is denoted by $id_X$ and defined by: $id_X(x,x) = 1$ and $id_X(x,y) = 0$ for $x \neq y$, with $x, y \in X$. 
A fuzzy relation $r$ on $X$ is {\em reflexive} if $id_X \leq r$, and {\em transitive} if $r \circ r \leq r$. A {\em fuzzy pre-order} (\FPO) on $X$ is a fuzzy relation on $X$ that is reflexive and transitive. %Note that, if $r$ is a fuzzy pre-order on $X$, then $r \circ r = r$. 

\subsection{Fuzzy automata}

A {\em fuzzy automaton} (over $\mL$) is a structure $\mA = \tuple{\QA, \Sigma, \IA, \TA, \FA}$, where $\QA$ is a nonempty set of states, $\Sigma$ is a finite nonempty set, called the {\em alphabet}, $\IA$ (respectively, $\FA$) is a fuzzy subset of $\QA$, called the {\em fuzzy set of initial} (respectively, {\em final}) {\em states}, and $\TA: \QA \times \Sigma \times \QA \to L$ is the {\em fuzzy transition function}. If $Q$ is finite, then $\mA$ is called a {\em fuzzy finite automaton} (\FfA).

We denote the empty word by $\epsilon$. Given a word $w \in \Sigma^*$, we define the fuzzy relation $\TA_w$ on $\QA$ inductively as follows:
\begin{eqnarray*}
\TA_\epsilon & \defeq & id_\QA; \\
\TA_\sigma(p, q) & \defeq & \TA(p,\sigma,q),\quad \textrm{for $\sigma \in \Sigma$ and $p,q \in \QA$}; \\
\TA_{\sigma u} & \defeq & \TA_\sigma \circ \TA_u,\quad \textrm{for $\sigma \in \Sigma$ and $u \in \Sigma^*$}.
\end{eqnarray*}
Note that $\TA_{uv} = \TA_u \circ \TA_v$, for any $u,v \in \Sigma^*$. 
Given $w \in \Sigma^*$, we also define the fuzzy subsets $\IA_w$ and $\FA_w$ of $\QA$ as follows:
\[ 
    \IA_w \defeq \IA \circ \TA_w,\qquad
    \FA_w \defeq \TA_w \circ \FA. 
\]
In words, $\TA_w(p,q)$ means the degree of reaching the state $q$ from the state $p$ by taking the sequence $w$ of actions, $\IA_w(q)$ means the degree of reaching the state $q$ from some initial state by taking the sequence $w$ of actions, and $\FA_w(q)$ means the degree of reaching some final state from $q$ taking the sequence $w$ of actions. 

We say that a state $q$ is {\em reachable} if there exists a word $w \in \Sigma^*$ such that $\IA_w(q) > 0$. Dually, $q$ is {\em productive} if there exists $w \in \Sigma^*$ such that $\FA_w(q) > 0$.

The \emph{fuzzy language recognized by $\mA$}, denoted by $\bL(\mA)$, is the fuzzy subset of $\Sigma^*$ defined as follows:
\[
    \bL(\mA)(w) = \IA \circ \TA_w \circ \FA = \IA_w \circ \FA = \IA \circ \FA_w,
\]
for $w \in \Sigma^*$. 
Given $k \in \NN$, the \emph{fuzzy language of words recognized by $\mA$ of length bounded by $k$}, denoted by $\bLdb(\mA)$, is the restriction of $\bL(\mA)$ to words of length bounded by $k$. That is, 
\begin{equation*}
    \bLdb(\mA)(w) = \begin{cases}
        \bL(\mA)(w) & \textrm{if}\ |w| \leq k, \\
        0 & \textrm{otherwise}.
    \end{cases}
\end{equation*}

The {\em reverse} of $\mA$ is the fuzzy automaton $\mA^{-1} = \tuple{\QA, \Sigma, \IA', \TA', \FA'}$ with $\IA' = \FA$, $\FA' = \IA$, and $\TA'_\sigma = \cnv{\TA}_\sigma$, for every $\sigma \in \Sigma$.

\section{Approximate residuated lattices} \label{sec:arl}

Recall that $\mL = \tuple{L, \leq, \land, \lor, \otimes, \to, 0, 1}$ denotes an arbitrary linear and complete residuated lattice. 
In this section, for a given $\varepsilon \in L$, we define $\lee$, $\lande$, $\lore$, $\otimese$ and $\toe$ so that $\mLe = \tuple{L, \lee, \lande, \lore, \otimese, \toe, 0, 1}$ is an approximation of~$\mL$, in the sense that all values of $L$ between 0 and $\varepsilon$ are treated as $\varepsilon$. The use of $\mLe$ is justifiable when $\varepsilon$ is a ``small'' value. A potential benefit of using $\mLe$ with $\varepsilon > 0$ is that it may be $\otimes$-locally finite when $\mL$ is not, as in the case of the product and Hamacher structures \LinhModified{(cf. Proposition~\ref{prop:productHamTNorm} for more details)}. We also provide related notions and their properties, which are needed for the next sections.  

In what follows, let $\varepsilon, x, y \in L$. 
We define that
\begin{itemize}
\item $x \lee y$ if $x \leq y$ or $x \leq \varepsilon$, 
\item $x \eqe y$ if $x \lee y$ and $y \lee x$.
\end{itemize}
Thus, $x \eqe y$ iff $x = y$ or ($x \leq \varepsilon$ and $y \leq \varepsilon$).
It is easily seen that $\lee$ is a pre-order and $\eqe$ is an equivalence relation.

We define
\[ x \lande y \defeq \left\{\!\!\!
		\begin{array}{ll}
		x \land y & \!\textrm{if } x \land y > \varepsilon, \\ 
		\varepsilon & \!\textrm{otherwise};
		\end{array}
        \right.
   \qquad\qquad
   x \lore y \defeq x \lor y \lor \varepsilon.
\]
Observe that these operations are associative, commutative and monotonically increasing with respect to both $\leq$ and $\lee$ and both the arguments. Furthermore, it is straightforward to check that 
\[
    x \lande y \eqe x \quad \textrm{iff}\quad x \lee y \quad \textrm{iff}\quad x \lore y \eqe y.
\]
Given $A \subseteq L$, we define 
\[ \bigwedgee\! A \defeq \left\{\!\!\!
		\begin{array}{ll}
		\bigwedge\! A & \!\textrm{if } \bigwedge\! A > \varepsilon, \\ 
		\varepsilon & \!\textrm{otherwise};
		\end{array}
        \right.
   \qquad\qquad
   \bigveee\! A \defeq \bigvee (A \cup \{\varepsilon\}).
\]

We also define
\begin{eqnarray}
x \otimese y & \defeq & 
        \left\{\!\!\!
		\begin{array}{ll}
		x \otimes y & \!\textrm{if } x \otimes y > \varepsilon, \\ 
		\varepsilon & \!\textrm{otherwise};
		\end{array}
        \right. \nonumber \\
(x \toe y) & \defeq & (x \lor \varepsilon \to y \lor \varepsilon). \label{eq: JFANM}
\end{eqnarray}
Observe that $\varepsilon \leq (x \toe y)$ and, if $x \lee y$, then $(x \toe y) = 1$.  
Also observe that the $\otimese$ operation is associative, commutative and monotonically increasing with respect to both $\leq$ and $\lee$ and both the arguments. 
On the other hand, the $\toe$ operation is monotonically decreasing with respect to its first argument and monotonically increasing with respect to its second argument, when using either $\leq$ or $\lee$ for comparison. 
The following lemma states that the $\otimese$ and $\toe$ operations satisfy the adjunction property as in residuated lattices.

\begin{lemma}\label{lemma: JRBAK}
For any $\varepsilon, x, y, z \in L$, 
\begin{equation}
x \lee (y \toe z) \quad\textrm{iff}\quad x \otimese y \lee z.
\end{equation}
\end{lemma}

\begin{proof}
Consider the left-to-right implication and suppose $x \lee (y \toe z)$. We have to prove that $x \otimese y \lee z$. The case $x \otimes y \leq \varepsilon$ is straightforward. So, suppose $x \otimes y > \varepsilon$. We have $x,y \geq x \otimes y$, hence $x,y > \varepsilon$. Since $x \lee (y \toe z)$, it follows that $x \leq (y \to z \lor \varepsilon)$. Consequently, $\varepsilon < x \otimes y \leq z \lor \varepsilon$ and therefore $x \otimese y \lee z$. 

Consider the right-to-left implication and suppose $x \otimese y \lee z$. We have to prove that $x \lee (y \toe z)$. The cases $x \leq \varepsilon$ or $y \leq z$ or $z < y \leq \varepsilon$ are trivial. So, suppose that $x,y > \varepsilon$ and $y > z$. Thus, we have to prove that $x \leq (y \to z \lor \varepsilon)$, or equivalently, $x \otimes y \leq z \lor \varepsilon$. The case $x \otimes y \leq \varepsilon$ is trivial. So, suppose $x \otimes y > \varepsilon$. Since $x \otimese y \lee z$, it follows that $\varepsilon < x \otimes y \lee z$. Hence, $x \otimes y \leq z \leq z \lor \varepsilon$, which completes the proof.
\end{proof}

{\markLinhModified
\begin{proposition}\label{prop:productHamTNorm}
Let $\varepsilon \in (0,1)$. If $\mL$ is the product or Hamacher structure, then $\mLe$ is $\otimes$-locally finite.\footnote{\markLinhModified This can be generalized for all structures from the Hamacher family, where $\otimes$ is a t-norm defined as follows, using any parameter $\lambda \geq 0$: 
\[
x \otimes y =
\begin{cases}
0, & \text{if $x = y = 0$ and $\lambda = 0$,} \\
\frac{x y}{\lambda + (1 - \lambda)(x + y - x y)}, & \text{otherwise.}
\end{cases}
\]
}
\end{proposition}

\begin{proof}
Suppose that $\mL$ is the product or Hamacher structure. 
Define $f(x) = x \otimes x$, for $x \in (0,1]$. Since $\otimes$ is monotonically increasing, to prove that $\mLe$ is $\otimes$-locally finite, it suffices to show that, for every $d \in (0,1)$, there exists $n_d$ such that $f^n(x) \leq \varepsilon$ for all $n \geq n_d$ ($n \in \NN$) and $x \in (0,d]$. In addition, it is sufficient to show this assertion only for the case where $\mL$ is the Hamacher structure. Consider this case and any $0 < x \leq d < 1$. We have
\[
    f(x) = \frac{x^2}{2x - x^2} = \frac{x}{2-x} \leq \frac{x}{2 - d} = cx,\ \ \textrm{where}\ c = \frac{1}{2 - d} \in (0.5, 1).
\]
By induction, $f^n(x) \leq c^n x \leq c^n d$. 
Taking $n_d = \left\lceil\log_c(\varepsilon/d)\right\rceil$ and any $n \geq n_d$, we have
\[
    f^n(x) \leq c^n d \leq c^{n_d} d \leq \varepsilon,
\]
which completes the proof.
\end{proof}
}

In what follows, let $f$ and $g$ be fuzzy subsets of $X$ and let $r$ and $s$ be fuzzy relations on $X$. 
We define the fuzzy subsets $f \lande g$, $f \lore g$, $f \circe r$ and $r \circe f$ of $X$, the value $f \circe g \in L$, and the fuzzy relations $r \circe s$, $f \setminuse g$ and $f \slashe g$ on $X$ as follows:
\begin{eqnarray*}
(f \lande g)(a) & \defeq & f(a) \lande g(a),\quad \textrm{for } a \in X; \\
(f \lore g)(a) & \defeq & f(a) \lore g(a),\quad \textrm{for } a \in X; \\
f \circe g & \defeq & \bigveee \{f(a) \otimese g(a) \mid a \in X \}; \\
(f \circe r)(a) & \defeq & \bigveee \{f(b) \otimese r(b,a) \mid b \in X \},\quad \textrm{for } a \in X; \\
(r \circe f)(a) & \defeq & \bigveee \{r(a,b) \otimese f(b) \mid b \in X \},\quad \textrm{for } a \in X; \\
(r \circe s)(a,b) & \defeq & \bigveee \{r(a,c) \otimese s(c,b) \mid c \in X \},\quad \textrm{for } a,b \in X; \\
(f \setminuse g)(a,b) & \defeq & (f(a) \toe g(b)),\quad \textrm{for } a,b \in X; \\
(f \slashe g)(a,b) & \defeq & (g(b) \toe f(a)),\quad \textrm{for } a,b \in X.
\end{eqnarray*}
Given a family $A$ of fuzzy subsets of $X$, we define the fuzzy subsets $\bigwedgee\!A$ and $\bigveee\!A$ of $X$ as follows:
\[
    (\bigwedgee\!A)(a) \defeq \bigwedgee \{f(a) \mid f \in A\}\quad \textrm{and}\quad
    (\bigveee\!A)(a) \defeq \bigveee \{f(a) \mid f \in A\},\quad \textrm{for } a \in X.
\]
We also define that
\begin{itemize}
\item $f \lee g$\ \ if\ \ $f(a) \lee g(a)$ for all $a \in X$,
\item $f \eqe g$\ \ if\ \ $f(a) \eqe g(a)$ for all $a \in X$. 
\end{itemize}

Clearly, $\lee$ is a pre-order and $\eqe$ is an equivalence relation on the family of fuzzy subsets of $X$. 
In addition, the operations $\lande$ and $\lore$ are associative and commutative, while $\circe$ is associative. All of these three operations (on the family of fuzzy sets/relations) are monotonically increasing with respect to both $\leq$ and $\lee$ and both the arguments. 

The {\em $\varepsilon$-truncation} of $f$ is the fuzzy subset $f_\varepsilon$ of $X$ defined as follows (cf.~\cite{MicicCMSN.24}):
\[ f_\varepsilon(a) \defeq \left\{\!\!\!
		\begin{array}{ll}
		f(a) & \!\textrm{if } f(a) > \varepsilon, \\ 
		\varepsilon & \!\textrm{otherwise}.
		\end{array}
        \right.
\]
Note that $f_\varepsilon = f$ for $\varepsilon = 0$. 
Furthermore, it can be checked that (cf.~\cite[Proposition~3.4]{MicicCMSN.24}):
\begin{itemize}
\item $f \leq f_\varepsilon$ and $(f_\varepsilon)_\varepsilon = f_\varepsilon$;
\item if $f \leq g$, then $f_\varepsilon \leq g_\varepsilon$;
\item if $\varepsilon \leq \varepsilon' \leq 1$, then $f_\varepsilon \leq f_{\varepsilon'}$ 
      and $(f_\varepsilon)_{\varepsilon'} = (f_{\varepsilon'})_\varepsilon = f_{\varepsilon'}$.
\end{itemize}

\begin{lemma}\label{lemma: HGDHO}
Let $f$ and $g$ be fuzzy subsets of $X$, and $r$, $r'$ and $s$ be fuzzy relations on $X$. Then, 
\begin{eqnarray}
& & f_{\!\varepsilon} \eqe f;\label{eq:HGDHO 1}\\  
& & 
f \circ\, g \eqe f \circe g,\quad
f \circ\, r \eqe f \circe r,\quad
r \,\circ f \eqe r \circe f,\quad
r \,\circ\, s \eqe r \circe s;\label{eq:HGDHO 2}\\
& & \textrm{if $f \eqe g$, then $f \circe r \eqe g \circe r$ and $r \circe f \eqe r \circe g$;}\label{eq:HGDHO 3}\\
& & \textrm{if $r \eqe s$, then $f \circe r \eqe f \circe s$ and $r \circe f \eqe s \circe f$;}\label{eq:HGDHO 4}\\
& & \textrm{if $r \eqe r'$, then $r \circe s \eqe r' \circe s$ and $s \circe r \eqe s \circe r'$}\label{eq:HGDHO 5}\\
& & f \slashe g \eqe f_{\!\varepsilon} \,\slashe g_\varepsilon;\label{eq:HGDHO 6}\\  
& & r \circ f \lee g \ \textrm{ iff }\ r \circe f \lee g \ \textrm{ iff }\ r \lee g \,\slashe f;\label{eq:HGDHO 7}\\ 
& & f \circ r \lee g \ \textrm{ iff }\ f \circe r \lee g \ \textrm{ iff }\ r \lee f \setminuse g.\label{eq:HGDHO 8}
\end{eqnarray}
\end{lemma}

\begin{proof}
The assertion~\eqref{eq:HGDHO 1} clearly holds. 

Consider the assertion~\eqref{eq:HGDHO 2}. We prove only the equation \mbox{$r \,\circ\, s \eqe r \circe s$} and leave the remaining three to the reader. Clearly, \mbox{$r \,\circ\, s \lee r \circe s$}. For the converse, consider arbitrary $a,b \in X$. We need to show that \mbox{$(r \circe s)(a,b) \lee (r \,\circ\, s)(a,b)$}. This holds because \mbox{$r(a,c) \otimese s(c,b) \lee r(a,c) \otimes s(a,b)$} for any $c \in X$. 

Concerning the assertions~\eqref{eq:HGDHO 3}-\eqref{eq:HGDHO 5}, we prove only that, if $f \eqe g$, then \mbox{$f \circe r \eqe g \circe r$}, and leave the remaining to the reader. Suppose $f \eqe g$ and consider an arbitrary $a \in X$. We need to prove that \mbox{$(f \circe r)(a) \eqe (g \circe r)(a)$}. It suffices to show that \mbox{$f(b) \otimese r(b,a) = g(b) \otimese r(b,a)$} for any $b \in X$. This is trivial for the case $f(b) = g(b)$. For the other case, since $f \eqe g$, we have $f(b) \leq \varepsilon$ and $g(b) \leq \varepsilon$, which implies \mbox{$f(b) \otimese r(b,a) = \varepsilon = g(b) \otimese r(b,a)$}.

Consider the assertion~\eqref{eq:HGDHO 6}. Let $a$ and $b$ be arbitrary elements of $X$. We need to prove that 
\[ (f \slashe g)(a,b) \eqe (f_{\!\varepsilon} \,\slashe g_\varepsilon)(a,b), \]
or equivalently, 
\[ (g(b) \toe f(a)) \eqe (g_\varepsilon(b) \toe f_{\!\varepsilon}(a)). \]
This latter holds due to~\eqref{eq: JFANM}.

Consider the assertion~\eqref{eq:HGDHO 7}. 
By~\eqref{eq:HGDHO 3}, \mbox{$r \circ f \lee g$} iff \mbox{$r \circe f \lee g$}. 
For the statement ``\mbox{$r \circe f \lee g$} implies \mbox{$r \lee g \,\slashe f$}'', suppose $r \circe f \lee g$ and consider arbitrary $a,b \in X$. We need to show that \mbox{$r(a,b) \lee (f(b) \toe g(a))$}. By definition, \mbox{$r(a,b) \otimese f(b) \leq (r \circe f)(a)$}. Since $r \circe f \lee g$, it follows that \mbox{$r(a,b) \otimese f(b) \lee g(a)$}. By Lemma~\ref{lemma: JRBAK}, we can further derive \mbox{$r(a,b) \lee (f(b) \toe g(a))$}. Now, consider the statement ``\mbox{$r \lee g \,\slashe f$} implies \mbox{$r \circe f \lee g$}'' and suppose that \mbox{$r \lee g \,\slashe f$}. Let $a$ be an arbitrary element of $X$. We need to prove that $(r \circe f)(a) \lee g(a)$. It suffices to show that \mbox{$r(a,b) \otimese f(b) \lee g(a)$} for any $b \in X$. Since \mbox{$r \lee g \,\slashe f$}, we have \mbox{$r(a,b) \lee (f(b) \toe g(a))$}, and by Lemma~\ref{lemma: JRBAK}, we can further derive \mbox{$r(a,b) \otimese f(b) \lee g(a)$}. 

The assertion~\eqref{eq:HGDHO 8} can be proved analogously as for~\eqref{eq:HGDHO 7}. 
\end{proof}

Given a fuzzy relation $r$ on $X$, we say that $r$ is {\em $\varepsilon$-transitive} if $r \circe r \lee r$. Note that, if $r$ is reflexive and $\varepsilon$-transitive, then \mbox{$r = r \circ id_{\!X} \leq r \circ r \leq r \circe r$} and therefore \mbox{$r \circe r \eqe r$}. We call $r$ a {\em fuzzy $\varepsilon$-pre-order} ($\varepsilon$-\FPO) on $X$ if it is reflexive, \mbox{$\varepsilon$-transitive} and equal to~$r_\varepsilon$. 

\begin{proposition}\label{prop: JHFLW}
Let $r$ be a fuzzy relation on $X$ and let $\varepsilon, \varepsilon' \in L$. 
\begin{enumerate}[(a)]
\item\label{item: JHFLW 1} If $r$ is reflexive and $\varepsilon$-transitive, then $r_\varepsilon$ is an $\varepsilon$-\FPO on~$X$.
\item\label{item: JHFLW 2} If $r$ is an \FPO on $X$, then $r_\varepsilon$ is an $\varepsilon$-\FPO on~$X$. 
\item\label{item: JHFLW 3} If $r$ is an $\varepsilon'\!$-\FPO on~$X$ and $\varepsilon' < \varepsilon$, then $r_\varepsilon$ is an $\varepsilon$-\FPO on~$X$.
\end{enumerate}
\end{proposition}

\begin{proof}
Consider the assertion~\eqref{item: JHFLW 1} and suppose that $r$ is reflexive and $\varepsilon$-transitive. 
Since $id_X \leq r \leq r_\varepsilon$, we have that $r_\varepsilon$ is reflexive. By the assertions of Lemma~\ref{lemma: HGDHO}, we have \mbox{$r_\varepsilon \circe r_\varepsilon \eqe r \circe r \lee r \leq r_\varepsilon$}, which implies that $r_\varepsilon$ is $\varepsilon$-transitive. Clearly, $(r_\varepsilon)_\varepsilon = r_\varepsilon$. Therefore, $r_\varepsilon$ is an $\varepsilon$-\FPO on~$X$.

Consider the assertion~\eqref{item: JHFLW 2} and suppose that $r$ is an \FPO on $X$. Similarly as above, we have 
$r_\varepsilon \circe r_\varepsilon \eqe r \circe r \eqe r \circ r \leq r \leq r_\varepsilon$,
which implies that $r_\varepsilon$ is $\varepsilon$-transitive. By the assertion~\eqref{item: JHFLW 1}, it follows that $r_\varepsilon$ is an $\varepsilon$-\FPO on~$X$.

Consider the assertion~\eqref{item: JHFLW 3} and suppose that $r$ is an $\varepsilon'\!$-\FPO on~$X$ with $\varepsilon' < \varepsilon$. 
Since \mbox{$r \circep r \eqep r$} and $\varepsilon' < \varepsilon$, we also have \mbox{$r \circe r \eqe r$}, which implies that $r_\varepsilon$ is $\varepsilon$-transitive. By the assertion~\eqref{item: JHFLW 1}, it follows that $r_\varepsilon$ is an $\varepsilon$-\FPO on~$X$.
\end{proof}

In the rest of this work, the words like ``greater'' and ``greatest'' are expressed with respect to the usual order~$\leq$. 

%-----------------------------------------------------------------------------------

\section{Approximate invariances}
\label{section: KJRKS}

In this section, we introduce and study approximate invariances on a fuzzy automaton. They are fundamental concepts needed for the soft state reduction of \FfAs presented in the next two sections. 

\begin{definition}
Let $\varepsilon \in L$, $k \in \NN$ and let $\mA = \tuple{\QA, \Sigma, \IA, \TA, \FA}$ be a fuzzy automaton and $Z$ a reflexive fuzzy relation on $\QA$. We say that $Z$ is 
\begin{enumerate}[(a)]
\item a {\em right $\varepsilon$-invariance} on $\mA$ if 
    \begin{equation}
    Z \circ \FA_w \eqe \FA_w \quad \textrm{for all } w \in \Sigma^*; \label{eq:inv1b}
    \end{equation}

\item a {\em right $(\varepsilon,k)$-invariance} on $\mA$ if 
    \begin{equation}
    Z \circ \FA_w \eqe \FA_w \quad \textrm{for all $w \in \Sigma^*$ with $|w| \leq k$}; \label{eq:inv2b}
    \end{equation}

\item a {\em left $\varepsilon$-invariance} on $\mA$ if 
    \begin{equation}
    \IA_w \circ Z \eqe \IA_w \quad \textrm{for all } w \in \Sigma^*; \label{eq:inv3b}
    \end{equation}

\item a {\em left $(\varepsilon,k)$-invariance} on $\mA$ if 
    \begin{equation}
    \IA_w \circ Z \eqe \FA_w \quad \textrm{for all $w \in \Sigma^*$ with $|w| \leq k$}. \label{eq:inv4b}
    \end{equation}
\end{enumerate}
\end{definition}

We refer to the concepts defined above as {\em approximate} ({\em right/left}) {\em invariances} on~$\mA$. These concepts involve two types of approximation. First, by using `$\eqe$' instead of `$=$', we treat all values between $0$ and $\varepsilon$ as equal. For instance, when $L = [0,1]$ and $\varepsilon$ is very small (e.g., $10^{-3}$ or $10^{-6}$), values smaller than $\varepsilon$ may result from noise or numerical errors and can be treated as equivalent. Second, rather than considering all words over $\Sigma$, we may focus only on words of length bounded by a constant~$k$. The advantage of employing approximate invariances on~$\mA$ lies in their potential to reduce the state space of $\mA$ more efficiently. This approach achieves lower time complexity and produces a fuzzy automaton that approximates $\mA$, with a size usually smaller than or equal to that in the case where such approximations are not used (see Corollary~\ref{cor: JHKWB 2}).

\begin{remark}\label{remark: JHFLW}
Clearly, every right (respectively, left) $\varepsilon$-invariance on $\mA$ is also a right (respectively, left) $(\varepsilon,k)$-invariance on~$\mA$. In addition, if $k > k'$, then every right (respectively, left) $(\varepsilon,k)$-invariance on~$\mA$ is also a right (respectively, left) $(\varepsilon,k')$-invariance on~$\mA$. In general, the converses do not hold. 
\myend
\end{remark}

\begin{proposition}\label{prop: JHEJA}
Let $\varepsilon \in L$, $k \in \NN$ and let $\mA$ be a fuzzy automaton.
\begin{enumerate}[(a)]
\item $Z$ is a left $\varepsilon$-invariance on $\mA$ iff $\cnv{Z}$ is a right $\varepsilon$-invariance on~$\cnv{\mA}$.
\item $Z$ is a left $(\varepsilon,k)$-invariance on $\mA$ iff $\cnv{Z}$ is a right $(\varepsilon,k)$-invariance on~$\cnv{\mA}$.
\end{enumerate}
\end{proposition}
This proposition holds due to the definition of $\cnv{\mA}$ and the fact that if $Z$ is reflexive, then so is $\cnv{Z}$. Having this proposition on the background, we will deal only with approximate right invariances. Our further results can easily be reformulated for approximate left invariances in the dual way.

\begin{remark}
The notion of right $\varepsilon$-invariances is a generalization of the notion of weakly right invariant \FPOs introduced in~\cite{SCI.14}.\footnote{Dually, left $\varepsilon$-invariances are a generalization of weakly left invariant \FPOs~\cite{SCI.14}. \FPOs are referred to in~\cite{SCI.14} as fuzzy quasi-orders.} 
Given a fuzzy automaton $\mA = \tuple{\QA, \Sigma, \IA, \TA, \FA}$, an \FPO $Z$ on $\QA$ is called an {\em weakly right invariant \FPO} on $\mA$ if it satisfies the following condition~\cite{SCI.14}:
    \begin{equation}
    Z \circ \FA_w = \FA_w \quad \textrm{for all } w \in \Sigma^*. \label{eq: JHFJS}
    \end{equation}
First, \eqref{eq: JHFJS} differs from~\eqref{eq:inv1b} in that `$=$' is used instead of `$\eqe$'. Thus, one can roughly say that weakly right invariant \FPOs are a specific type of right $\varepsilon$-invariances with $\varepsilon = 0$. Second, we only require right $\varepsilon$-invariances to be reflexive, while weakly right invariant \FPOs are in nature not only reflexive but also transitive. 

Moreover, it should be noted that the approximation by considering only words of a bounded length in right $(\varepsilon,k)$-invariances was not taken into account for weakly right invariant \FPOs in~\cite{SCI.14}. 
\myend
\end{remark}

\begin{example}\label{example: HEJSA}
Consider the case where $\mL$ is the product structure and $\mA = \tuple{\QA, \Sigma, \IA, \TA, \FA}$ is the fuzzy automaton over $\Sigma = \{\sigma\}$ that is illustrated in Figure~\ref{fig: HEJSA} (on page~\pageref{fig: HEJSA}) and specified below: 
{\footnotesize\[
\QA = \{0,1,2,3,4,5,6\},\quad
\IA =\begin{bmatrix} 1 & 0 & 0 & 1 & 0 & 0 & 0 \end{bmatrix},\quad
\TA_\sigma =\begin{bmatrix}
0 & 0.6 & 0 & 0 & 0 & 0 & 0 \\
0 & 0 & 0.8 & 0 & 0 & 0 & 0 \\
0.4 & 0 & 0 & 0 & 0 & 0 & 0 \\
0 & 0 & 0 & 0 & 1 & 0 & 0 \\
0 & 0.5 & 0 & 0 & 0 & 0.4 & 0 \\ 
0 & 0 & 0.7 & 0 & 0 & 0 & 0.8 \\ 
0 & 0 & 0 & 0 & 0.4 & 0 & 0 
\end{bmatrix},\quad
\FA =\begin{bmatrix} 0 \\ 0 \\ 0.5 \\ 0 \\ 0 \\ 0 \\ 0.5 \end{bmatrix}.
\]}

We have
{\footnotesize\[
\FA_\epsilon =\begin{bmatrix} 0 \\ 0 \\ 0.5 \\ 0 \\ 0 \\ 0 \\ 0.5 \end{bmatrix},\qquad
\FA_\sigma =\begin{bmatrix} 0 \\ 0.4 \\ 0 \\ 0 \\ 0 \\ 0.4 \\ 0 \end{bmatrix},\qquad
\FA_{\sigma^2} =\begin{bmatrix} 0.24 \\ 0 \\ 0 \\ 0 \\ 0.2 \\0 \\ 0 \end{bmatrix},\qquad
\FA_{\sigma^3} =\begin{bmatrix} 0 \\ 0 \\ 0.096 \\ 0.2 \\ 0 \\ 0 \\ 0.08 \end{bmatrix},\qquad
\FA_{\sigma^4} =\begin{bmatrix} 0 \\ 0.0768 \\ 0 \\ 0 \\ 0 \\ 0.0672 \\ 0 \end{bmatrix},\qquad \ldots
\]}

Let
{\footnotesize\[
Z_{0.1} =\begin{bmatrix}
1 & 0.25 & 0.2 & 0.5 & 1 & 0.25 & 0.2 \\
5/12 & 1 & 0.2 & 0.5 & 0.5 & 1 & 0.2 \\
5/12 & 0.25 & 1 & 0.5 & 0.5 & 0.25 & 1 \\
5/12 & 0.25 & 0.2 & 1 & 0.5 & 0.25 & 0.2 \\
5/6 & 0.25 & 0.2 & 0.5 & 1 & 0.25 & 0.2 \\
5/12 & 1 & 0.2 & 0.5 & 0.5 & 1 & 0.2 \\
5/12 & 0.25 & 1 & 0.5 & 0.5 & 0.25 & 1
\end{bmatrix},\qquad\qquad 
Z_{(0, 3)} =\begin{bmatrix}
1 & 0 & 0 & 0 & 1 & 0 & 0 \\
0 & 1 & 0 & 0 & 0 & 1 & 0 \\
0 & 0 & 1 & 0.48 & 0 & 0 & 1 \\
0 & 0 & 0 & 1 & 0 & 0 & 0 \\
5/6 & 0 & 0 & 0 & 1 & 0 & 0 \\
0 & 1 & 0 & 0 & 0 & 1 & 0 \\
0 & 0 & 5/6 & 0.4 & 0 & 0 & 1
\end{bmatrix}.
\]}

Notice that $\FA_{\sigma^4}(q) < 0.1$ for all $q \in \QA$. Hence, $\FA_{\sigma^i}(q) < 0.1$ for all $i \geq 4$ and $q \in \QA$. 
By using~\eqref{eq:inv1b} and~\eqref{eq:inv2b}, it can be verified that:\footnote{One can also use our implemented code, mentioned in Section~\ref{sec: alg}, for verification.\label{footnote: JDJSA}}
\begin{itemize}
\item $Z_{0.1}$ is a right $0.1$-invariance on $\mA$,
\item $Z_{(0, 3)}$ is a right $(0, 3)$-invariance on $\mA$.
\end{itemize}
In fact, they are the greatest right invariances of their respective types on $\mA$, but the verification is deferred until Example~\ref{example: JHDKS}.
\myend
\end{example}

\begin{lemma}\label{lemma: JRMAB}
Let $\varepsilon \in L$, $k \in \NN$ and let $\mA = \tuple{\QA, \Sigma, \IA, \TA, \FA}$ be a fuzzy automaton and $Z$ a reflexive fuzzy relation on~$\QA$.  
\begin{enumerate}[(a)]
\item\label{item:JRMAB 1a} If $Z$ is a right $\varepsilon$-invariance on $\mA$, then $Z \circ Z$ is also a right $\varepsilon$-invariance on~$\mA$.

\item\label{item:JRMAB 2a} If $Z$ is a right $(\varepsilon,k)$-invariance on $\mA$, then $Z \circ Z$ is also a right $(\varepsilon,k)$-invariance on~$\mA$.

\item\label{item:JRMAB 1} $Z$ is a right $\varepsilon$-invariance on $\mA$ iff 
    \begin{equation}
    Z \circ \FA_w \lee \FA_w \quad \textrm{for all } w \in \Sigma^*; \label{eq:inv1}
    \end{equation}

\item\label{item:JRMAB 2} $Z$ is a right $(\varepsilon,k)$-invariance on $\mA$ iff 
    \begin{equation}
    Z \circ \FA_w \lee \FA_w \quad \textrm{for all $w \in \Sigma^*$ with $|w| \leq k$}. \label{eq:inv2}
    \end{equation}
\end{enumerate}
\end{lemma}

\begin{proof}
Consider the assertion~\eqref{item:JRMAB 1a} and suppose that $Z$ is a right $\varepsilon$-invariance on~$\mA$. Clearly, \mbox{$id_Q \leq Z \circ Z$}. Let $w \in \Sigma^*$ be arbitrary. By the assertions of Lemma~\ref{lemma: HGDHO}, we have 
\[ 
    (Z \circ Z) \circ \FA_w = Z \circ (Z \circ \FA_w) 
    \eqe Z \circe (Z \circ \FA_w) 
    \eqe Z \circe \FA_w
    \eqe Z \circ \FA_w
    \eqe \FA_w.
\]
Therefore, $Z \circ Z$ is a right $\varepsilon$-invariance on~$\mA$. 
The assertion~\eqref{item:JRMAB 2a} can be proved analogously. 
The assertions~\eqref{item:JRMAB 1} and~\eqref{item:JRMAB 2} hold because $\FA_w = id_\QA \circ \FA_w \leq Z \circ \FA_w$. 
\end{proof}

\begin{remark}
Observing~\eqref{eq:inv1}, one can notice that the notion of right $\varepsilon$-invariances is related to the notion of $\varepsilon$-weak forward simulations introduced in~\cite{MicicCMSN.24}.\footnote{Dually, left $\varepsilon$-invariances are related to $\varepsilon$-weak backward simulations~\cite{MicicCMSN.24}.} These notions, however, differ from each other in several essential aspects:
\begin{itemize}
\item Right $\varepsilon$-invariances are defined on a fuzzy automaton $\mA$ over $\mL$, with $\mL$ being any linear and complete residuated lattice, while $\varepsilon$-weak forward simulations are defined between two fuzzy automata $\mA$ and $\mB$ over the (truncated) product structure. One can take $\mB = \mA$ (to deal with so-called $\varepsilon$-weak forward auto-simulations), but the product structure is only a particular residuated lattice. 
\item Restricting to the case where $\mL$ is the product structure, the conditions~(26) and~(27) in~\cite{MicicCMSN.24} essentially state that, given a fuzzy automaton $\mA = \tuple{\QA, \Sigma, \IA, \TA, \FA}$, a fuzzy relation $Z$ on $\QA$ is an $\varepsilon$-weak forward simulation on $\mA$ (i.e., between $\mA$ and itself) iff it satisfies the two following conditions:
    \begin{align}
    & \IA \leq \IA \circe \cnv{Z}, \label{eq: JHFNA 1}\\
    & \cnv{Z} \circe \FA_w \lee \FA_w,\quad \textrm{for all } w \in \Sigma^*. \label{eq: JHFNA 2}
    \end{align}
Clearly, not only that~\eqref{eq:inv1} differs from~\eqref{eq: JHFNA 2} in that $Z$ is used instead of $\cnv{Z}$, but the definition of right $\varepsilon$-invariances does not require~\eqref{eq: JHFNA 1}. On the other hand, right $\varepsilon$-invariances are required to be reflexive, while $\varepsilon$-weak forward simulations are not. 
\end{itemize}

Moreover, it should be noted that the approximation by considering only words with a bounded length in right $(\varepsilon,k)$-invariances was not taken into account for $\varepsilon$-weak forward simulations in~\cite{MicicCMSN.24}. 
\myend
\end{remark}

\begin{proposition}\label{prop:JHRWS}
Let $\varepsilon,\varepsilon' \in L$, $k \in \NN$ and let $\mA = \tuple{\QA, \Sigma, \IA, \TA, \FA}$ be a fuzzy automaton. If $\varepsilon < \varepsilon'$, then 
\begin{enumerate}[(a)]
\item\label{item:JHRWS 1} every right $\varepsilon$-invariance on $\mA$ is also a right $\varepsilon'$-invariance on~$\mA$; 
\item\label{item:JHRWS 2} every right $(\varepsilon,k)$-invariance on $\mA$ is also a right $(\varepsilon',k)$-invariance on~$\mA$. 
\end{enumerate}
\end{proposition}

\begin{proof}
Suppose $\varepsilon < \varepsilon'$. 
Consider the assertion~\eqref{item:JHRWS 1} and let $Z$ be a right $\varepsilon$-invariance on~$\mA$. By Lemma~\ref{lemma: JRMAB}, we only need to prove that \mbox{$(Z \circ \FA_w)(q) \leep \FA_w(q)$} for any $w \in \Sigma^*$ and $q \in \FA$. It suffices to show that, for any $p \in \QA$, \mbox{$Z(q,p) \otimes \FA_w(p) \leep \FA_w(q)$}. This is trivial for the case \mbox{$Z(q,p) \otimes \FA_w(p) \leq \FA_w(q)$}. So, suppose the opposite. Since $Z$ satisfies~\eqref{eq:inv1}, it follows that \mbox{$Z(q,p) \otimes \FA_w(p) \leq \varepsilon < \varepsilon'$}. Therefore, \mbox{$Z(q,p) \otimes \FA_w(p) \leep \FA_w(q)$}, which completes the proof of~\eqref{item:JHRWS 1}. 
The assertion~\eqref{item:JHRWS 2} can be proved analogously. 
\end{proof}

\begin{lemma}\label{lemma: KRBZM}
Let $\varepsilon \in L$, $k \in \NN$ and let $\mA = \tuple{\QA, \Sigma, \IA, \TA, \FA}$ be a fuzzy automaton and $Z$ a reflexive fuzzy relation on~$\QA$.  
\begin{enumerate}[(a)]
\item\label{item: KRBZM 1} $Z$ is a right $\varepsilon$-invariance on $\mA$ iff 
    \begin{equation}
    Z \lee \FA_w\ \slashe \FA_w \quad \textrm{for all } w \in \Sigma^*; \label{eq:inv1c}
    \end{equation}

\item\label{item: KRBZM 2} $Z$ is a right $(\varepsilon,k)$-invariance on $\mA$ iff 
    \begin{equation}
    Z \lee \FA_w\ \slashe \FA_w \quad \textrm{for all $w \in \Sigma^*$ with $|w| \leq k$}. \label{eq:inv2c}
    \end{equation}
\end{enumerate}
\end{lemma}

This lemma immediately follows from Lemma~\ref{lemma: JRMAB} and the assertion~\eqref{eq:HGDHO 7} of Lemma~\ref{lemma: HGDHO}.

In general, given $\varepsilon \in L$, $k \in \NN$ and a fuzzy automaton $\mA = \tuple{\QA, \Sigma, \IA, \TA, \FA}$, we are interested in the greatest right $\varepsilon$-invariance and the greatest right $(\varepsilon,k)$-invariance on $\mA$. The following theorem and its corollaries are devoted to such fuzzy relations.

\begin{theorem}\label{theorem: JRKKA}
Let $\varepsilon \in L$, $k \in \NN$ and let $\mA = \tuple{\QA, \Sigma, \IA, \TA, \FA}$ be a fuzzy automaton. 
\begin{enumerate}[(a)]
\item\label{item: JRKKA 1} The greatest right $\varepsilon$-invariance on $\mA$ exists and is equal to 
    \begin{equation}
    \bigwedgee \{\FA_w\ \slashe \FA_w \mid w \in \Sigma^*\}. \label{eq: JRKKA 1}
    \end{equation}

\item\label{item: JRKKA 2} The greatest right $(\varepsilon,k)$-invariance on $\mA$ exists and is equal to 
    \begin{equation}
    \bigwedgee \{\FA_w\ \slashe \FA_w \mid w \in \Sigma^* \textrm{ with } |w| \leq k\}. \label{eq: JRKKA 2}
    \end{equation}
\end{enumerate}
\end{theorem}

\begin{proof}
Consider the assertion~\eqref{item: JRKKA 1}. Let $Z = \bigwedgee \{\FA_w\ \slashe \FA_w \mid w \in \Sigma^*\}$.
By~\eqref{eq:HGDHO 7}, \mbox{$id_\QA \lee \FA_w\ \slashe \FA_w$} for all $w \in \Sigma^*$, and therefore, $Z$ is reflexive. We first prove that $Z$ satisfies~\eqref{eq:inv1} and is therefore a right $\varepsilon$-invariance on $\mA$. Let $w \in \Sigma^*$ be arbitrary. By the assertion~\eqref{eq:HGDHO 7} of Lemma~\ref{lemma: HGDHO}, it suffices to prove that \mbox{$Z \lee \FA_w\ \slashe \FA_w$}. This latter holds because
\[
    Z = \bigwedgee \{\FA_u\ \slashe \FA_u \mid u \in \Sigma^*\} \lee \FA_w\ \slashe \FA_w.
\]
Now, let $Z'$ be an arbitrary right $\varepsilon$-invariance on $\mA$ and let $p$ and $q$ be arbitrary states of~$\mA$. To prove $Z' \leq Z$, we show that $Z'(p,q) \leq Z(p,q)$. By Lemma~\ref{lemma: KRBZM}, we have
\begin{equation}
Z'(p,q) \lee (\FA_w\ \slashe \FA_w)(p,q) \quad \textrm{for all } w \in \Sigma^*.\label{eq: JHFHW}
\end{equation}
Observe that $\varepsilon \leq Z(p,q)$. If $Z'(p,q) \leq \varepsilon$, then clearly $Z'(p,q) \leq Z(p,q)$. So, suppose $Z'(p,q) > \varepsilon$. By~\eqref{eq: JHFHW}, it follows that 
\[ Z'(p,q) \leq (\FA_w\ \slashe \FA_w)(p,q) \quad \textrm{for all } w \in \Sigma^*. \]
Therefore, 
\[ 
    Z'(p,q) \leq \bigwedge_{w \in \Sigma^*}(\FA_w\ \slashe \FA_w)(p,q) \leq \bigwedgee \{(\FA_w\ \slashe \FA_w)(p,q) \mid w \in \Sigma^*\} = Z(p,q).
\]

The assertion~\eqref{item: JRKKA 2} can be proved analogously.
\end{proof}

\begin{corollary}\label{cor: JRKKA}
Let $\varepsilon \in L$, $k \in \NN$ and let $\mA = \tuple{\QA, \Sigma, \IA, \TA, \FA}$ be a fuzzy automaton. 
\begin{enumerate}[(a)]
\item\label{item: JRKKA 3} The greatest right $\varepsilon$-invariance on $\mA$ is equal to 
    \begin{equation}
    \bigwedgee \{(\FA_w)_\varepsilon\, \slashe (\FA_w)_\varepsilon \mid w \in \Sigma^*\}. \label{eq: JRKKA 3}
    \end{equation}

\item\label{item: JRKKA 4} The greatest right $(\varepsilon,k)$-invariance on $\mA$ is equal to 
    \begin{equation}
    \bigwedgee \{(\FA_w)_\varepsilon\, \slashe (\FA_w)_\varepsilon \mid w \in \Sigma^* \textrm{ with } |w| \leq k\}. \label{eq: JRKKA 4}
    \end{equation}
\end{enumerate}
\end{corollary}

This corollary immediately follows from Theorem~\ref{theorem: JRKKA} and the assertion~\eqref{eq:HGDHO 6} of Lemma~\ref{lemma: HGDHO}. 

\begin{example}\label{example: JHDKS}
Let us continue Example~\ref{example: HEJSA}. Let $\mL$, $\mA = \tuple{\QA, \Sigma, \IA, \TA, \FA}$, $Z_{0.1}$ and $Z_{(0, 3)}$ be as in that example. 
We exploit the fuzzy sets $\FA_{\sigma^i}$ computed in that example, for $0 \leq i \leq 4$, to show that
\begin{itemize}
\item $Z_{0.1}$ is the greatest right $0.1$-invariance on $\mA$,
\item $Z_{(0, 3)}$ is the greatest right $(0, 3)$-invariance on $\mA$.
\end{itemize}
Recall that $\FA_{\sigma^i}(q) < 0.1$ and $(\FA_{\sigma^i})_{0.1}(q) = 0.1$ for all $i \geq 4$ and $q \in \QA$. By Corollary~\ref{cor: JRKKA}, it follows that:
\begin{itemize}
\item the greatest right $0.1$-invariance on $\mA$ is equal to 
    \begin{equation*}
    \textstyle\bigwedge_{0.1} \{(\FA_{\sigma^i})_{0.1}\, \slash_{\!0.1} (\FA_{\sigma^i})_{0.1} \mid 0 \leq i \leq 4\}, \textrm{ which is equal to } Z_{0.1};
    \end{equation*}
\item the greatest right $(0, 3)$-invariance on $\mA$ is equal to 
    \begin{equation*}
    \bigwedge_{0 \leq i \leq 3} \FA_{\sigma^i}\, \slash \FA_{\sigma^i}, \textrm{ which is equal to } Z_{(0,3)}.
    \end{equation*}
\end{itemize}
Similarly, 
\begin{itemize}
\item since $\FA_{\sigma^i}(q) \leq 0.2$ and $(\FA_{\sigma^i})_{0.2}(q) = 0.2$ for all $i \geq 3$ and $q \in \QA$, the greatest right $0.2$-invariance on $\mA$ is 
\[ Z_{0.2} = \textstyle\bigwedge_{0.2} \{(\FA_{\sigma^i})_{0.2}\, \slash_{\!0.2} (\FA_{\sigma^i})_{0.2} \mid 0 \leq i \leq 3\}, \]
\item the greatest right $(0, 4)$-invariance on $\mA$ is
\[ 
    Z_{(0,4)} = \bigwedge_{0 \leq i \leq 4} \big(\FA_{\sigma^i}\, \slash \FA_{\sigma^i}\big).
\]
\end{itemize}
It can be verified that\footref{footnote: JDJSA} 
{\footnotesize\[
Z_{0.2} =\begin{bmatrix}
1 & 0.5 & 0.4 & 1 & 1 & 0.5 & 0.4 \\
5/6 & 1 & 0.4 & 1 & 1 & 1 & 0.4 \\
5/6 & 0.5 & 1 & 1 & 1 & 0.5 & 1 \\
5/6 & 0.5 & 0.4 & 1 & 1 & 0.5 & 0.4 \\
5/6 & 0.5 & 0.4 & 1 & 1 & 0.5 & 0.4 \\
5/6 & 1 & 0.4 & 1 & 1 & 1 & 0.4 \\
5/6 & 0.5 & 1 & 1 & 1 & 0.5 & 1
\end{bmatrix},\qquad\qquad
Z_{(0,4)} =\begin{bmatrix}
1 & 0 & 0 & 0 & 1 & 0 & 0 \\
0 & 1 & 0 & 0 & 0 & 1 & 0 \\
0 & 0 & 1 & 0.48 & 0 & 0 & 1 \\
0 & 0 & 0 & 1 & 0 & 0 & 0 \\
5/6 & 0 & 0 & 0 & 1 & 0 & 0 \\
0 & 0.875 & 0 & 0 & 0 & 1 & 0 \\
0 & 0 & 5/6 & 0.4 & 0 & 0 & 1
\end{bmatrix}.
\]}
Notice that $Z_{(0,4)}$ differs from $Z_{(0,3)}$ only in that $Z_{(0,4)}(6,1) = 0.875$ (instead of 1).
\myend
\end{example}

\begin{corollary}\label{cor: KGBXH}
Let $\varepsilon \in L$, $k \in \NN$ and let $\mA = \tuple{\QA, \Sigma, \IA, \TA, \FA}$ be a fuzzy automaton. 
Then, both the greatest right $\varepsilon$-invariance and the greatest right $(\varepsilon,k)$-invariance on $\mA$ are $\varepsilon$-\FPOs.
\end{corollary}

\begin{proof}
The existence of the mentioned approximate invariances follows from Theorem~\ref{theorem: JRKKA}. 
Let $Z$ be either the greatest right $\varepsilon$-invariance or the greatest right $(\varepsilon,k)$-invariance on~$\mA$. 
By Theorem~\ref{theorem: JRKKA}, we have $Z = Z_\varepsilon$. 
By the assertions~\eqref{item:JRMAB 1a} and~\eqref{item:JRMAB 2a} of Lemma~\ref{lemma: JRMAB}, we have \mbox{$Z \circ Z \leq Z$}. 
By~\eqref{eq:HGDHO 2}, it follows that \mbox{$Z \circe Z \eqe Z \circ Z \leq Z$}, and hence \mbox{$Z \circe Z \lee Z$}, which means that $Z$ is $\varepsilon$-transitive. Therefore, $Z$ is an $\varepsilon$-\FPO.
\end{proof}

\begin{corollary}\label{cor: JHFRK}
Let $\varepsilon, \varepsilon' \in L$, $k,k' \in \NN$ and let $\mA = \tuple{\QA, \Sigma, \IA, \TA, \FA}$ be a fuzzy automaton. 
\begin{enumerate}[(a)]
\item\label{item: JHFRK 1} The greatest right $(\varepsilon,k)$-invariance on $\mA$ is greater than or equal to the greatest right $\varepsilon$-invariance on~$\mA$.

\item\label{item: JHFRK 2} If $k > k'$, then the greatest right $(\varepsilon,k')$-invariance on~$\mA$ is greater than or equal to the greatest right $(\varepsilon,k)$-invariance on~$\mA$.

\item\label{item: JHFRK 3} If $\varepsilon > \varepsilon'$, then the greatest right $\varepsilon$-invariance on~$\mA$ is greater than or equal to the greatest right $\varepsilon'$-invariance on~$\mA$.

\item\label{item: JHFRK 4} If $\varepsilon > \varepsilon'$, then the greatest right $(\varepsilon,k)$-invariance on~$\mA$ is greater than or equal to the greatest right $(\varepsilon',k)$-invariance on~$\mA$.
\end{enumerate}
\end{corollary}

\begin{proof}
The existence of the mentioned approximate invariances follows from Theorem~\ref{theorem: JRKKA}. The corollary itself follows immediately from Remark~\ref{remark: JHFLW} and Proposition~\ref{prop:JHRWS}.
\end{proof}

The following theorem is a preparation for using approximate invariances in the soft state reduction algorithm presented in Section~\ref{sec: alg}. 

\begin{theorem}\label{theorem: HGDJL}
Suppose that $\mLe$ is $\otimes$-locally finite and let $\mA = \tuple{\QA, \Sigma, \IA, \TA, \FA}$ be an \FfA. 
There exists the smallest number $k \in \NN$ such that 
    \begin{equation}
    \{(\FA_w)_\varepsilon \mid w \in \Sigma^* \textrm{ with } |w| \leq k\} = \{(\FA_w)_\varepsilon \mid w \in \Sigma^* \textrm{ with } |w| \leq k+1\}. \label{eq: HGDJL 1}
    \end{equation}
With that $k$, the greatest right $(\varepsilon,k)$-invariance on $\mA$ is also the greatest right $\varepsilon$-invariance on~$\mA$. 
Furthermore, they are equal to
    \begin{equation}
    \bigwedgee \{(\FA_w)_\varepsilon\ \slashe (\FA_w)_\varepsilon \mid w \in \Sigma^* \textrm{ with } |w| \leq k\}. \label{eq: HGDJL 2}
    \end{equation}
\end{theorem}

\begin{proof}
Since $\mLe$ is $\otimes$-locally finite and $\mA = \tuple{\QA, \Sigma, \IA, \TA, \FA}$ is finite, the set $\{(\FA_w)_\varepsilon \mid w \in \Sigma^*\}$ is finite. Therefore, there exists the smallest number $k \in \NN$ such that \eqref{eq: HGDJL 1} holds. Fix that $k$ and denote
\begin{equation}\label{eq: JHKWL}
\mF = \{(\FA_w)_\varepsilon \mid w \in \Sigma^* \textrm{ with } |w| \leq k\}.
\end{equation}

Given $u \in \Sigma^*$ with $|u| \geq k+1$, we first prove by induction on $|u|$ that \mbox{$(\FA_u)_\varepsilon \in \mF$}.
The base case occurs when $|u| = k+1$ and holds due to~\eqref{eq: HGDJL 1}. For the induction step, assume that \mbox{$(\FA_u)_\varepsilon \in \mF$} for all $u \in \Sigma^*$ with $|u| \leq h$, for some fixed $h \geq k+1$. Consider any $v = \sigma u \in \Sigma^*$ with $|u| = h$ and $\sigma \in \Sigma$. We need to prove that \mbox{$(\FA_v)_\varepsilon \in \mF$}.
By the inductive assumption, there exists $w' \in \Sigma^*$ such that $|w'| \leq k$ and $(\FA_u)_\varepsilon = (\FA_{w'})_\varepsilon$. 
Also by the inductive assumption, \mbox{$(\FA_{\sigma w'})_\varepsilon \in \mF$}.
By using the assertions of Lemma~\ref{lemma: HGDHO}, we have 
\begin{align*}
(\FA_v)_\varepsilon & = (\FA_{\sigma u})_\varepsilon = (\TA_\sigma \circ \FA_u)_\varepsilon 
    \eqe \TA_\sigma \circ \FA_u 
    \eqe \TA_\sigma \circe \FA_u 
    \eqe \TA_\sigma \circe (\FA_u)_\varepsilon \\
& 
    = \TA_\sigma \circe (\FA_{w'})_\varepsilon 
    \eqe \TA_\sigma \circe \FA_{w'} 
    \eqe \TA_\sigma \circ \FA_{w'} 
    = \FA_{\sigma w'} 
    \eqe (\FA_{\sigma w'})_\varepsilon.
\end{align*}
Since \mbox{$(\FA_{\sigma w'})_\varepsilon \in \mF$}, it follows that \mbox{$(\FA_v)_\varepsilon \in \mF$}, which completes the induction proof, and as a consequence, we obtain
\begin{equation}\label{eq: JHKWL 2}
\mF = \{(\FA_w)_\varepsilon \mid w \in \Sigma^*\}.
\end{equation}

By~\eqref{eq: JHKWL}, \eqref{eq: JHKWL 2} and Corollary~\ref{cor: JRKKA}, we conclude that the greatest right $(\varepsilon,k)$-invariance on $\mA$ is equal to the greatest right $\varepsilon$-invariance on $\mA$, and they are equal to the expression specified by~\eqref{eq: HGDJL 2}.
\end{proof}

\section{Soft state reduction of fuzzy automata}
\label{section: JHRJS}

It was shown in~\cite{SCI.14} how weakly right invariant \FPOs can be used for state reduction of fuzzy automata. In this section, we generalize that method of~\cite{SCI.14} for using $\varepsilon$-\FPOs which are right $\varepsilon$-invariances or right $(\varepsilon,k)$-invariances for soft state reduction of fuzzy automata (i.e., to produce a smaller fuzzy automaton that approximates the original one when possible).

Given a fuzzy relation $Z$ on $Q$, for $q \in Q$, the {\em $Z$-afterset} $Z_q$ and the {\em $Z$-foreset} $Z^q$ of $q$ are the fuzzy subsets of $Q$ defined as follows: $Z_q(p) = Z(q, p)$ and $Z^q(p) = Z(p,q)$, for $p \in Q$. The set of all $Z$-aftersets of states from $Q$ is denoted as $Q/Z$. 
It was proved in~\cite[Theorem~3.1]{SCI.14} that, if $Z$ is an \FPO on $Q$, then for any $p,q \in Q$, $Z_p = Z_q$ iff $Z^p = Z^q$. 
We generalize this property for $\varepsilon$-\FPOs as follows. 

\begin{lemma}\label{lemma: HGKWK}
If $Z$ is an $\varepsilon$-\FPO on $Q$, then for any $p,q \in Q$, $Z_p = Z_q$ iff $Z^p = Z^q$. 
\end{lemma}

\begin{proof}
Let $Z$ be an $\varepsilon$-\FPO on $Q$. For the left-to-right implication, suppose $Z_p = Z_q$. 
Thus, $Z(p,q) = Z(q,q) = 1$. Let $s \in Q$ be arbitrary. By using~\eqref{eq:HGDHO 2}, we have 
\[
    Z(s,p) = Z(s,p) \otimes Z(p,q) \leq (Z \circ Z)(s,q) \eqe (Z \circe Z)(s,q) \lee Z(s,q).
\]
That is, $Z(s,p) \lee Z(s,q)$. Since $Z = Z_\varepsilon$, it follows that $Z(s,p) \leq Z(s,q)$. Similarly, $Z(s,q) \leq Z(s,p)$, and consequently, $Z^p = Z^q$. 
The right-to-left implication can be proved analogously.
\end{proof}

In what follows, let $\mA = \tuple{\QA, \Sigma, \IA, \TA, \FA}$ be a fuzzy automaton, $\varepsilon$ a value from $L$ and $Z$ an $\varepsilon$-\FPO on $\QA$. 
The {\em $(Z,\varepsilon)$-afterset fuzzy automaton} of $\mA$, denoted by $\mA_{Z,\varepsilon}$, is defined to be the fuzzy automaton $\tuple{\QAp, \Sigma, \IAp, \TAp, \FAp}$ such that:
\[
    \QAp = \QA/Z,\quad
    \IAp(Z_q) = \IA \circe Z^q,\quad
    \TAp(Z_q,\sigma,Z_p) = Z_q \circe \TA_\sigma \circe Z^p,\quad
    \FAp(Z_p) = Z_p \circe \FA,
\]
for $p, q \in \QA$ and $\sigma \in \Sigma$.
By Lemma~\ref{lemma: HGKWK}, the fuzzy automaton $\mA_{Z,\varepsilon}$ is well-defined. 

Given an \FfA $\mA = \tuple{\QA, \Sigma, \IA, \TA, \FA}$ with $n$ states and $\QA = \{0,\ldots,n-1\}$, we can identify any $Z_q \in \QA/Z$ with the smallest number $p \in \QA$ such that $Z_p = Z_q$. We use this convention in the following example. 

\begin{example}\label{example: HEJSA 2}
Let us continue Examples~\ref{example: HEJSA} and~\ref{example: JHDKS}. Let $\mL$, $\mA = \tuple{\QA, \Sigma, \IA, \TA, \FA}$, $Z_{0.1}$, $Z_{0.2}$, $Z_{(0, 3)}$ and $Z_{(0, 4)}$ be as in those examples. 

Observe that $(Z_{0.1})_p = (Z_{0.1})_q$ for $(p,q) \in \{(1,5), (2,6)\} \subset \QA \times \QA$. Thus, the fuzzy automaton $\mA_{Z_{0.1},0.1} = \tuple{\QAp, \Sigma, \IAp, \TAp, \FAp}$ has $\QAp = \QA/Z_{0.1} = \{0,1,2,3,4\}$ and it can be verified that\footref{footnote: JDJSA}  
{\footnotesize\[
\IAp =\begin{bmatrix} 1 & 0.25 & 0.2 & 1 & 1 \end{bmatrix},\qquad
\TAp_\sigma =\begin{bmatrix}
5/12 & 0.6 & 0.2 & 0.3 & 0.5 \\
5/12 & 0.25 & 0.8 & 0.4 & 0.5 \\
5/12 & 0.25 & 0.2 & 0.25 & 0.5 \\
5/6 & 0.25 & 0.2 & 0.5 & 1 \\ 
5/12 & 0.5 & 0.2 & 0.25 & 0.5
\end{bmatrix},\qquad
\FAp =\begin{bmatrix} 0.1 \\ 0.1 \\ 0.5 \\ 0.1 \\ 0.1 \end{bmatrix}.
\]}

Also observe that:
\begin{itemize}
\item The set of states of  $\mA_{Z_{0.2},0.2}$ is $\QA/Z_{0.2} = \{0,1,2,3\}$.
\item The set of states of  $\mA_{Z_{(0,3)},0}$ is $\QA/Z_{(0,3)} = \{0,1,2,3,4,6\}$.
\item The set of states of  $\mA_{Z_{(0,4)},0}$ is $\QA/Z_{(0,4)} = \{0,1,2,3,4,5,6\}$. That is, $|\QA/Z_{(0,4)}| = |\QA|$.
\myend
\end{itemize}
\end{example}

\begin{proposition}\label{prop: HRKWG}
Let $\mA = \tuple{\QA, \Sigma, \IA, \TA, \FA}$ be a fuzzy automaton, $\varepsilon$ a value from $L$ and $Z, Z'$ $\varepsilon$-\FPOs on $\QA$.
If $Z' \leq Z$, then $\dbloverline{\QA/Z'} \geq \dbloverline{\QA/Z}$, which means that $\mA_{Z,\varepsilon}$ has at most as many states as $\mA_{Z',\varepsilon}$.\footnote{We write $\dbloverline{X}$ to denote the cardinality of a set $X$. When $X$ is finite, we also write $|X|$ to denote the number of elements of~$X$.}
\end{proposition}

\begin{proof}
Suppose $Z' \leq Z$. It suffices to prove that, for any $p,q \in \QA$, if $Z'_p = Z'_q$, then $Z_p = Z_q$. Suppose $Z'_p = Z'_q$. Thus, $Z'(q,p) = Z'(p,p) = 1$. Since $Z' \leq Z$, it follows that $Z(q,p) = 1$. Let $s$ be an arbitrary state from~$Q$. Since $Z \circe Z \eqe Z$ and $Z_\varepsilon = Z$, we have 
\[ Z(p,s) = Z(q,p) \otimes Z(p,s) = Z(q,p) \otimese Z(p,s) \lee (Z \circe Z)(q,s) \eqe Z(q,s). \]
That is, $Z(p,s) \lee Z(q,s)$, and consequently, $Z(p,s) \leq Z(q,s)$ (since $Z_\varepsilon = Z$). Similarly, we can also derive $Z(q,s) \leq Z(p,s)$. 
Hence, $Z(q,s) = Z(p,s)$, for any $s \in \QA$. Therefore, $Z_p = Z_q$, which completes the proof.   
\end{proof}

\begin{corollary}\label{cor: JHFRK2}
Let $\varepsilon, \varepsilon' \in L$, $k,k' \in \NN$ and let $\mA = \tuple{\QA, \Sigma, \IA, \TA, \FA}$ be a fuzzy automaton. 
\begin{enumerate}[(a)]
\item\label{item: JHFRK2 1} If $Z$ and $Z'$ are the greatest right $(\varepsilon,k)$-invariance and the greatest right $\varepsilon$-invariance on $\mA$, respectively, then $\mA_{Z,\varepsilon}$ has at most as many states as $\mA_{Z',\varepsilon}$.

\item\label{item: JHFRK2 2} If $k > k'$, and $Z$ and $Z'$ are the greatest right $(\varepsilon,k)$-invariance and the greatest right $(\varepsilon,k')$-invariance on $\mA$, respectively, then $\mA_{Z',\varepsilon}$ has at most as many states as $\mA_{Z,\varepsilon}$.

\item\label{item: JHFRK2 3} If $\varepsilon > \varepsilon'$, and $Z$ and $Z'$ are the greatest right $\varepsilon$-invariance and the greatest right $\varepsilon'$-invariance on $\mA$, respectively, then $\mA_{Z,\varepsilon}$ has at most as many states as $\mA_{Z',\varepsilon}$.

\item\label{item: JHFRK2 4} If $\varepsilon > \varepsilon'$, and $Z$ and $Z'$ are the greatest right $(\varepsilon,k)$-invariance and the greatest right $(\varepsilon',k)$-invariance on $\mA$, respectively, then $\mA_{Z,\varepsilon}$ has at most as many states as $\mA_{Z',\varepsilon}$.
\end{enumerate}
\end{corollary}

This corollary immediately follows from Corollaries~\ref{cor: KGBXH} and~\ref{cor: JHFRK}, and Proposition~\ref{prop: HRKWG}.

\begin{lemma}\label{lemma:KJKZR}
Let $\mA = \tuple{\QA, \Sigma, \IA, \TA, \FA}$ be a fuzzy automaton, $\varepsilon$ a value from $L$ and $Z$ an $\varepsilon$-\FPO on $\QA$. 
For any $w = \sigma_1\ldots\sigma_n \in \Sigma^*$ with $n \geq 0$ and $\sigma_1,\ldots,\sigma_n \in \Sigma$, we have
\begin{equation}\label{eq:fuzzyLangAR}
	\bL({\mA_{Z,\varepsilon}})(w) \eqe \IA \circe Z \circe \TA_{\sigma_1} \circe Z \circe \cdots \circe \TA_{\sigma_n} \circe Z \circe \FA.
\end{equation}
When $w = \epsilon$ (the empty word), the above equation has the form:\footnote{We use different symbols, $\varepsilon$ and $\epsilon$, for different purposes. The symbol $\varepsilon$ denotes a value from $L$, while $\epsilon$ denotes the empty word.}
\begin{equation}\label{eq:fuzzyLangARe}
	\bL({\mA_{Z,\varepsilon}})(\epsilon) \eqe \IA \circe Z \circe \FA.
\end{equation}
\end{lemma}

\begin{proof}
Denote $\mA_{Z,\varepsilon} = \tuple{\QAp, \Sigma, \IAp, \TAp, \FAp}$. We have 
\[
\begin{array}{l}
\bL({\mA_{Z,\varepsilon}})(w) = \IAp \circ \TAp_{\sigma_1} \circ \cdots \circ \TAp_{\sigma_n} \circ \FAp = \\[1ex]
= \bigvee_{q_0,\ldots,q_n \in Q} (\IA \circe Z^{q_0}) \otimes (Z_{q_0} \circe \TA_{\sigma_1} \circe Z^{q_1}) \otimes \cdots \otimes (Z_{q_{n-1}} \circe \TA_{\sigma_n} \circe Z^{q_n}) \otimes (Z_{q_n} \circe \FA) \\[1ex]
\eqe \bigvee_{q_0,\ldots,q_n \in Q} (\IA \circe Z^{q_0}) \otimese (Z_{q_0} \circe \TA_{\sigma_1} \circe Z^{q_1}) \otimese \cdots \otimese (Z_{q_{n-1}} \circe \TA_{\sigma_n} \circe Z^{q_n}) \otimese (Z_{q_n} \circe \FA)
\end{array}
\]
By \eqref{eq:HGDHO 2}--\eqref{eq:HGDHO 5} and due to the associativity of $\circe$, we can continue the above transformation to derive 
\begin{eqnarray*}
\bL({\mA_{Z,\varepsilon}})(w) & \eqe & \IA \circe (Z \circe Z) \circe \TA_{\sigma_1} \circe (Z \circe Z) \circe \cdots \circe \TA_{\sigma_n} \circe (Z \circe Z) \circe \FA.
\end{eqnarray*}
Since $Z$ is an $\varepsilon$-\FPO, we have $Z \circe Z \eqe Z$. Therefore, \eqref{eq:fuzzyLangAR} holds.
\end{proof}

The following theorem states that if an $\varepsilon$-\FPO $Z$ is a right $\varepsilon$-invariance or $(\varepsilon,k)$-invariance on a fuzzy automaton~$\mA$, then the fuzzy automaton $\mA_{Z,\varepsilon}$ approximates~$\mA$. 

\begin{theorem}\label{theorem:JHKWB}
Let $\varepsilon \in L$, $k \in \NN$ and let $\mA = \tuple{\QA, \Sigma, \IA, \TA, \FA}$ be a fuzzy automaton and $Z$ an $\varepsilon$-\FPO on $\QA$. 
\begin{enumerate}[(a)]
\item\label{item:JHKWB 1} If $Z$ is a right $\varepsilon$-invariance on $\mA$, then $\mA_{Z,\varepsilon}$ is $\varepsilon$-equivalent to $\mA$, which means
    \begin{equation}
    \bL(\mA_{Z,\varepsilon}) \eqe \bL(\mA). 
    \end{equation}
\item\label{item:JHKWB 2} If $Z$ is a right $(\varepsilon,k)$-invariance on $\mA$, then $\mA_{Z,\varepsilon}$ is $(\varepsilon,k)$-equivalent to $\mA$, which means
    \begin{equation}
    \bLdb(\mA_{Z,\varepsilon}) \eqe \bLdb(\mA). 
    \end{equation}
\end{enumerate}
\end{theorem}

\begin{proof}
Consider the assertion~\eqref{item:JHKWB 1} and suppose that $Z$ is a right $\varepsilon$-invariance on $\mA$. Consider an arbitrary word $w = \sigma_1\ldots\sigma_n \in \Sigma^*$ with $n \geq 0$ and $\sigma_1,\ldots,\sigma_n \in \Sigma$. By Lemma~\ref{lemma:KJKZR} and~\eqref{eq:inv1b}, as well as \eqref{eq:HGDHO 2}--\eqref{eq:HGDHO 5} and due to the associativity of $\circe$, we have:

\begin{eqnarray*}
\bL(\mA_{Z,\varepsilon})(w) 
& \eqe & \IA \circe Z \circe \TA_{\sigma_1} \circe Z \circe \cdots \circe \TA_{\sigma_{n-2}} \circe Z \circe \TA_{\sigma_{n-1}} \circe Z \circe \TA_{\sigma_n} \circe Z \circe \FA \\
& \eqe & \IA \circe Z \circe \TA_{\sigma_1} \circe Z \circe \cdots \circe \TA_{\sigma_{n-2}} \circe Z \circe \TA_{\sigma_{n-1}} \circe Z \circe \TA_{\sigma_n} \circe \FA \\
& \eqe & \IA \circe Z \circe \TA_{\sigma_1} \circe Z \circe \cdots \circe \TA_{\sigma_{n-2}} \circe Z \circe \TA_{\sigma_{n-1}} \circe Z \circe (\TA_{\sigma_n} \circ \FA) \\
& \eqe & \IA \circe Z \circe \TA_{\sigma_1} \circe Z \circe \cdots \circe \TA_{\sigma_{n-2}} \circe Z \circe \TA_{\sigma_{n-1}} \circe Z \circe \FA_{\sigma_n} \\
& \eqe & \IA \circe Z \circe \TA_{\sigma_1} \circe Z \circe \cdots \circe \TA_{\sigma_{n-2}} \circe Z \circe \TA_{\sigma_{n-1}} \circe \FA_{\sigma_n} \\
& \eqe & \IA \circe Z \circe \TA_{\sigma_1} \circe Z \circe \cdots \circe \TA_{\sigma_{n-2}} \circe Z \circe (\TA_{\sigma_{n-1}} \circ \FA_{\sigma_n}) \\
& \eqe & \IA \circe Z \circe \TA_{\sigma_1} \circe Z \circe \cdots \circe \TA_{\sigma_{n-2}} \circe Z \circe \FA_{\sigma_{n-1}\sigma_n} \\
& \eqe & \ldots \\
& \eqe & \IA \circe Z \circe \FA_{\sigma_1\ldots\sigma_{n-2}\sigma_{n-1}\sigma_n} \\
& \eqe & \IA \circe \FA_w \\
& \eqe & \IA \circ \FA_w \\
& \eqe & \bL(\mA)(w).
\end{eqnarray*}

The assertion~\eqref{item:JHKWB 2} can be proved analogously. In particular, we modify the proof by using~\eqref{eq:inv2b} instead of~\eqref{eq:inv1b}, requiring that $|w| \leq k$, which means $n \leq k$, and adding at the end the following: 
\[ \bLdb(\mA_{Z,\varepsilon})(w) = \bL(\mA_{Z,\varepsilon})(w) \eqe \bL(\mA)(w) = \bLdb(\mA)(w). \]

\vspace*{-1em}
\end{proof}

The following corollary states that if $Z$ is the greatest right $\varepsilon$-invariance or the greatest right $\varepsilon$-invariance on a fuzzy automaton~$\mA$, then the fuzzy automaton $\mA_{Z,\varepsilon}$ approximates~$\mA$, and moreover, $\mA_{Z,\varepsilon}$ has a size less than or equal to that obtained when $Z$ is not the greatest one.

\begin{corollary}\label{cor: JHKWB 2}
Let $\varepsilon \in L$, $k \in \NN$ and let $\mA = \tuple{\QA, \Sigma, \IA, \TA, \FA}$ be a fuzzy automaton. 
\begin{enumerate}[(a)]
\item\label{item:JHKWB 3} If $Z$ is the greatest right $\varepsilon$-invariance on $\mA$, then $\mA_{Z,\varepsilon}$ is $\varepsilon$-equivalent to $\mA$ (i.e., \( \bL(\mA_{Z,\varepsilon}) \eqe \bL(\mA) \)) and $\mA_{Z,\varepsilon}$ has at most as many states as $\mA_{Z',\varepsilon}$, for any $\varepsilon$-\FPO $Z'$ that is a right $\varepsilon$-invariance on $\mA$. 

\item\label{item:JHKWB 4} If $Z$ is the greatest right $(\varepsilon,k)$-invariance on $\mA$, then $\mA_{Z,\varepsilon}$ is $(\varepsilon,k)$-equivalent to $\mA$ (i.e., \( \bLdb(\mA_{Z,\varepsilon}) \eqe \bLdb(\mA) \)) and $\mA_{Z,\varepsilon}$ has at most as many states as $\mA_{Z',\varepsilon}$, for any $\varepsilon$-\FPO $Z'$ that is a right $(\varepsilon,k)$-invariance on $\mA$. 
\end{enumerate}
\end{corollary}

This corollary immediately follows from Corollary~\ref{cor: KGBXH}, Theorem~\ref{theorem:JHKWB} and Proposition~\ref{prop: HRKWG}. 

\begin{example}\label{example: HEJSA 3}
Let us continue Examples~\ref{example: HEJSA}, \ref{example: JHDKS} and~\ref{example: HEJSA 2}. \LinhModified{This example demonstrates how the construction of the $(Z,\varepsilon)$-afterset fuzzy automaton reflects the trade-off between approximation and state reduction for different values of $\varepsilon$.} Let $\mL$, $\mA = \tuple{\QA, \Sigma, \IA, \TA, \FA}$, $Z_{0.1}$, $Z_{0.2}$, $Z_{(0, 3)}$ and $Z_{(0, 4)}$ be as in those examples. Recall that $|\QA| = 7$. By Corollary~\ref{cor: JHKWB 2}, we have that:
\begin{itemize}
\item $\mA_{Z_{0.1},0.1}$ has five states and $\bL(\mA_{Z_{0.1},0.1}) =_{0.1} \bL(\mA)$, 
\item $\mA_{Z_{0.2},0.2}$ has four states and $\bL(\mA_{Z_{0.2},0.2}) =_{0.2} \bL(\mA)$, 
\item $\mA_{Z_{(0,3)},0}$ has six states and $\bL^{\leq 3}(\mA_{Z_{(0,3)},0}) = \bL^{\leq 3}(\mA)$, 
\item $\mA_{Z_{(0,4)},0}$ has seven states like $\mA$.
\end{itemize}
\LinhModified{These cases clearly illustrate the effect of state reduction under approximation: as $\varepsilon$ increases, the automaton undergoes stronger reductions (fewer states), while the language is preserved only up to the tolerance level~$\varepsilon$. In addition, bounding the word length using $k=3$ allows a reduction with exact agreement on all words of length at most three. Moreover, by} Corollary~\ref{cor: JHFRK} and Proposition~\ref{prop: HRKWG}, as a consequence of the last item above, we have: 
\begin{itemize}
\item If $k > 4$ and $Z_{(0,k)}$ is the greatest right $(0,k)$-invariance on $\mA$, then $\mA_{Z_{(0,k)},0}$ has the same number of states as $\mA$.
\item If $Z_0$ is the greatest right $0$-invariance on $\mA$, then $\mA_{Z_0,0}$ has the same number of states as $\mA$.
That is, if $Z_\varepsilon$ is the greatest right $\varepsilon$-invariance on $\mA$, then $\mA_{Z_\varepsilon,\varepsilon}$ may have fewer states than $\mA$ only when $\varepsilon > 0$. In fact, $\mA_{Z_\varepsilon,\varepsilon}$ has fewer states than $\mA$ when $\varepsilon > 0.0768$, and the same number of states as $\mA$ when $\varepsilon < 0.0768$.\footref{footnote: JDJSA} 
\myend
\end{itemize}
\end{example}

%-----------------------------------------------------------------------------------

\section{A soft state reduction algorithm}
\label{sec: alg}

\begin{figure*}
\begin{function}[H]
\caption{ReductionByRightInvariance($\mA, \varepsilon, k$)\label{funcRBRI}}
\Input{an \FfA $\mA = \tuple{\QA, \Sigma, \IA, \TA, \FA}$, $\varepsilon \in L$, $k \in \NN \cup \{\infty\}$.}
\Output{an \FfA isomorphic to $\mA_{Z,\varepsilon}$, where $Z$ is the greatest right $\varepsilon$-invariance on $\mA$ if $k = \infty$, or the greatest right $(\varepsilon,k)$-invariance on $\mA$ otherwise.}

\BlankLine
initialize $\mF$ and $\dF$ to empty sets\label{step: funcRBRI 1}\;
insert $\FA_\varepsilon$ to both $\mF$ and $\dF$\label{step: funcRBRI 2}\tcp*{$d$ in $\dF$ stands for `delta' or `unprocessed'}
\RepTimes{$k$\label{step: funcRBRI 3}}{
    set $\dF_2$ to an empty set\;
    \ForEach{$f \in \dF$ and $\sigma \in \Sigma$}{
        $g := \TA_\sigma \circe f$\label{step: funcRBRI 6}\;
        \lIf{$g \notin \mF$}{insert $g$ to both $\mF$ and $\dF_2$\label{step: funcRBRI 7}}
    }
    $\dF := \dF_2$\;
    \lIf{$\dF$ is empty}{\label{step: funcRBRI 9}\Break}
}
$Z := \bigwedgee \{f\, \slashe f \mid f \in \mF\}$\label{step: funcRBRI 10}\;

\BlankLine
let $d = |\QA/Z|$ and choose states $q_1,\ldots,q_d$ from $\QA$ such that $\QA/Z = \{Z_{q_1},\ldots,Z_{q_d}\}$\label{step: funcRBRI 11}\;
$\QAp := \{q_1,\ldots,q_d\}$\;
declare $\IAp$, $\FAp$ and $\TAp$ as fuzzy subsets of $\QAp$, $\QAp$ and $\QAp \times \Sigma \times \QAp$, respectively\;
\lForEach{$q \in \QAp$}{$\IAp(q) := \IA \circe Z^q$ and $\FAp(q) := Z_q \circe \FA$}
\lForEach{$\sigma \in \Sigma$ and $p,q \in \QAp$}{$\TAp_\sigma(p,q) := Z_p \circe \TA_\sigma \circe Z^q$\label{step: funcRBRI 15}}

\BlankLine
\Return the \FfA $\tuple{\QAp, \Sigma, \IAp, \TAp, \FAp}$\label{step: funcRBRI 16}\;
\end{function}

\begin{function}[H]
\caption{SoftStateReduction$_0$($\mA, \varepsilon, k$)\label{funSSRZ}}
%\Input{$\mA = \tuple{\QA, \Sigma, \IA, \TA, \FA}$ is an \FfA, $\varepsilon \in L$, $k \in \NN \cup \{\infty\}$.}
\Input{as for the above function.} 
\Output{an \FfA $\mAp$ such that $\bL(\mAp) \eqe \bL(\mA)$ if $k = \infty$, or $\bLdb(\mAp) \eqe \bLdb(\mA)$ otherwise, with the property that either $\mAp$ has fewer states than $\mA$ or $\mAp = \mA$ (as a failure).}

\BlankLine
$\mAp := \mA$\;
\While{$\True$}{
    $\mA_2 := \ReductionByRightInvariance(\mAp, \varepsilon, k)$\label{step: funSSRZ 3}\;
    $\mA_3 := \cnv{\big(\ReductionByRightInvariance(\cnv{(\mA_2)}, \varepsilon, k)\big)}$\label{step: funSSRZ 4}\;
    \uIf{$\mA_3$ has fewer states than $\mAp$}{
        $\mAp := \mA_3$\;
    }
    \Else{
        \Return $\mAp$\;
    }
}
\end{function}

\begin{algorithm}[H]
\caption{SoftStateReduction\label{algSSR}}
%\Input{$\mA = \tuple{\QA, \Sigma, \IA, \TA, \FA}$ is an \FfA, $\varepsilon \in L$, $k \in \NN \cup \{\infty\}$.}
%\Output{an \FfA $\mAp$ such that $\bL(\mAp) \eqe \bL(\mA)$ if $k = \infty$, or $\bLdb(\mAp) \eqe \bLdb(\mA)$ otherwise, with the property that either $\mAp$ has fewer states than $\mA$ or $\mAp = \mA$ (as a failure).}
\Input{as for the two above functions.} 
\Output{as for the above function \SoftStateReductionZ.}

\BlankLine
eliminate from $\mA$ all unreachable or unproductive states\label{step: algSSR 0}\;
$\mA_2 := \SoftStateReductionZ(\mA, \varepsilon, k)$\label{step: algSSR 1}\;
$\mA_3 := \cnv{\big(\SoftStateReductionZ(\cnv{(\mA)}, \varepsilon, k)\big)}$\label{step: algSSR 2}\;
\Return $\mA_2$ if $\mA_2$ has fewer states than $\mA_3$, else $\mA_3$\;
\end{algorithm}
\end{figure*}

In this section, we present an algorithm, named \SoftStateReduction, for softly reducing states of a given \FfA~$\mA$. In particular, given also $\varepsilon \in L$ and $k \in \NN \cup \{\infty\}$, it terminates when $\mLe$ is $\otimes$-locally finite or $k \neq \infty$ and returns an \FfA $\mAp$ such that $\bL(\mAp) \eqe \bL(\mA)$ if $k = \infty$, or $\bLdb(\mAp) \eqe \bLdb(\mA)$ otherwise, with the property that either $\mAp$ has fewer states than $\mA$ or $\mAp = \mA$ (as a failure).\footnote{The algorithm may terminate also when $\mLe$ is not $\otimes$-locally finite and $k = \infty$, e.g., when $\mA$ uses only crisp values.} It is based on Theorem~\ref{theorem: HGDJL} and Corollary~\ref{cor: JHKWB 2}. 
We also describe our implementation of this algorithm and present illustrative examples.

\LinhModified{We first provide an intuitive background for the \SoftStateReduction algorithm. It relies on the auxiliary functions \ReductionByRightInvariance and \SoftStateReductionZ that have the same input as the algorithm. The first function firstly computes the greatest right $\varepsilon$-invariance (or the greatest right $(\varepsilon,k)$-invariance) on $\mA$ (denoted by $Z$). Afterwards, it computes an \FfA isomorphic to $\mA_{Z,\varepsilon}$ (which is the $(Z,\varepsilon)$-afterset fuzzy automaton of $\mA$). 
According to Theorem~\ref{theorem:JHKWB}, $\mA_{Z,\varepsilon}$ is $\varepsilon$-equivalent (or $(\varepsilon,k)$-equivalent) to $\mA$.
The second function (\SoftStateReductionZ) first applies the right reduction (using \ReductionByRightInvariance), and afterwards the left reduction in a dual way (reverse the \FfA~$\mA$, apply right reduction to the reversed \FfA, then reverse back). Next, the resulting \FfA is compared with the starting one: if the number of states has decreased, repeat the process, if not, stop. 
Intuitively, \SoftStateReductionZ iteratively applies reductions based on approximate invariances until no further reduction is possible. This ensures that the result is ``as reduced as possible'' under the given approximation settings. 
The \SoftStateReduction algorithm first cleans up the input \FfA (removes unreachable/unproductive states), then runs \SoftStateReductionZ twice: directly on  $\mA$, then on the reversed \FfA (and then reverses back). Finally, it returns whichever result is smaller. 
More details and discussion are given below.} 

The \SoftStateReduction algorithm uses the function $\ReductionByRightInvariance(\mA, \varepsilon, k)$ (presented on page~\pageref{funcRBRI}) with the same parameters. This function returns an \FfA isomorphic to $\mA_{Z,\varepsilon}$, where $Z$ is the greatest right $\varepsilon$-invariance on $\mA$ if $k = \infty$, or the greatest right $(\varepsilon,k)$-invariance on $\mA$ otherwise. 

The statements \ref{step: funcRBRI 1}-\ref{step: funcRBRI 9} of the \ReductionByRightInvariance function computes the set 
\[ \mF = \{(\FA_w)_\varepsilon \mid w \in \Sigma^* \textrm{ with } |w| \leq k\}. \]
In particular, the statements \ref{step: funcRBRI 1} and~\ref{step: funcRBRI 2} initializes $\mF$ to 
\mbox{$\{F_\varepsilon\} = \{(\FA_w)_\varepsilon \mid w \in \Sigma^* \textrm{ with } |w| \leq 0\}$}, 
and as an invariant of the ``repeat'' loop, the $i$-th iteration of this loop, for $1 \leq i \leq k$, changes $\mF$
from \mbox{$\{(\FA_w)_\varepsilon \mid w \in \Sigma^* \textrm{ with } |w| \leq i-1\}$}
to \mbox{$\{(\FA_w)_\varepsilon \mid w \in \Sigma^* \textrm{ with } |w| \leq i\}$}. 
Roughly speaking, at the beginning of each iteration, $\dF$ is the difference between the contents of $\mF$ at the end of the two previous iterations. 
By Theorem~\ref{theorem: HGDJL}, the ``repeat'' loop always terminates when $\mLe$ is $\otimes$-locally finite or $k \in \NN$. 
If $k \in \NN$ and the ``repeat'' loop terminates not by the ``break'' statement, then by the mentioned invariant and Corollary~\ref{cor: JRKKA}, the fuzzy relation $Z$ computed by the statement~\ref{step: funcRBRI 10} is in fact the greatest right $(\varepsilon,k)$-invariance on $\mA$. 
If the ``repeat'' loop terminates by the ``break'' statement at the $i$-iteration, for some $1 \leq i \leq k$, then by the mentioned invariant and Theorem~\ref{theorem: HGDJL}, the fuzzy relation $Z$ computed by the statement~\ref{step: funcRBRI 10} is equal to both the greatest right $(\varepsilon,i-1)$-invariance on $\mA$ and the greatest right $\varepsilon$-invariance on~$\mA$, and by Corollary~\ref{cor: JHFRK}, it follows that this fuzzy relation $Z$ is also equal to the greatest right $(\varepsilon,k)$-invariance on $\mA$.  

It is easily seen that the statements \ref{step: funcRBRI 11}-\ref{step: funcRBRI 15} of the \ReductionByRightInvariance function computes an \FfA \mbox{$\mAp = \tuple{\QAp, \Sigma, \IAp, \TAp, \FAp}$} that is isomorphic to $\mA_{Z,\varepsilon}$. The isomorphism is the bijection from $\QA/Z$ to $Q'$ that maps each $Z_{q_i}$ to~$q_i$, for $1 \leq i \leq d$. 

By the above analysis of the \ReductionByRightInvariance function, we obtain the following result.

\begin{lemma}\label{lemma: HEKAN}
The function $\ReductionByRightInvariance(\mA, \varepsilon, k)$ terminates at least when $\mLe$ is $\otimes$-locally finite or $k \in \NN$ and always returns a correct result (when terminates).
\end{lemma}

The \SoftStateReduction algorithm utilizes the \ReductionByRightInvariance function through the auxiliary function $\SoftStateReductionZ(\mA, \varepsilon, k)$, which takes the same parameters. This auxiliary function operates as follows: it iteratively applies \ReductionByRightInvariance and the dual function until no further state reduction is possible. The dual reduction is performed by first reversing the current fuzzy automaton, then applying \ReductionByRightInvariance to reduce it, and finally reversing the resulting fuzzy automaton. By Proposition~\ref{prop: JHEJA}, it can be said that the dual reduction uses left invariances (instead of right invariances). In the case where the number of states is not reduced at all, the \SoftStateReductionZ function returns the input fuzzy automaton. 

Given a fuzzy language $L : \Sigma^* \to L$ over an alphabet $\Sigma$, the {\em reverse} of $L$ is the fuzzy language \mbox{$\cnv{L} : \Sigma^* \to L$} such that, for any $k \in \NN$ and any $\sigma_1,\ldots,\sigma_k \in \Sigma$, $\cnv{L}(\sigma_1\sigma_2\ldots\sigma_k) = L(\sigma_k\sigma_{k-1}\ldots\sigma_1)$. 

\begin{lemma}\label{lemma: IUDHB}
The function $\SoftStateReductionZ(\mA, \varepsilon, k)$ terminates at least when $\mLe$ is $\otimes$-locally finite or $k \in \NN$ and always returns a correct result (when terminates).
\end{lemma}

\begin{proof}
We prove that the following is an invariant of the ``while'' loop of the \SoftStateReductionZ function: either $\mAp = \mA$ or $\mAp$ has fewer states than $\mA$ and, if $k = \infty$, then $\bL(\mAp) \eqe \bL(\mA)$, else $\bLdb(\mAp) \eqe \bLdb(\mA)$. For this, it suffices to prove that, after executing the statements~\ref{step: funSSRZ 3} and~\ref{step: funSSRZ 4}, we have that: 
\begin{equation}\label{eq: JHJAO}
\textrm{if $k = \infty$, then $\bL(\mA_3) \eqe \bL(\mAp)$, else $\bLdb(\mA_3) \eqe \bLdb(\mAp)$.}
\end{equation}
Consider the statements~\ref{step: funSSRZ 3} and~\ref{step: funSSRZ 4} (of the \SoftStateReductionZ function) and denote 
\[ \mA'_2 = \ReductionByRightInvariance(\cnv{(\mA_2)}, \varepsilon, k). \] 
Thus, $\mA_3 = \cnv{(\mA'_2)}$. There are the two following cases.
\begin{itemize}
\item Case $k = \infty$: By Lemma~\ref{lemma: HEKAN} and Corollary~\ref{cor: JHKWB 2}, we have $\bL(\mA_2) \eqe \bL(\mAp)$ and $\bL(\mA'_2) \eqe \bL\big(\cnv{(\mA_2)}\big)$. Therefore,
\[
    \bL(\mA_3) = \bL(\cnv{(\mA'_2)}) = \cnv{(\bL(\mA'_2))} 
    \eqe \cnv{\big(\bL\big(\cnv{(\mA_2)}\big)\big)}
    = \cnv{\big(\cnv{\big(\bL(\mA_2)\big)}\big)}
    = \bL(\mA_2) \eqe \bL(\mAp).
\]
That is, $\bL(\mA_3) \eqe \bL(\mAp)$. 
\item Case $k \neq \infty$ (i.e., $k \in \NN$): By Lemma~\ref{lemma: HEKAN} and Corollary~\ref{cor: JHKWB 2}, we have $\bLdb(\mA_2) \eqe \bLdb(\mAp)$ and $\bLdb(\mA'_2) \eqe \bLdb\big(\cnv{(\mA_2)}\big)$. Analogously as for the above case, we can derive $\bLdb(\mA_3) \eqe \bLdb(\mAp)$. 
\end{itemize}
We have proved~\eqref{eq: JHJAO} and therefore also the mentioned invariant of the ``while'' loop of the \SoftStateReductionZ function.

Since the function $\ReductionByRightInvariance(\mA, \varepsilon, k)$ terminates at least when $\mLe$ is $\otimes$-locally finite or $k \in \NN$, by the mentioned invariant of the ``while'' loop of the \SoftStateReductionZ function, it follows that this latter function also terminates at least when $\mLe$ is $\otimes$-locally finite or $k \in \NN$. This invariant also guarantees that the \SoftStateReductionZ function always returns a correct result (when terminates).
\end{proof}

The \SoftStateReduction algorithm first eliminates unreachable or unproductive states. Clearly, the modified \FfA is equivalent to the original one. The algorithm then applies the \SoftStateReductionZ function in two ways: either directly to the input fuzzy automaton or to its reverse, followed by reversing the obtained result. It then returns the better of the two results. We obtain the following theorem, whose proof immediately follows from Lemma~\ref{lemma: IUDHB}. 

\begin{theorem}\label{theorem: KFHSA}
The \SoftStateReduction algorithm terminates at least when $\mLe$ is $\otimes$-locally finite or $k \in \NN$ and always returns a correct result (when terminates).
\end{theorem}

We have implemented the \SoftStateReduction algorithm in Python and made the script publicly available~\cite{SRFA-prog}. To reduce possible numerical errors, the statement~\ref{step: funcRBRI 11} in the \ReductionByRightInvariance function is executed with precision up to~$10^{-12}$ in this implementation (i.e., two fuzzy values that differ by no more than~$10^{-12}$ are treated as equal). \LinhModified{The precision can easily be set by the user, for example to $0$, but analyzing numerical stability is beyond the scope of this work.} The script can be executed as follows:
\begin{verbatim}
      python3 SRFA-prog.py epsilon [options] < input_file > output_file
\end{verbatim}
where {\em epsilon} is a real number in $[0,1]$ and stands for the parameter $\varepsilon$ of the algorithm, {\em input\_$\,$file} specifies the input fuzzy automaton, and options are among:
\begin{description}
\item[k=value] :~where {\em value} is a natural number or {\em infinity} (default),
\item[structure=S] :~where $S \in \{P,H,G,L,N\}$ specifies which among the product (default), Hamacher, G\"odel, \L{}ukasiewicz, and nilpotent structures is used as the underlying residuated lattice, 
\item[sparse] :~using the input format for sparse automata,
\item[verbose] :~enabling the verbose mode.
\end{description}
If the user wants to use a linear and complete residuated lattice that is different from the product, Hamacher, G\"odel, \L{}ukasiewicz, and nilpotent structures, he or she can modify the definition of the {\em tnorm} and {\em residuum} functions in the script. As the G\"odel, \L{}ukasiewicz, and nilpotent structures are $\otimes$-locally finite, it is worth using them together with $\varepsilon = 0$. The file {\em README.txt} in~\cite{SRFA-prog} provides examples of script execution along with details on input formats. When the verbose option is used, execution details are displayed, including two counters: the total number of iterations of the ``while'' loop in the \SoftStateReductionZ function and the total number of executions of the statement~\ref{step: funcRBRI 6} in the \ReductionByRightInvariance function  (the most computationally intensive step). %Examples~\ref{example: HEJSA}, \ref{example: JHDKS} and~\ref{example: HEJSA 2} can be verified using the input file {\em in1.txt} in~\cite{SRFA-prog} with the verbose option enabled. 

\LinhModified{We now present examples illustrating the execution of the \SoftStateReduction algorithm. Note that all these examples and experiments reported in this work can be easily reproduced using our script and the provided input files.}

\begin{figure}[thb]
\begin{center}
\begin{tikzpicture}[->,>=stealth,auto,black]
\tikzset{every state/.style={inner sep=0.06cm,minimum size=0.7cm}}
\tikzstyle{every node}=[font=\footnotesize]
\node[state] (u) {$0$};
\node[state] (v) [node distance=2.0cm, right of=u] {$1$};
\node[state] (w) [node distance=2.0cm, right of=v] {$2$};
\node[state] (up) [node distance=1.5cm, below of=u] {$4$};
\node[state] (vp) [node distance=2.0cm, right of=up] {$5$};
\node[state] (wp) [node distance=2.0cm, right of=vp] {$6$};
\node[state] (tp) [node distance=1.5cm, left of=up] {$3$};
\node[left of=u] (in) {};
\node[right of=w] (out) {};
\node[left of=tp] (inp) {};
\node[right of=wp] (outp) {};
\draw[densely dotted, thick] (in) to node[above]{1} (u);
\draw[densely dotted, thick] (inp) to node[below]{1} (tp);
\draw[densely dotted, thick] (w) to node[above]{0.5} (out);
\draw[densely dotted, thick] (wp) to node[below]{0.5} (outp);
\draw (u) to node[above]{0.6} (v);
\draw (v) to node[above]{0.8} (w);
\draw (w) edge[above,out=145,in=35] node{0.4} (u);
\draw (up) to node[left,xshift=-3pt]{0.5} (v);
\draw (up) to node[below]{0.4} (vp);
\draw (vp) to node[left,xshift=-3pt]{0.7} (w);
\draw (vp) to node[below]{0.8} (wp);
\draw (wp) edge[below,out=-145,in=-35] node{0.4} (up);
\draw (tp) to node[below]{1} (up);
\draw (u) edge[loop above,out=60,in=120,looseness=5] node{$0.9$} (u);
\draw (up) edge[loop below,out=-60,in=-120,looseness=5] node{$0.9$} (up);
\end{tikzpicture}
\caption{An illustration of the fuzzy automaton $\mA$ used in Example~\ref{example: HDIAJ}.\label{fig: HDIAJ}}
\end{center}
\end{figure}

\begin{example}\label{example: HDIAJ}
Consider the fuzzy automaton $\mA = \tuple{\QA, \Sigma, \IA, \TA, \FA}$ over $\Sigma = \{\sigma\}$ that is illustrated in Figure~\ref{fig: HDIAJ} and specified below:

{\footnotesize\[
\QA = \{0,1,2,3,4,5,6\},\quad
\IA =\begin{bmatrix} 1 & 0 & 0 & 1 & 0 & 0 & 0 \end{bmatrix},\quad
\TA_\sigma =\begin{bmatrix}
0.9 & 0.6 & 0 & 0 & 0 & 0 & 0 \\
0 & 0 & 0.8 & 0 & 0 & 0 & 0 \\
0.4 & 0 & 0 & 0 & 0 & 0 & 0 \\
0 & 0 & 0 & 0 & 1 & 0 & 0 \\
0 & 0.5 & 0 & 0 & 0.9 & 0.4 & 0 \\ 
0 & 0 & 0.7 & 0 & 0 & 0 & 0.8 \\ 
0 & 0 & 0 & 0 & 0.4 & 0 & 0 
\end{bmatrix},\quad
\FA =\begin{bmatrix} 0 \\ 0 \\ 0.5 \\ 0 \\ 0 \\ 0 \\ 0.5 \end{bmatrix}.
\]}

\noindent
This fuzzy automaton differs from the one specified in Example~\ref{example: HEJSA} and illustrated in Figure~\ref{fig: HEJSA} in that $\delta_\sigma(0,0) = \delta_\sigma(4,4) = 0.9$. Its specification is stored in the file {\em in2.txt} in~\cite{SRFA-prog}. 

Consider the execution of the \SoftStateReduction algorithm for $\mA$ using the product structure and $\varepsilon = 0.1$ by running
\begin{verbatim}
   python3 SRFA-prog.py 0.1 verbose sparse < in2.txt > out2-0.1-verbose-sparse.txt
\end{verbatim}
The execution details in {\em out2-0.1-verbose-sparse.txt} include, among others, the following information:
\begin{itemize}
\item The call to the \SoftStateReductionZ function in the statement~\ref{step: algSSR 1} of the \SoftStateReduction algorithm produces a fuzzy automaton with four states. The reduction is achieved within the first iteration of the ``while'' loop in that function. Specifically, the statement~\ref{step: funSSRZ 3} of the function produces a fuzzy automaton with five states, denoted here as $\mA_1 = \tuple{\QA_1,\Sigma,\IA_1, \TA_1, \FA_1}$. Subsequently, the statement~\ref{step: funSSRZ 4} further reduces this automaton, producing another fuzzy automaton with four states, denoted here as $\mA_2 = \tuple{\QA_2,\Sigma,\IA_2, \TA_2, \FA_2}$. These automata are specified as follows:

{\footnotesize\[
\begin{array}{lll}
\IA_1 =\begin{bmatrix} 1 & 0.25 & 0.2 & 1 & 1 \end{bmatrix}, & 
(\TA_1)_\sigma =\begin{bmatrix}
0.9 & 0.6 & 0.2 & 0.9 & 1 \\
5/12 & 0.25 & 0.8 & 0.45 & 0.5 \\
5/12 & 0.25 & 0.2 & 0.45 & 0.5 \\
5/6 & 0.25 & 0.2 & 0.9 & 1 \\
0.75 & 0.5 & 0.2 & 0.81 & 0.9 
\end{bmatrix}, & 
\FA_1 =\begin{bmatrix} 0.1 \\ 0.1 \\ 0.5 \\ 0.1 \\ 0.1 \end{bmatrix}, \\
\\
\IA_2 =\begin{bmatrix} 1 & 0.25 & 0.2 & 1 \end{bmatrix}, &
(\TA_2)_\sigma =\begin{bmatrix}
0.9 & 0.6 & 0.2 & 1 \\ 
0.9 & 0.8 & 0.8 & 1 \\
0.9 & 0.8 & 0.8 & 1 \\
0.81 & 0.54 & 0.2 & 0.9
\end{bmatrix}, &
\FA_2 =\begin{bmatrix} 0.1 \\ 1/6 \\ 0.5 \\ 0.1 \end{bmatrix}.
\end{array}
\]}

\item Executing the statement~\ref{step: algSSR 2} of the \SoftStateReduction algorithm results in the same $\mA_2$ as specified above. In particular, the first call to the \ReductionByRightInvariance function (for the reverse of $\mA$) does not perform any reduction. Consequently, the next two calls to this function yield exactly the fuzzy automata $\mA_1$ and $\mA_2$ as previously described. Therefore, the reduction of the reverse of $\mA$ is completed within the first two iterations of the ``while'' loop in the \SoftStateReductionZ function, which then terminates after the third iteration.

\item The execution of the \SoftStateReduction algorithm returns $\mA_2$ as the result. 
As the most computationally intensive step, the statement~\ref{step: funcRBRI 6} in the function \ReductionByRightInvariance is executed totally 180 times in the considered run. 
\end{itemize}

By Theorem~\ref{theorem: KFHSA}, we have $\bL(\mA_2) =_{0.1} \bL(\mA)$ under the product semantics.

An execution of the \SoftStateReduction algorithm for $\mA$ using the product structure and $\varepsilon = 0$ terminates solely due to possible numerical errors and the precision used in the statement~\ref{step: funcRBRI 11} of the \ReductionByRightInvariance function, without performing any state reduction. The most computationally intensive step, the statement~\ref{step: funcRBRI 6} in the \ReductionByRightInvariance function, is executed 28,184 times.

Similarly, executing the \SoftStateReduction algorithm for $\mA$ using the product structure, $\varepsilon = 0$ and \LinhModified{$k \geq 4$} also results in no state reduction. However, when $k = 3$ is used, the algorithm produces a fuzzy automaton with four states, denoted here as $\mA_3 = \tuple{\QA_3,\Sigma,\IA_3, \TA_3, \FA_3}$, which is specified below.

{\footnotesize\[
\IA_3 =\begin{bmatrix} 1 & 0 & 0 & 1 \end{bmatrix},\qquad
(\TA_3)_\sigma =\begin{bmatrix}
0.9 & 0.6 & 0 & 1 \\
0.9 & 0.8 & 0.8 & 1 \\
0.9 & 0.8 & 0.8 & 1 \\ 
0.81 & 0.54 & 0 & 0.9 
\end{bmatrix},\qquad
\FA_3 =\begin{bmatrix} 0 \\ 0 \\ 0.5 \\ 0 \end{bmatrix}.
\]}

\noindent
By Theorem~\ref{theorem: KFHSA}, we have $\bL^{\leq 3}(\mA_3) = \bL^{\leq 3}(\mA)$ under the product semantics.
\myend
\end{example}

\begin{figure}[thb]
\begin{center}
\begin{tikzpicture}[->,>=stealth,auto,black]
\tikzset{every state/.style={inner sep=0.06cm,minimum size=0.7cm}}
\tikzstyle{every node}=[font=\footnotesize]
\node[state] (u) {$0$};
\node[state] (v) [node distance=2.0cm, right of=u] {$1$};
\node[state] (w) [node distance=2.0cm, right of=v] {$2$};
\node[state] (up) [node distance=1.5cm, below of=u] {$3$};
\node[left of=u] (in) {};
\node[right of=w] (out) {};
\node[left of=up] (inp) {};
\draw[densely dotted, thick] (in) to node[above]{1} (u);
\draw[densely dotted, thick] (inp) to node[below]{1} (up);
\draw[densely dotted, thick] (w) to node[above]{0.5} (out);
\draw (u) to node[below]{0.6} (v);
\draw (v) to node[below]{0.8} (w);
\draw (v) to node[right,yshift=-3pt]{0.4} (up);
\draw (w) edge[above,out=150,in=30] node{0.4} (u);
\draw (up) to node[left]{1} (u);
\end{tikzpicture}
\caption{An illustration of the fuzzy automaton $\mA_G$ mentioned in Example~\ref{example: JHJKS}.\label{fig: JHJKS}}
\end{center}
\end{figure}

\begin{example}\label{example: JHJKS}
{\markLinhModified
The \FfA $\mA$ illustrated in Figure~\ref{fig: HEJSA} and discussed in Examples~\ref{example: HEJSA}, \ref{example: JHDKS}, \ref{example: HEJSA 2}, and~\ref{example: HEJSA 3} is specified in the file {\em in1.txt} in~\cite{SRFA-prog}. These examples can be verified using this input file with the verbose option enabled by running our script as follows:
\begin{verbatim}
      python3 SRFA-prog.py 0.1 verbose sparse < in1.txt > out1-0.1.txt
      python3 SRFA-prog.py 0.2 verbose sparse < in1.txt > out1-0.2.txt
      python3 SRFA-prog.py 0 k=3 verbose sparse < in1.txt > out1-0-k3.txt
      python3 SRFA-prog.py 0 k=4 verbose sparse < in1.txt > out1-0-k4.txt
\end{verbatim}
The details are stored in the output files ({\em out1*}).
In particular, these commands produce the \FfAs $\mA_{Z_{0.1},0.1}$, $\mA_{Z_{0.2},0.2}$, $\mA_{Z_{(0,3)},0}$, and $\mA_{Z_{(0,4)},0}$, respectively, which were defined or referred to in Examples~\ref{example: HEJSA 2} and~\ref{example: HEJSA 3}. Recall that $\mA_{Z_{0.1},0.1}$ has five states.
} % \markLinhModified

Executing the \SoftStateReduction algorithm for $\mA$ using the product structure, $\varepsilon = 0$ and any $k \geq 4$ does not reduce the number of states. However, by using $k = 3$ (respectively, $k = 2$) we get a fuzzy automaton with five (respectively, four) states.

Executing the \SoftStateReduction algorithm for $\mA$ using the Hamacher structure and $\varepsilon = 0.22$ by running
\begin{verbatim}
  python3 SRFA-prog.py 0.22 structure=H verbose sparse < in1.txt > out1-0.22-H.txt
\end{verbatim}
produces a fuzzy automaton with five states. However, reducing $\varepsilon$ to $0.21$ prevents any state reduction.

Executing the \SoftStateReduction algorithm for $\mA$ using the G\"odel, nilpotent, and \L{}ukasiewicz structures together with $\varepsilon = 0$, we get fuzzy automata $\mA_G$, $\mA_N$ and $\mA_L$, respectively, where $\mA_G$ and $\mA_N$ have four states, while $\mA_L$ has three states. The fuzzy automaton $\mA_G$ is illustrated in Figure~\ref{fig: JHJKS}. By Theorem~\ref{theorem: KFHSA}, we have $\bL(\mA_G) = \bL(\mA)$ under the G\"odel semantics.
\myend
\end{example}

\begin{figure}[thb]
\begin{center}
\begin{tikzpicture}[->,>=stealth,auto,black]
\tikzset{every state/.style={inner sep=0.05cm,minimum size=0.6cm}}
\tikzstyle{every node}=[font=\scriptsize]
\node[state] (ua) {$0_a$};
\node[state] (va) [node distance=1.6cm, right of=ua] {$1_a$};
\node[state] (wa) [node distance=1.6cm, right of=va] {$2_a$};
\node[state] (upa) [node distance=1.2cm, below of=ua] {$4_a$};
\node[state] (vpa) [node distance=1.6cm, right of=upa] {$5_a$};
\node[state] (wpa) [node distance=1.6cm, right of=vpa] {$6_a$};
\node[state] (tpa) [node distance=1.2cm, left of=upa] {$3_a$};
\node[left of=ua] (ina) {};
\node[left of=tpa] (inpa) {};
\draw[densely dotted, thick] (ina) to node[above]{1} (ua);
\draw[densely dotted, thick] (inpa) to node[below]{1} (tpa);
\draw (ua) to node[above]{0.6} (va);
\draw (va) to node[above]{0.8} (wa);
\draw (wa) edge[above,out=145,in=35] node{0.4} (ua);
\draw (upa) to node[left,xshift=-3pt]{0.5} (va);
\draw (upa) to node[below]{0.4} (vpa);
\draw (vpa) to node[left,xshift=-3pt]{0.7} (wa);
\draw (vpa) to node[below]{0.8} (wpa);
\draw (wpa) edge[below,out=-145,in=-35] node{0.4} (upa);
\draw (tpa) to node[below]{1} (upa);
\node[state] (tpb) [node distance=1.2cm, right of=wpa] {$3_b$};
\node[state] (upb) [node distance=1.2cm, right of=tpb] {$4_b$};
\node[state] (vpb) [node distance=1.6cm, right of=upb] {$5_b$};
\node[state] (wpb) [node distance=1.6cm, right of=vpb] {$6_b$};
\node[state] (ub) [node distance=1.2cm, above of=upb] {$0_b$};
\node[state] (vb) [node distance=1.6cm, right of=ub] {$1_b$};
\node[state] (wb) [node distance=1.6cm, right of=vb] {$2_b$};
\draw[dashed] (ub) to node[above]{0.6} (vb);
\draw[dashed] (vb) to node[above]{0.8} (wb);
\draw[dashed] (wb) edge[above,out=145,in=35] node{0.4} (ub);
\draw[dashed] (upb) to node[left,xshift=-3pt]{0.5} (vb);
\draw[dashed] (upb) to node[below]{0.4} (vpb);
\draw[dashed] (vpb) to node[left,xshift=-3pt]{0.7} (wb);
\draw[dashed] (vpb) to node[below]{0.8} (wpb);
\draw[dashed] (wpb) edge[below,out=-145,in=-35] node{0.4} (upb);
\draw[dashed] (tpb) to node[below]{1} (upb);
\draw (wa) to node[above]{0.5} (ub);
\draw (wpa) to node[below]{0.5} (tpb);
\node[state] (ucx) [node distance=2.4cm, right of=wb] {$0_c$};
\node[state] (tpcx) [node distance=1.2cm, right of=wpb] {$3_c$};
\draw[dashed] (wb) to node[above]{0.5} (ucx);
\draw[dashed] (wpb) to node[below]{0.5} (tpcx);
\node[state] (uc) [node distance=2.2cm, below of=upa] {$0_c$};
\node[state] (vc) [node distance=1.6cm, right of=uc] {$1_c$};
\node[state] (wc) [node distance=1.6cm, right of=vc] {$2_c$};
\node[state] (upc) [node distance=1.2cm, below of=uc] {$4_c$};
\node[state] (vpc) [node distance=1.6cm, right of=upc] {$5_c$};
\node[state] (wpc) [node distance=1.6cm, right of=vpc] {$6_c$};
\node[state] (tpc) [node distance=1.2cm, left of=upc] {$3_c$};
\draw (uc) to node[above]{0.6} (vc);
\draw (vc) to node[above]{0.8} (wc);
\draw (wc) edge[above,out=145,in=35] node{0.4} (uc);
\draw (upc) to node[left,xshift=-3pt]{0.5} (vc);
\draw (upc) to node[below]{0.4} (vpc);
\draw (vpc) to node[left,xshift=-3pt]{0.7} (wc);
\draw (vpc) to node[below]{0.8} (wpc);
\draw (wpc) edge[below,out=-145,in=-35] node{0.4} (upc);
\draw (tpc) to node[below]{1} (upc);
\node[state] (tpd) [node distance=1.2cm, right of=wpc] {$3_d$};
\node[state] (upd) [node distance=1.2cm, right of=tpd] {$4_d$};
\node[state] (vpd) [node distance=1.6cm, right of=upd] {$5_d$};
\node[state] (wpd) [node distance=1.6cm, right of=vpd] {$6_d$};
\node[state] (ud) [node distance=1.2cm, above of=upd] {$0_d$};
\node[state] (vd) [node distance=1.6cm, right of=ud] {$1_d$};
\node[state] (wd) [node distance=1.6cm, right of=vd] {$2_d$};
\draw[dashed] (ud) to node[above]{0.6} (vd);
\draw[dashed] (vd) to node[above]{0.8} (wd);
\draw[dashed] (wd) edge[above,out=145,in=35] node{0.4} (ud);
\draw[dashed] (upd) to node[left,xshift=-3pt]{0.5} (vd);
\draw[dashed] (upd) to node[below]{0.4} (vpd);
\draw[dashed] (vpd) to node[left,xshift=-3pt]{0.7} (wd);
\draw[dashed] (vpd) to node[below]{0.8} (wpd);
\draw[dashed] (wpd) edge[below,out=-145,in=-35] node{0.4} (upd);
\draw[dashed] (tpd) to node[below]{1} (upd);
\draw (wc) to node[above]{0.5} (ud);
\draw (wpc) to node[below]{0.5} (tpd);
\node[right of=wd] (outd) {};
\node[right of=wpd] (outpd) {};
\draw[densely dotted, thick] (wd) to node[above]{0.5} (outd);
\draw[densely dotted, thick] (wpd) to node[below]{0.5} (outpd);
\end{tikzpicture}
\caption{Illustration of the fuzzy automaton $\mA$ used in Example~\ref{example: KJRHW}. 
Normal arrows represent transitions of~$\sigma$, while dashed arrows indicate transitions of~$\varrho$. 
The states $0_c$ and $3_c$ appear twice in order to divide the figure into two parts. 
\label{fig: KJRHW}}
\end{center}
\end{figure}
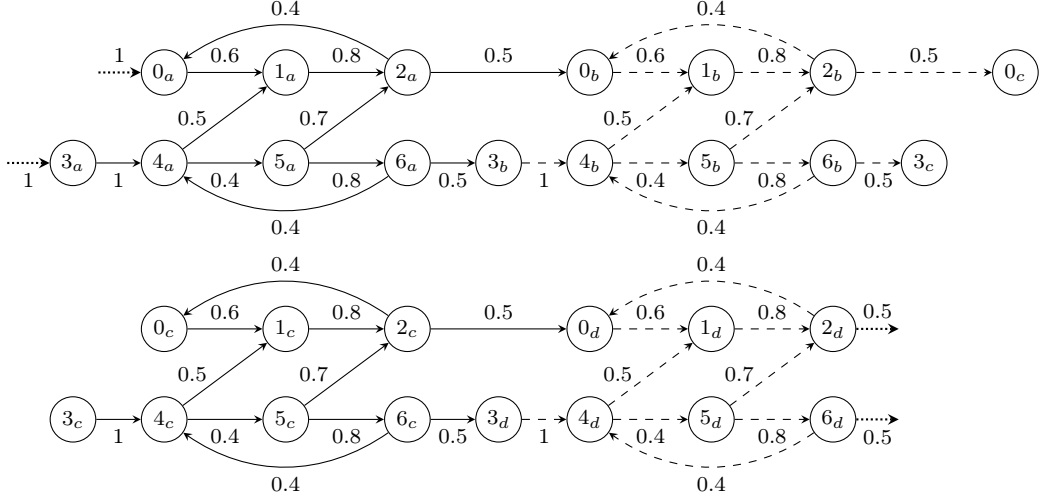

\begin{example}\label{example: KJRHW}
Consider the fuzzy automaton $\mA = \tuple{\QA, \Sigma, \IA, \TA, \FA}$ over $\Sigma = \{\sigma,\varrho\}$, illustrated in Figure~\ref{fig: KJRHW}.  
It has 28 states and is specified in the file {\em in3.txt} in~\cite{SRFA-prog}. Consider executions of the \SoftStateReduction algorithm for $\mA$ and let $\mAp$ denote the resulting fuzzy automaton. The following assertions can be checked by using our script~\cite{SRFA-prog}:
\begin{itemize}
\item If the product structure is used together with $\varepsilon \in \{0.1$, $0.01$, $0.004$, $0.003$, $0.002$, $0.001\}$, $\mAp$~has 6, 19, 24, 25, 27 or 28 states, respectively. 
\item If the product structure and $\varepsilon = 0$ are used, either alone or with $k \geq 14$, no state reduction occurs (i.e., $\mAp = \mA$). 
\item If the G\"odel, \L{}ukasiewicz, or nilpotent structure is used together with $\varepsilon = 0$, $\mAp$ has 25, 3 or 2 states, respectively. 
\item If the Hamacher structure is used together with $\varepsilon \in \{0.2, 0.1\}$, $\mAp$ has 9 or 27 states, respectively. 
\myend
\end{itemize}   
{\markLinhModified It is worth noting that:
\begin{itemize}
\item the algorithms proposed in~\cite{CSIP.10,SCI.14,SCB.18,EPTCS386.6,FSS-D-25-00029} operate under the considered residuated lattices only for $\varepsilon = 0$, and even in this case they do not yield any better reduction for~$\mA$;
\item the algorithms described in~\cite{StanimirovicMC.22,PEEVA20084152,YangL.19,BASAK2002223} can be applied to~$\mA$ when the G\"odel structure is employed;
\item none of the algorithms presented in~\cite{LEI20071413,WU20101635,StanimirovicMC.22,PEEVA20084152,MALIK1999323,CHENG2004439,YangL.19,BASAK2002223} is applicable to~$\mA$ when the product, Hamacher, \L{}ukasiewicz, or nilpotent structure is used.
\myend
\end{itemize}
} % {\markLinhModified
\end{example}

\LinhModified{Additional experimental results from executions of the~\SoftStateReduction algorithm are provided in Appendix~\ref{appendix}.}

\LinhModified{We now analyze the time complexity of the \SoftStateReduction algorithm.} 
Given an \FfA $\mA$ over $\mL$ and an $\varepsilon \in L$ such that $\mLe$ is $\otimes$-locally finite, we define the {\em $\otimes$-local-finiteness degree} of $\mLe$ with respect to $\mA$, denoted by $\lfdeg(\mA,\varepsilon)$, to be the number of all values from $L$ that can be generated using $\otimese$ from the values used for specifying $\mA$. 
For example, if $\mA$ is a crisp automaton, then $\lfdeg(\mA,\varepsilon)$ is 1 or 2. 
As another example, if $\mL$ is the product structure, $\varepsilon \in (0,1]$, there are $c$ numbers from $(0,1)$ used for specifying $\mA$ and $d$ is the biggest among them, then $\lfdeg(\mA,\varepsilon) \leq 2 + c^{\log_d \varepsilon}$. 

\begin{proposition}
Let $\mA = \tuple{\QA, \Sigma, \IA, \TA, \FA}$ be an \FfA with $|\QA| = n$ and $|\Sigma| = s$. 
When $\mLe$ is $\otimes$-locally finite, let $l = \lfdeg(\mA,\varepsilon)$. 
Then, the time complexity of the \SoftStateReduction algorithm is within 
\begin{itemize}
\item $O(n^2 s^k(n + k \log{s}) + n^5 s)$ when $k \in \NN$,
\item $O(n^3 s l^n \log{l} + n^5 s)$ when $\mLe$ is $\otimes$-locally finite.
\end{itemize}
Under the assumptions that $s$ is a constant, $s > 1$, $k \geq n$, $l > 1$ and $\log{l}$ is within $O(n)$, the time complexity order is:
\begin{itemize}
\item $O(k n^2 s^k)$ when $k \in \NN$,
\item $O(n^4 l^n)$ when $\mLe$ is $\otimes$-locally finite, 
\item $O(\min(k n^2 s^k, n^4 l^n))$ when $k \in \NN$ and $\mLe$ is $\otimes$-locally finite.
\end{itemize}
\end{proposition}

\begin{proof}
\LinhModified{First, note that the statement~\ref{step: algSSR 0} of the \SoftStateReduction algorithm, which eliminates unreachable or unproductive states, can be performed in $O(sn^2)$ time, where $|\QA| = n$ and $|\Sigma| = s$.} 

\LinhModified{Consider} the time complexity of the \ReductionByRightInvariance function. 
The statements~\ref{step: funcRBRI 1} and~\ref{step: funcRBRI 2} run in constant time. 
Consider the ``repeat'' loop in the statements \ref{step: funcRBRI 3}-\ref{step: funcRBRI 9}. 
We have $|\mF| \leq s^n$ when $k \in \NN$, and $|\mF| \leq l^n$ when $\mLe$ is $\otimes$-locally finite. 
The time used by this loop is dominated by the time used for the statements~\ref{step: funcRBRI 6} and~\ref{step: funcRBRI 7}. The number of times these two statements are executed is bounded by $s^k$ when $k \in \NN$, and by $s |\mF| \leq s l^n$ when $\mLe$ is $\otimes$-locally finite. 
The time needed for an execution of the statements~\ref{step: funcRBRI 6} and~\ref{step: funcRBRI 7} is of the order $O(n^2 + n\log{|\mF|})$. 
Thus, the time used by the ``repeat'' loop is of the order:
\begin{itemize}
\item $O(s^k(n^2 + n\log{|\mF|}))$ when $k \in \NN$,
\item $O(s l^n(n^2 + n\log{|\mF|}))$ when $\mLe$ is $\otimes$-locally finite.
\end{itemize}
The statement~\ref{step: funcRBRI 10} runs in time of the order $O(n^2|\mF|)$. 
The statements~\ref{step: funcRBRI 11}-\ref{step: funcRBRI 15} run in time of the order $O(n^4 s)$. 
Summing up, the time complexity of the \ReductionByRightInvariance function is of the order:
\begin{itemize}
\item $O(s^k(n^2 + n\log{|\mF|}) + n^2|\mF| + n^4 s) = O(n s^k(n + k \log{s}) + n^4 s)$ when $k \in \NN$,
\item $O(s l^n(n^2 + n\log{|\mF|}) + n^2|\mF| + n^4 s) = O(n^2 s l^n \log{l} + n^4 s)$ when $\mLe$ is $\otimes$-locally finite.
\end{itemize}

Now consider the \SoftStateReductionZ function. The ``while'' loop executes at most $n$ times. Each iteration of this loop runs in time of the order as given above for the \ReductionByRightInvariance function. Therefore, the time complexity of the \SoftStateReductionZ function, and consequently also of the \SoftStateReduction algorithm, is of the order stated in the theorem.
\end{proof}

%-----------------------------------------------------------------------------------

\section{Conclusions}\label{sec: conc}

{\markLinhModified
In this work, we introduced \emph{soft state reduction} for fuzzy finite automata over linear and complete residuated lattices, based on the novel concepts of \emph{approximate residuated lattices} and \emph{approximate invariances}. By treating fuzzy values below a threshold $\varepsilon$ as negligible and, optionally, restricting attention to words of bounded length $k$, we developed a flexible framework that achieves state reduction while preserving language equivalence up to a controlled tolerance. The central result is the \SoftStateReduction algorithm, which alternates reductions based on right and left approximate invariances and is guaranteed to terminate whenever the approximate residuated lattice is $\otimes$-locally finite or a finite length bound $k$ is specified.

Compared with existing approaches, our framework demonstrates clear advantages. In particular, methods from \cite{MALIK1999323,BASAK2002223,CHENG2004439,LI20071423,Belohlvek2009OnAM,LITFS2015,Halamish2015,YangL.19,YANG202172,ShamsizadehZG24,LEI20071413,PEEVA20084152,WU20101635,StanimirovicMC.22} are not applicable within the algebraic framework of our work. Although the exact reduction approach~\cite{CSIP.10,SCI.14,SCB.18} and approximate reduction approach based on bounded length~$k$~\cite{EPTCS386.6,FSS-D-25-00029} are applicable in principle, they often fail to achieve a reduction in many scenarios, as demonstrated in our work. In particular, the latter approach becomes impractical for larger values of $k$, limiting its usefulness in real-world applications. In contrast, our approach overcomes these limitations by ensuring substantial reductions for fuzzy finite automata over general linear and complete residuated lattices, even in cases where earlier methods either fail to terminate or yield no reduction.

Taken together, these contributions establish a general and practical reduction method for fuzzy finite automata that extends beyond the limitations of previous techniques. The proposed framework not only advances the foundations of state reduction but also opens the way for applications in formal verification, learning, and approximate reasoning with fuzzy models.
}

%-----------------------------------------------------------------------------------

\bibliographystyle{elsarticle-num}
\bibliography{references}

%-----------------------------------------------------------------------------------
%\newpage
\appendix
\section{Additional experimental results}
\label{appendix}

In this appendix, we present additional experimental results from executions of the~\SoftStateReduction algorithm using our implementation~\cite{SRFA-prog}, as described in Section~\ref{sec: alg}. We begin by analyzing how the algorithm behaves under various parameter settings.

\begin{example}\label{example: additional-in1-txt}
	
	Reconsider the \FfA $\mA = \tuple{\QA, \Sigma, \IA, \TA, \FA}$ over $\Sigma = \{\sigma\}$ specified in Example~\ref{example: HEJSA} and reused in Example~\ref{example: JHJKS}. Recall that it has seven states. As previously mentioned, its full specification is provided in the file~{\em in1.txt}, which is available in~\cite{SRFA-prog}. 
	
	Table~\ref{tab:reduction-results} presents the number of states remaining after executing the \SoftStateReduction algorithm for~$\mA$ under different configurations of the approximation threshold~$\varepsilon$, the word length bound~$k$, and the underlying residuated lattice. It covers five residuated lattices (the product~(P), Hamacher~(H), G\"odel~(G), \L{}ukasiewicz~(L), and nilpotent~(N) structures), three values of $\varepsilon$ (0, 0.1 and 0.2) and three options for $k$: $k = 2$, $k = 3$, or $k \geq 4$ (including $k = \infty$). A discussion on this table is given below.
	
	\begin{itemize}
		\item Using the product structure, the number of remaining states after the reduction is directly influenced by the values of~$\varepsilon$ and~$k$. At~$\varepsilon=0$, the reduction is limited, especially as~$k$ increases. Specifically, for $k \geq4$, the algorithm does not perform any reduction, resulting in seven states, the same as the original \FfA. However, when~$k$ is reduced to~2 or~3, the number of states is reduced to four or five, respectively. At~$\varepsilon=0.1$, the reduction process becomes more effective: the \FfA is reduced to four states at $k=2$, and to five states for $k\geq 3$. At~$\varepsilon=0.2$, the reduction effect is even stronger: regardless of the value of~$k$, the \FfA is consistently reduced to just four states. 
		
		\item Using the Hamacher structure, the reduction pattern remains unchanged across the three values of~$\varepsilon$. In particular, the \FfA is reduced to four states at~$k=2$, to five states at~$k=3$, and no reduction is observed for~$k\geq4$, where the number of states remains at seven.
		
		\item Using the G\"odel structure gives a steady and moderate reduction in all cases. At~$k=2$, the number of states is always reduced to three. For~$k\geq3$, it gives four states, no matter which value of~$\varepsilon$ is used. This shows that the result does not change much when~$\varepsilon$ varies.
		
		\item Using the \L{}ukasiewicz (respectively, nilpotent) structure always reduces the \FfA to three (respectively, four) states, no matter which values of~$\varepsilon$ or~$k$ are used. This means that the reduction result remains the same across all the mentioned settings. 
		\myend
	\end{itemize} 
	
	\begin{table}[t]
		\centering
		\renewcommand{\arraystretch}{1.2}
		\begin{tabular}{|c|ccc|ccc|ccc|}
			\hline
			\multirow{2}{*}{\textbf{Structure}} & \multicolumn{3}{c|}{\boldmath$\varepsilon = 0$} & \multicolumn{3}{c|}{\boldmath$\varepsilon = 0.1$} & \multicolumn{3}{c|}{\boldmath$\varepsilon = 0.2$} \\
			\cline{2-10}
			& $k=2$ & $k=3$ & $k\geq4$ & $k=2$ & $k=3$ & $k\geq4$ & $k=2$ & $k=3$ & $k\geq4$ \\
			\hline
			P & 4 & 5 & 7 & 4 & 5 & 5 & 4 & 4 & 4 \\
			H & 4 & 5 & 7 & 4 & 5 & 7 & 4 & 5 & 7 \\
			G & 3 & 4 & 4 & 3 & 4 & 4 & 3 & 4 & 4 \\
			L & 3 & 3 & 3 & 3 & 3 & 3 & 3 & 3 & 3 \\
			N & 4 & 4 & 4 & 4 & 4 & 4 & 4 & 4 & 4 \\
			\hline
		\end{tabular}
		\caption{The number of remaining states after executing the \SoftStateReduction algorithm for the \FfA~$\mA$ considered in Example~\ref{example: additional-in1-txt}, using different sets of parameter values.\label{tab:reduction-results}}
	\end{table}
	
\end{example}

\begin{figure}[h!]
	\centering
	\begin{tikzpicture}
		\begin{axis}[
			every axis plot/.append style={
				nodes near coords,
				every node near coord/.append style={font=\tiny, yshift=5pt}
			},
			width=13cm,
			height=7cm,
			xlabel={$k$},
			ylabel={Number of remaining states},
			legend pos=north west,
			grid=major,
			xtick={2,4,6,8,10,12,14},
			ytick={0,5,...,30},
			ymin=0, ymax=30,
			mark options={scale=1.2},
			legend style={font=\small},
			]
			
			% Data for G
			\addplot+[smooth, mark=*] coordinates {
				(2,4) (4,7) (6,12) (8,16) (10,20) (12,25) (14,25)
			};
			\addlegendentry{G\"odel}

			% Data for N
			\addplot+[smooth, mark=square*] coordinates {
				(2,2) (4,2) (6,2) (8,2) (10,2) (12,2) (14,2)
			};
			\addlegendentry{nilpotent}
			
			% Data for L
			\addplot+[smooth, mark=triangle*] coordinates {
				(2,3) (4,3) (6,3) (8,3) (10,3) (12,3) (14,3)
			};
			\addlegendentry{\L{}ukasiewicz}

			% Data for H
			\addplot+[smooth, mark=diamond*] coordinates {
				(2,5) (4,8) (6,12) (8,18) (10,20) (12,26) (14,28)
			};
			\addlegendentry{Hamacher}
			
			% Data for P
			\addplot+[smooth, mark=star] coordinates {
				(2,5) (4,8) (6,14) (8,18) (10,20) (12,26) (14,28)
			};
			\addlegendentry{product}
			
		\end{axis}
	\end{tikzpicture}
	\caption{The number of remaining states after executing the \SoftStateReduction algorithm for the \FfA~$\mA$ from Example~\ref{example: additional-in3-txt}, using various values of~$k$ and different residuated lattices with~$\varepsilon = 0$.}
	\label{fig: additional-in3-txt}
\end{figure}

\begin{figure}[h!]
	\centering
	
	% ----- eps = 0.1 -----
	\begin{subfigure}[t]{0.48\textwidth}
		\centering
		\begin{tikzpicture}
			\begin{axis}[
				every axis plot/.append style={
					nodes near coords,
					every node near coord/.append style={font=\tiny, yshift=5pt}
				},
				width=\textwidth,
				height=6cm,
				title={$\varepsilon = 0.1$},
				xlabel={$k$}, ylabel={Number of remaining states},
				legend style={at={(0.02,0.98)}, anchor=north west, font=\small, draw=none, fill=none},
				xtick={2,4,6,8,10,12,14},
				ymin=0, ymax=30,
				grid=major,
				mark options={scale=1.2}
				]
				\addplot+[mark=*] coordinates {(2,4) (4,7) (6,12) (8,16) (10,20) (12,25) (14,25)};
				\addlegendentry{G\"odel}
				
				\addplot+[mark=square*] coordinates {(2,2) (4,2) (6,2) (8,2) (10,2) (12,2) (14,2)};
				\addlegendentry{nilpotent}
				
				\addplot+[mark=triangle*] coordinates {(2,3) (4,3) (6,3) (8,3) (10,3) (12,3) (14,3)};
				\addlegendentry{\L{}ukasiewicz}
				
				\addplot+[mark=diamond*] coordinates {(2,5) (4,10) (6,14) (8,18) (10,22) (12,26) (14,27)};
				\addlegendentry{Hamacher}
				
				\addplot+[mark=star] coordinates {(2,5) (4,6) (6,6) (8,6) (10,6) (12,6) (14,6)};
				\addlegendentry{product}
			\end{axis}
		\end{tikzpicture}
	\end{subfigure}
	%
	% ----- eps = 0.2 -----
	\begin{subfigure}[t]{0.48\textwidth}
		\centering
		\begin{tikzpicture}
			\begin{axis}[
				every axis plot/.append style={
					nodes near coords,
					every node near coord/.append style={font=\tiny, yshift=5pt}
				},
				width=\textwidth,
				height=6cm,
				title={$\varepsilon = 0.2$},
				xlabel={$k$}, ylabel={Number of remaining states},
				legend style={at={(0.02,0.98)}, anchor=north west, font=\small, draw=none, fill=none},
				xtick={2,4,6,8,10,12,14},
				ymin=0, ymax=30,
				grid=major,
				mark options={scale=1.2}
				]
				\addplot+[mark=*] coordinates {(2,4) (4,7) (6,12) (8,16) (10,20) (12,25) (14,25)};
				\addlegendentry{G\"odel}
				
				\addplot+[mark=square*] coordinates {(2,2) (4,2) (6,2) (8,2) (10,2) (12,2) (14,2)};
				\addlegendentry{nilpotent}
				
				\addplot+[mark=triangle*] coordinates {(2,3) (4,3) (6,3) (8,3) (10,3) (12,3) (14,3)};
				\addlegendentry{\L{}ukasiewicz}
				
				\addplot+[mark=diamond*] coordinates {(2,5) (4,9) (6,9) (8,9) (10,9) (12,9) (14,9)};
				\addlegendentry{Hamacher}
				
				\addplot+[mark=star] coordinates {(2,4) (4,4) (6,4) (8,4) (10,4) (12,4) (14,4)};
				\addlegendentry{product}
			\end{axis}
		\end{tikzpicture}
	\end{subfigure}
	
	\vspace{0.5cm}
	
	% ----- eps = 0.3 -----
	\begin{subfigure}[t]{0.48\textwidth}
		\centering
		\begin{tikzpicture}
			\begin{axis}[
				every axis plot/.append style={
					nodes near coords,
					every node near coord/.append style={font=\tiny, yshift=5pt}
				},
				width=\textwidth,
				height=6cm,
				title={$\varepsilon = 0.3$},
				xlabel={$k$}, ylabel={Number of remaining states},
				legend style={at={(0.02,0.98)}, anchor=north west, font=\small, draw=none, fill=none},
				xtick={2,4,6,8,10,12,14},
				ymin=0, ymax=30,
				grid=major,
				mark options={scale=1.2}
				]
				\addplot+[mark=*] coordinates {(2,4) (4,7) (6,12) (8,16) (10,20) (12,25) (14,25)};
				\addlegendentry{G\"odel}
				
				\addplot+[mark=square*, opacity=0.6] coordinates {(2,2) (4,2) (6,2) (8,2) (10,2) (12,2) (14,2)};
				\addlegendentry{nilpotent}
				
				\addplot+[mark=triangle*] coordinates {(2,2) (4,2) (6,2) (8,2) (10,2) (12,2) (14,2)};
				\addlegendentry{\L{}ukasiewicz}
				
				\addplot+[mark=diamond*] coordinates {(2,5) (4,5) (6,5) (8,5) (10,5) (12,5) (14,5)};
				\addlegendentry{Hamacher}
				
				\addplot+[mark=star] coordinates {(2,3) (4,3) (6,3) (8,3) (10,3) (12,3) (14,3)};
				\addlegendentry{product}
			\end{axis}
		\end{tikzpicture}
	\end{subfigure}
	%
	% ----- eps = 0.4 -----
	\begin{subfigure}[t]{0.48\textwidth}
		\centering
		\begin{tikzpicture}
			\begin{axis}[
				every axis plot/.append style={
					nodes near coords,
					every node near coord/.append style={font=\tiny, yshift=5pt}
				},
				width=\textwidth,
				height=6cm,
				title={$\varepsilon = 0.4$},
				xlabel={$k$}, ylabel={Number of remaining states},
				legend style={at={(0.02,0.98)}, anchor=north west, font=\small, draw=none, fill=none},
				xtick={2,4,6,8,10,12,14},
				ymin=0, ymax=30,
				grid=major,
				mark options={scale=1.2}
				]
				\addplot+[mark=*] coordinates {(2,4) (4,6) (6,8) (8,10) (10,12) (12,14) (14,14)};
				\addlegendentry{G\"odel}
				
				\addplot+[mark=square*, opacity=0.6] coordinates {(2,2) (4,2) (6,2) (8,2) (10,2) (12,2) (14,2)};
				\addlegendentry{nilpotent}
				
				\addplot+[mark=triangle*] coordinates {(2,2) (4,2) (6,2) (8,2) (10,2) (12,2) (14,2)};
				\addlegendentry{\L{}ukasiewicz}
				
				\addplot+[mark=diamond*] coordinates {(2,2) (4,2) (6,2) (8,2) (10,2) (12,2) (14,2)};
				\addlegendentry{Hamacher}
				
				\addplot+[mark=star] coordinates {(2,2) (4,2) (6,2) (8,2) (10,2) (12,2) (14,2)};
				\addlegendentry{product}
			\end{axis}
		\end{tikzpicture}
	\end{subfigure}
	
	\caption{The number of remaining states after executing the \SoftStateReduction algorithm for the \FfA~$\mA$ from Example~\ref{example: additional-in3-txt}, using various values of~$k$ and different residuated lattices with~$\varepsilon \in \{0.1, 0.2, 0.3, 0.4\}$.}
	\label{fig: additional-in3-txt-2}
\end{figure}
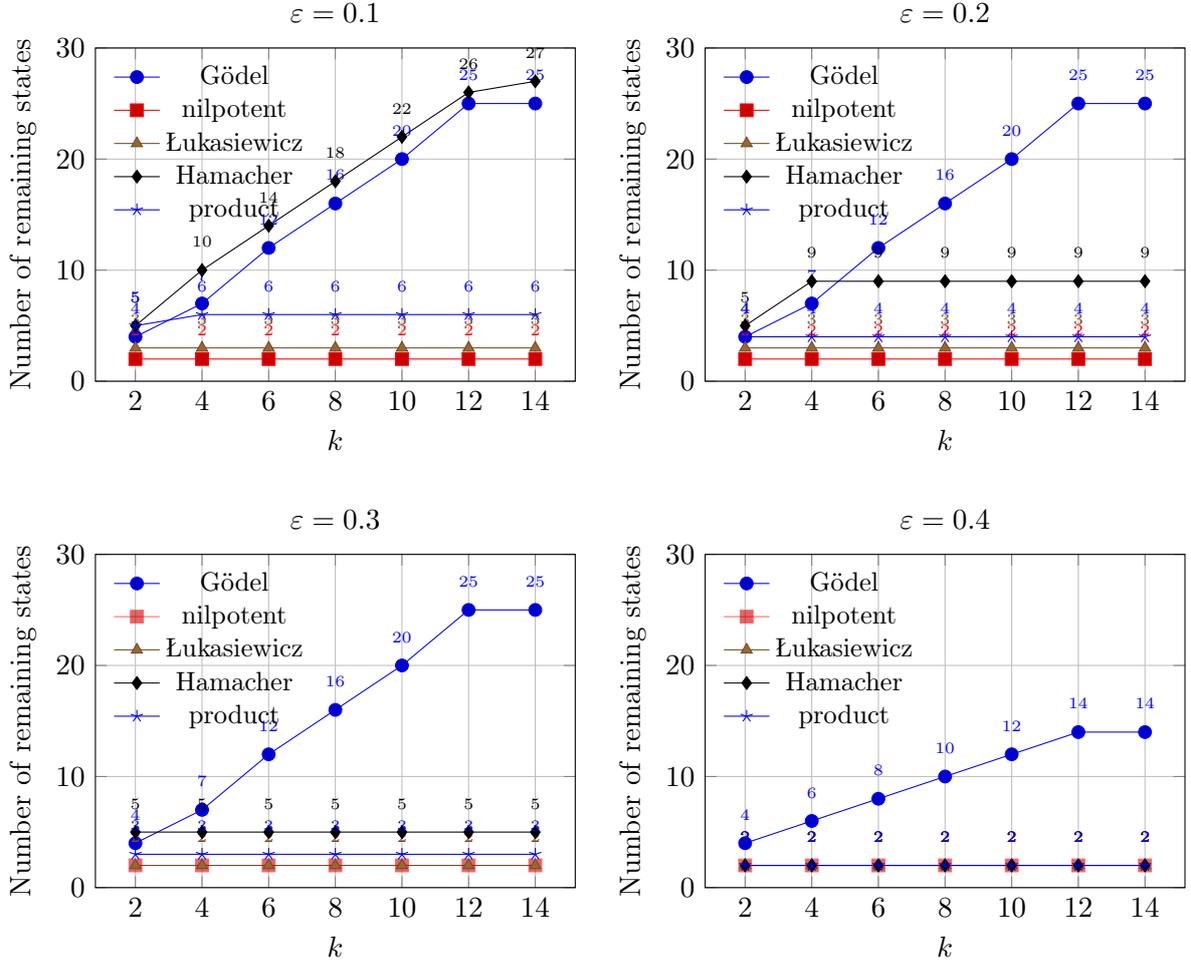

\begin{example} \label{example: additional-in3-txt} 
	Reconsider the \FfA $\mA = \tuple{\QA, \Sigma, \IA, \TA, \FA}$ over~\mbox{$\Sigma = \{\sigma,\varrho\}$} specified in Example~\ref{example: KJRHW} and depicted in Figure~\ref{fig: KJRHW}. Recall that $\mA$ has 28 states and its full specification is provided in the file~{\em in3.txt}, available in~\cite{SRFA-prog}. In this example, we examine the effects of executing the \SoftStateReduction algorithm for this \FfA under various parameter settings, using five residuated lattices (as in the previous example), five values of~$\varepsilon$ (0, 0.1, 0.2, 0.3, and 0.4) and a range of values for~$k$: 2, 3, \ldots, 13, and $\geq\!14$ (including $\infty$). Details are given below.
	\begin{itemize}
		\item Using the G\"odel structure with~$\varepsilon$ between~0 and~0.3, the number of remaining states gradually increases with~$k$, starting from~4 states at $k=2$, reaching~23 states at $k=11$, and then stabilizing at~25 states for~$k \geq 12$. However, when~$\varepsilon$ is~0.4, the number of remaining states also increases with~$k$ (starting from~4), but it stabilizes at a lower value of~14 states for~$k \geq 11$. This clearly shows the influence of~$\varepsilon$, especially for larger values of~$k$: the number of states remains at~25 for~$\varepsilon$ between~0 and~0.3, but significantly drops to~14 when~$\varepsilon$ increases to~0.4 (for $k \geq 11$). The reduction when using the G\"odel structure is thus sensitive to both the parameters~$k$ and~$\varepsilon$. Overall, the number of remaining states ranges from~4 to~25 when the G\"odel structure is used.
		
		\item Using the Hamacher structure with~$\varepsilon$ equal to 0 (respectively, 0.1), the number of remaining states also gradually increases with~$k$, starting from~5 states at $k = 2$, and stabilizing at~28 states for~$k \geq 14$ (respectively, at~27 states for~$k \geq 13$). For~$\varepsilon=0.2$, the number of remaining states is~5, 7, 9, and~8 at $k$ equal to~2, 3, 4, and~5, respectively, and stabilizes at~9 for~$k \geq 6$. Despite Corollary~\ref{cor: JHFRK2}\eqref{item: JHFRK2 2}, there is a decrease from~9 to~8 when $k$ increases from~4 to~5. This is because the~\SoftStateReduction algorithm is significantly more sophisticated than merely computing $\ReductionByRightInvariance(\mA, \varepsilon, k)$. Using a larger~$\varepsilon$ (0.3 or~0.4), the number of remaining states stays low regardless of~$k$: it is~5 states at~$\varepsilon = 0.3$ for all $k \geq 2$, and~2 states at~$\varepsilon = 0.4$ for all $2 \leq k \leq 1000$. This clearly demonstrates the strong impact of~$\varepsilon$, particularly for larger values of~$k$: the number of remaining states is up to 28 when $\varepsilon = 0$, but decreases sharply as~$\varepsilon$ increases. The reduction when using the Hamacher structure thus proves to be highly sensitive to both~$k$ and~$\varepsilon$. Larger values of~$\varepsilon$ lead to significantly greater reductions in the number of remaining states. Overall, the number of remaining states ranges from~3 to~28 when the Hamacher structure is used.
		
		\item Using the product structure, the experimental results concerning the number of remaining states are dependent on both the parameters~$k$ and~$\varepsilon$. Notably,~$\varepsilon$ has a particularly strong influence on these results. At~$\varepsilon=0$, the number of remaining states gradually increases as~$k$ increases. Starting at~5 states at~$k=2$, this number rises, reaching~27 states at $k=13$. It then stabilizes at the highest observed value of~28 states for all~$k\geq14$. When~$\varepsilon\in\{0.1,0.2,0.3,0.4\}$, the behavior changes drastically compared to~$\varepsilon=0$, resulting in significantly better state reductions. For~$\varepsilon=0.1$, the number of remaining states is~5 at $k = 2$, then jumps to and stabilizes at~6 for~$k \geq 3$. For higher values of~$\varepsilon$, such a number is even smaller. It is~4 states at~$\varepsilon=0.2$, 3~states at~$\varepsilon=0.3$, and reaches a minimum of~2 states at~$\varepsilon=0.4$ (for all~$k \geq2 $). Overall, the number of remaining states is highly sensitive to both the parameters~$k$ and~$\varepsilon$ when the product structure is used. Its range spans from~2 to~28.
		
		\item Using the \L{}ukasiewicz and nilpotent structure results in highly stable reductions, notably showing no dependence on the parameter~$k$. In the case of the \L{}ukasiewicz structure, the number of remaining states depends only on~$\varepsilon$, consistently yielding~3 states when~$\varepsilon \in \{0, 0.1, 0.2\}$, before dropping to a minimum of~2 states for~$\varepsilon \geq 0.3$. Using the nilpotent structure results in the most stable and minimal reduction overall. The number of remaining states consistently stays at~2, showing complete insensitivity to both~$k$ and~$\varepsilon$.
	\end{itemize}
	
	The number of remaining states after executing the \SoftStateReduction algorithm for the considered \FfA~$\mA$ is depicted in Figures~\ref{fig: additional-in3-txt} and~\ref{fig: additional-in3-txt-2}. These figures illustrate the results when using various residuated lattices with different values of~$k$ and~$\varepsilon \in \{0, 0.1, 0.2, 0.3, 0.4\}$.
	\myend
	
\end{example}

\begin{example}\label{examp:16-states}
	We consider an \FfA consisting of~16 states, fully specified in the file~{\em in4.txt} of~\cite{SRFA-prog}, with a graphical representation given in Figure~\ref{fig:16-states}. Following the setup of the previous example, this analysis evaluates the behavior of our \SoftStateReduction algorithm applied to this \FfA using the five previously mentioned residuated lattices, across four approximation thresholds $\varepsilon \in \{0, 0.1, 0.2, 0.3\}$ and various word length bounds~$k$. The experimental results are described in detail below.
		
	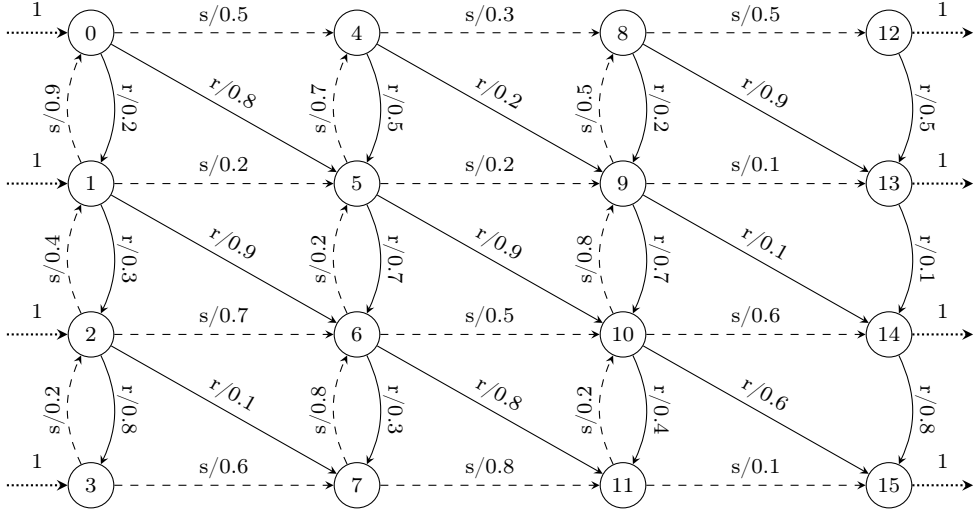
\begin{figure}[h!]
		\centering
		\begin{tikzpicture}[
			->, >=stealth, 
			every node/.style={circle, draw, minimum size=0.6cm, inner sep=1pt},
			font=\scriptsize,
			edge label/.style={midway, sloped, fill=none, draw=none}
			]
			
			% ==== Nodes: 4 columns × 4 rows ====
			\foreach \i/\x/\y in {
				0/0/6, 1/0/4, 2/0/2, 3/0/0,
				4/3.5/6, 5/3.5/4, 6/3.5/2, 7/3.5/0,
				8/7/6, 9/7/4, 10/7/2, 11/7/0,
				12/10.5/6, 13/10.5/4, 14/10.5/2, 15/10.5/0
			}{
				\node (n\i) at (\x,\y) {\i};
			}
			
			% ==== Initial arrows ====
			\foreach \i in {0,1,2,3} {
				\draw[densely dotted, thick, ->] ($(n\i.west)+(-0.8,0)$) -- (n\i.west) node[edge label, above] {1};
			}
			
			% ==== Final arrows ====
			\foreach \i in {12,13,14,15} {
				\draw[densely dotted, thick, ->] (n\i.east) -- ($(n\i.east)+(0.8,0)$) node[edge label, above] {1};
			}
			
			% --- Transition r ---
			\draw[bend left=25] (n0) to node[edge label, sloped, pos=0.5, yshift=5pt] {r/0.2} (n1);
			\draw[bend left=25] (n1) to node[edge label, sloped, pos=0.5, yshift=5pt] {r/0.3} (n2);
			\draw[bend left=25] (n2) to node[edge label, sloped, pos=0.5, yshift=5pt] {r/0.8} (n3);
			%----
			\draw[bend left=25] (n4) to node[edge label, sloped, pos=0.5, yshift=5pt] {r/0.5} (n5);
			\draw[bend left=25] (n5) to node[edge label, sloped, pos=0.5, yshift=5pt] {r/0.7} (n6);
			\draw[bend left=25] (n6) to node[edge label, sloped, pos=0.5, yshift=5pt] {r/0.3} (n7);
			%----
			\draw[bend left=25] (n8) to node[edge label, sloped, pos=0.5, yshift=5pt] {r/0.2} (n9);
			\draw[bend left=25] (n9) to node[edge label, sloped, pos=0.5, yshift=5pt] {r/0.7} (n10);
			\draw[bend left=25] (n10) to node[edge label, sloped, pos=0.5, yshift=5pt] {r/0.4} (n11);
			%-----
			\draw[bend left=25] (n12) to node[edge label, sloped, pos=0.5, yshift=5pt] {r/0.5} (n13);
			\draw[bend left=25] (n13) to node[edge label, sloped, pos=0.5, yshift=5pt] {r/0.1} (n14);
			\draw[bend left=25] (n14) to node[edge label, sloped, pos=0.5, yshift=5pt] {r/0.8} (n15);
			%------------------------------
			\draw (n0) to node[edge label, above, yshift=-5pt] {r/0.8} (n5);
			\draw (n1) to node[edge label, above, yshift=-5pt] {r/0.9} (n6);
			\draw (n2) to node[edge label, above, yshift=-5pt] {r/0.1} (n7);
			%-----
			\draw (n4) to node[edge label, above, yshift=-5pt] {r/0.2} (n9);
			\draw (n5) to node[edge label, above, yshift=-5pt] {r/0.9} (n10);
			\draw (n6) to node[edge label, above, yshift=-5pt] {r/0.8} (n11);
			%-----
			\draw (n8) to node[edge label, above, yshift=-5pt] {r/0.9} (n13);
			\draw (n9) to node[edge label, above, yshift=-5pt] {r/0.1} (n14);
			\draw (n10) to node[edge label, above, yshift=-5pt] {r/0.6} (n15);
			%-----
			% --- Transition s ---
			\draw[bend left=25, dashed] (n1) to node[edge label, sloped, pos=0.5, yshift=5pt] {s/0.9} (n0);
			\draw[bend left=25, dashed] (n2) to node[edge label, sloped, pos=0.5, yshift=5pt] {s/0.4} (n1);
			\draw[bend left=25, dashed] (n3) to node[edge label, sloped, pos=0.5, yshift=5pt] {s/0.2} (n2);
			%-----
			\draw[bend left=25, dashed] (n5) to node[edge label, sloped, pos=0.5, yshift=5pt] {s/0.7} (n4);
			\draw[bend left=25, dashed] (n6) to node[edge label, sloped, pos=0.5, yshift=5pt] {s/0.2} (n5);
			\draw[bend left=25, dashed] (n7) to node[edge label, sloped, pos=0.5, yshift=5pt] {s/0.8} (n6);
			%-----
			\draw[bend left=25, dashed] (n9) to node[edge label, sloped, pos=0.5, yshift=5pt] {s/0.5} (n8);
			\draw[bend left=25, dashed] (n10) to node[edge label, sloped, pos=0.5, yshift=5pt] {s/0.8} (n9);
			\draw[bend left=25, dashed] (n11) to node[edge label, sloped, pos=0.5, yshift=5pt] {s/0.2} (n10);
			%-----
			%-----
			
			%------------------------------
			\draw[dashed] (n0) to node[edge label, above, yshift=-5pt] {s/0.5} (n4);
			\draw[dashed] (n1) to node[edge label, above, yshift=-5pt] {s/0.2} (n5);
			\draw[dashed] (n2) to node[edge label, above, yshift=-5pt] {s/0.7} (n6);
			\draw[dashed] (n3) to node[edge label, above, yshift=-5pt] {s/0.6} (n7);
			%-----
			\draw[dashed] (n4) to node[edge label, above, yshift=-5pt] {s/0.3} (n8);
			\draw[dashed] (n5) to node[edge label, above, yshift=-5pt] {s/0.2} (n9);
			\draw[dashed] (n6) to node[edge label, above, yshift=-5pt] {s/0.5} (n10);
			\draw[dashed] (n7) to node[edge label, above, yshift=-5pt] {s/0.8} (n11);
			%-----
			\draw[dashed] (n8) to node[edge label, above, yshift=-5pt] {s/0.5} (n12);
			\draw[dashed] (n9) to node[edge label, above, yshift=-5pt] {s/0.1} (n13);
			\draw[dashed] (n10) to node[edge label, above, yshift=-5pt] {s/0.6} (n14);
			\draw[dashed] (n11) to node[edge label, above, yshift=-5pt] {s/0.1} (n15);
			%-----
			%-----
			
		\end{tikzpicture}
		\caption{An illustration of the fuzzy automaton mentioned in Example~\ref{examp:16-states}.\label{fig:16-states}}
	\end{figure}

	\begin{itemize}
	\item Using the G\"odel structure, the reduction pattern depends on both the word length bound~$k$ and the approximation threshold~$\varepsilon$. For~$k = 3$, no reduction occurs at~$\varepsilon = 0$ (the \FfA remains at 16 states), but the number of states gradually decreases: to 15 at~$\varepsilon = 0.1$, to 14 at~$\varepsilon = 0.2$, and finally to 11 at~$\varepsilon = 0.3$. In contrast, for~$k \geq 4$, no reduction is observed -- the \FfA consistently remains at 16 states across all values of~$\varepsilon$.
	
	\item Using the \L{}ukasiewicz structure allows for stronger and more consistent reductions. For~$k = 3$, the number of states decreases steadily from~12 at~$\varepsilon = 0$ to~10,~8, and~6 when~$\varepsilon$ increases. A similar trend is observed for~$k \geq 4$, starting from~14 at $\varepsilon = 0$ and also reaching~6 at~$\varepsilon = 0.3$. This indicates that the \L{}ukasiewicz structure supports more flexible state merging, particularly when using higher approximation thresholds.
	
	\item Using the nilpotent structure, reduction is relatively low. For~$k = 3$, the number of states drops slightly from~13 at~$\varepsilon = 0$ to~12 at~$\varepsilon = 0.2$ and remains unchanged at~$\varepsilon = 0.3$. For~$k = 4$, only a single state is removed, with the \FfA consistently reduced to~15 states regardless of~$\varepsilon$. For~$k \geq 5$, all~16 states are retained across all values of~$\varepsilon$. 
	
	\item Using the Hamacher structure allows for a limited but consistent reduction pattern. At~$k = 3$, the number of states drops from~16 to~14 as~$\varepsilon$ increases. A similar trend appears for~$k \geq 4$, though less pronounced, with a reduction to~15 states occurring only at~$\varepsilon = 0.3$. This suggests that the Hamacher structure requires relatively high approximation thresholds to enable state merging.
	
	\item Lastly, using the product structure allows for a smooth and predictable reduction trend. At~$k = 3$, the number of states gradually drops from~16 to~14 and then to~10 as~$\varepsilon$ increases. For~$k \geq 4$, the reduction still occurs but is more modest, reaching~11 states at~$\varepsilon = 0.3$. This behavior suggests that the product structure strikes a good balance between allowing flexible approximation and maintaining structural integrity.
	\end{itemize}
	
Our experiments with the \SoftStateReduction algorithm demonstrate varied results in fuzzy automaton minimization. The underlying residuated lattice plays a key role in determining the number of remaining states. The parameters~$k$ and~$\varepsilon$ influence the reduction, but their impact varies across different residuated lattices. In this example, the \L{}ukasiewicz and nilpotent structures proved to be the most effective, allowing for the greatest reduction.
\myend
\end{example}

\begin{example}\label{Example-range-0-0.5}
Consider the \FfA~$\mA = \tuple{\QA, \Sigma, \IA, \TA, \FA}$ defined below, with eight states over the alphabet~$\Sigma = \{\sigma, \varrho\}$. All fuzzy membership values (including those in the initial state vector~$\IA$, the final state vector~$\FA$, and the transition matrices~$\TA_\sigma$ and~$\TA_\varrho$) have been randomly selected within the interval~$[0.0, 0.5]$. Moreover, each row in every transition matrix contains at least one non-zero entry, ensuring that every state has at least one fuzzy transition for each input symbol. The full specification of $\mA$ is stored in the file~\emph{in5.txt}, available in~\cite{SRFA-prog}.

{\footnotesize\[
	\QA = \{0,1,2,3,4,5,6,7\},\quad
	\IA =\begin{bmatrix} 0.3 & 0.2 & 0 & 0 & 0 & 0 & 0 & 0 \end{bmatrix},
	\]}
{\footnotesize\[
	\TA_\sigma =\begin{bmatrix}
		0.2 & 0.5 & 0.3 & 0.1 & 0   & 0.3 & 0   & 0   \\
		0.2 & 0.4 & 0.5 & 0   & 0.3 & 0.4 & 0   & 0   \\
		0.5 & 0.2 & 0.3 & 0.4 & 0   & 0.2 & 0   & 0.2 \\
		0.3 & 0   & 0.5 & 0.3 & 0.4 & 0.1 & 0   & 0.4 \\
		0   & 0.1 & 0   & 0.5 & 0.2 & 0   & 0.4 & 0.2 \\
		0.4 & 0.2 & 0.1 & 0.3 & 0   & 0.3 & 0.2 & 0.4 \\
		0   & 0   & 0   & 0   & 0.2 & 0.4 & 0.5 & 0   \\
		0   & 0   & 0.1 & 0.3 & 0.2 & 0.5 & 0   & 0.3
	\end{bmatrix},\quad
	\TA_\varrho =\begin{bmatrix}
		0.3 & 0.2 & 0.4 & 0.2 & 0   & 0.3 & 0   & 0   \\
		0.5 & 0.1 & 0.4 & 0   & 0.3 & 0.4 & 0   & 0   \\
		0.2 & 0.4 & 0.2 & 0.5 & 0   & 0.1 & 0.4 & 0.3 \\
		0.3 & 0   & 0.2 & 0.3 & 0.1 & 0.4 & 0   & 0.5 \\
		0   & 0.4 & 0   & 0.2 & 0.4 & 0.1 & 0.3 & 0.4 \\
		0.1 & 0.4 & 0.2 & 0.5 & 0.3 & 0.3 & 0   & 0.2 \\
		0   & 0   & 0.2 & 0   & 0.1 & 0   & 0.5 & 0   \\
		0   & 0   & 0.4 & 0.5 & 0.1 & 0.2 & 0   & 0.3
	\end{bmatrix},\quad
	\FA =\begin{bmatrix} 0 \\ 0 \\ 0 \\ 0.4 \\ 0.3 \\ 0.1 \\ 0 \\ 0 \end{bmatrix}.
	\]}

Executing the \SoftStateReduction algorithm for the \FfA~$\mA$ with $\varepsilon \in \{0, 0.01, 0.1, 0.2, 0.3\}$, $k = \infty$, and one of the five aforementioned residuated lattices, we obtained the following results.

\begin{itemize}
	\item Using the G\"odel structure, the number of remaining states is reduced to four at~$\varepsilon = 0$, and remains unchanged at~$\varepsilon = 0.1$, suggesting that some states are already similar. At~$\varepsilon = 0.2$, it drops to three, indicating that a higher approximation threshold allows further merging.
	
	\item Using the \L{}ukasiewicz structure, the \FfA is reduced to three states at~$\varepsilon=0$, and further to two states for~\mbox{$\varepsilon\in\{0.1, 0.2\}$}. This shows that it responds quickly to approximation, with effective merging even for low values of~$\varepsilon$.
		
	\item Using the nilpotent structure, the \FfA undergoes a significant reduction: the number of states decreases from eight to one, regardless of the value of~$\varepsilon$.
	
	\item Using the Hamacher structure results in no reduction when~$\varepsilon \leq 0.1$, with all eight states remaining. However, at~$\varepsilon = 0.2$, the number of remaining states drops to four, indicating that with this structure, reduction occurs only when the approximation threshold is sufficiently large.
	
	\item Using the product structure enables a gradual reduction. There is no reduction when~$\varepsilon \leq 0.01$, but the \FfA reduces to three states at~$\varepsilon=0.1$, and to two states at~$\varepsilon=0.2$. This suggests that the structure supports smooth merging as the approximation threshold increases, balancing precision and simplification.
\end{itemize}

At~$\varepsilon = 0.3$, the \FfA is reduced to a single state, regardless of the chosen residuated lattice. This indicates that such a high approximation threshold causes all states to behave similarly, effectively collapsing the automaton's structure. While this represents the strongest possible reduction, it also risks eliminating important behavioral distinctions. For this reason, practical analyses often rely on smaller~$\varepsilon$ values, where meaningful state reduction can be achieved without compromising the automaton's essential structure.
\myend
\end{example}	

%=====================================================

Overall, the analysis in Examples~\ref{example: additional-in1-txt}-\ref{Example-range-0-0.5} on the use of different residuated lattices reveals a clear distinction in their behavior regarding state reduction. Using the nilpotent and \L{}ukasiewicz structures consistently results in highly effective reductions, often approaching or reaching the minimum possible number of states. When using them, the number of remaining states shows little to no sensitivity to the parameter~$k$, depending only on~$\varepsilon$ in some cases. In contrast, the G\"odel, Hamacher, and product structures exhibit more variable state reduction, where the number of remaining states tends to increase with the parameter~$k$. However, using these structures, increasing the value of~$\varepsilon$ generally leads to a significantly greater degree of reduction. Therefore, the choice of the underlying residuated lattice is a critical factor determining both the level of state reduction achievable and how sensitive that reduction is to the tuning parameters~$k$ and~$\varepsilon$.
	
In general, small values of the parameter~$k$ (e.g., values smaller than the number of states in the considered \FfA) tend to be overly restrictive. Furthermore, the choice of an appropriate residuated lattice for a given \FfA depends on the application and is usually predetermined. The \L{}ukasiewicz and nilpotent structures appear less suitable for \FfAs because the acceptance degrees of long words tend to be either~1 or~0 when these structures are used. As mentioned earlier, since the G\"odel, \L{}ukasiewicz, and nilpotent structures are $\otimes$-locally finite, it is advisable to use them in combination with~$\varepsilon = 0$. Also note that, since the Hamacher t-norm is greater than or equal to the product t-norm, it generally requires a higher threshold~$\varepsilon$ to achieve a comparable level of reduction.

%=====================================================

The next two examples provide further insight into the performance of the~\SoftStateReduction algorithm when either the product or Hamacher structure -- both of which are not $\otimes$-locally finite -- is used.

\begin{example}
This example uses the \FfA~$\mA = \tuple{\QA, \Sigma, \IA, \TA, \FA}$ defined below, which is a modified version of that in Example~\ref{Example-range-0-0.5}. Specifically, the non-zero entries in the transition matrices~$\TA_\sigma$ and~$\TA_\varrho$ are now selected from the interval~$[0.6, 0.9]$, instead of~$[0.1, 0.5]$. The full specification of $\mA$ is stored in the file~\emph{in6.txt}, available in~\cite{SRFA-prog}.
	
	{\footnotesize\[
		\QA = \{0,1,2,3,4,5,6,7\},\quad
		\IA =\begin{bmatrix} 1 & 1 & 0 & 0 & 0 & 0 & 0 & 0 \end{bmatrix},
		\]}
	{\footnotesize\[
		\TA_\sigma =\begin{bmatrix}
			0.9 & 0.7 & 0.6 & 0.8 & 0   & 0.6 & 0   & 0   \\
			0.7 & 0.9 & 0.8 & 0   & 0.7 & 0.8 & 0   & 0   \\
			0.8 & 0.7 & 0.9 & 0.6 & 0   & 0.7 & 0   & 0.8 \\
			0.7 & 0   & 0.8 & 0.9 & 0.7 & 0.6 & 0   & 0.8 \\
			0   & 0.6 & 0   & 0.8 & 0.9 & 0   & 0.7 & 0.8 \\
			0.6 & 0.6 & 0.8 & 0.7 & 0   & 0.9 & 0.8 & 0.7 \\
			0   & 0   & 0   & 0   & 0.6 & 0.7 & 0.9 & 0   \\
			0   & 0   & 0.6 & 0.7 & 0.8 & 0.9 & 0   & 0.9
		\end{bmatrix},\quad
		\TA_\varrho =\begin{bmatrix}
			0.9 & 0.6 & 0.7 & 0.7 & 0   & 0.8 & 0   & 0   \\
			0.8 & 0.9 & 0.7 & 0   & 0.6 & 0.8 & 0   & 0   \\
			0.7 & 0.8 & 0.9 & 0.8 & 0   & 0.6 & 0.7 & 0.8 \\
			0.8 & 0   & 0.6 & 0.9 & 0.7 & 0.8 & 0   & 0.7 \\
			0   & 0.7 & 0   & 0.6 & 0.9 & 0.7 & 0.6 & 0.8 \\
			0.6 & 0.8 & 0.7 & 0.8 & 0.6 & 0.9 & 0   & 0.7 \\
			0   & 0   & 0.6 & 0   & 0.6 & 0   & 0.9 & 0   \\
			0   & 0   & 0.7 & 0.8 & 0.6 & 0.7 & 0   & 0.9
		\end{bmatrix},\quad
		\FA =\begin{bmatrix} 0 \\ 0 \\ 0 \\ 0 \\ 1 \\ 1 \\ 0 \\ 0 \end{bmatrix}.
		\]}

Theoretically, the \FfA~$\mA$ is equivalent to a two-state \FfA over any linear and complete residuated lattice.

Executing the \SoftStateReduction algorithm for~$\mA$ with $\varepsilon = 0$, $k=\infty$, and either the G\"odel, \L{}ukasiewicz, or nilpotent structure (which are $\otimes$-locally finite) as the underlying residuated lattice, the number of states is quickly reduced to two. The same result is obtained using the product structure. In the latter case, at $\varepsilon = 10^{-6}$, the algorithm terminates after 16,068 executions of the statement~\ref{step: funcRBRI 6} in the \ReductionByRightInvariance function. When $\varepsilon$ is set to~0, it terminates solely due to possible numerical errors and the precision used in the statement~\ref{step: funcRBRI 11} in that function, after 2,520,112 executions of the mentioned statement~6.

Using the Hamacher structure with $\varepsilon = 0.001$ and $k = \infty$, the algorithm reduces~$\mA$ to three states, after 7,497,026 executions of the statement~\ref{step: funcRBRI 6} in the \ReductionByRightInvariance function. If the precision used in the statement~\ref{step: funcRBRI 11} in that function is relaxed from~$10^{-12}$ to~$10^{-9}$, and all other parameters remain unchanged, the algorithm reduces $\mA$ to two states, after 7,466,972 executions of the mentioned statement~6.

Executing the \SoftStateReduction algorithm for~$\mA$ with $\varepsilon = 0$, $k=1000$, and the product (respectively, Hamacher) structure as the underlying residuated lattice, the number of states is reduced to two, after 266,126 (respectively, 291,306) executions of the statement~\ref{step: funcRBRI 6} in the \ReductionByRightInvariance function.
\myend
\end{example}	

%=====================================================

\begin{example}
Consider the \FfA~$\mA = \tuple{\QA, \Sigma, \IA, \TA, \FA}$ defined below, with eight states over the alphabet~$\Sigma = \{\sigma, \varrho\}$. All fuzzy membership values (including those in the initial state vector~$\IA$, the final state vector~$\FA$, and the transition matrices~$\TA_\sigma$ and~$\TA_\varrho$) are either 0 or  0.5. This uniform configuration ensures that all fuzzy transitions and state memberships are defined with equal weight. The full specification of $\mA$ is stored in the file~\emph{in7.txt}, available in~\cite{SRFA-prog}.

{\footnotesize\[
	\QA = \{0,1,2,3,4,5,6,7\},\quad
	\IA =\begin{bmatrix} 0.5 & 0.5 & 0 & 0 & 0 & 0 & 0 & 0  \end{bmatrix},
	\]}
{\footnotesize\[
	\TA_\sigma =\begin{bmatrix}
		0.5 & 0   & 0   & 0   & 0   & 0   & 0   & 0   \\
		0   & 0.5 & 0   & 0   & 0   & 0   & 0   & 0   \\
		0   & 0   & 0.5 & 0   & 0   & 0   & 0   & 0   \\
		0   & 0   & 0   & 0.5 & 0   & 0   & 0   & 0   \\
		0   & 0   & 0   & 0   & 0.5 & 0   & 0   & 0   \\
		0   & 0   & 0   & 0   & 0   & 0.5 & 0   & 0   \\
		0   & 0   & 0   & 0   & 0   & 0   & 0.5 & 0   \\
		0   & 0   & 0   & 0   & 0   & 0   & 0   & 0.5 \\
	\end{bmatrix},\quad
	\TA_\varrho =\begin{bmatrix}
		0   & 0   & 0   & 0   & 0   & 0   & 0   & 0.5 \\
		0   & 0   & 0   & 0   & 0   & 0   & 0.5 & 0   \\
		0   & 0   & 0   & 0   & 0   & 0.5 & 0   & 0   \\
		0   & 0   & 0   & 0   & 0.5 & 0   & 0   & 0   \\
		0   & 0   & 0   & 0.5 & 0   & 0   & 0   & 0   \\
		0   & 0   & 0.5 & 0   & 0   & 0   & 0   & 0   \\
		0   & 0.5 & 0   & 0   & 0   & 0   & 0   & 0   \\
		0.5 & 0   & 0   & 0   & 0   & 0   & 0   & 0 \\
	\end{bmatrix},\quad
	\FA =\begin{bmatrix} 0 \\ 0 \\ 0 \\ 0 \\ 0 \\ 0 \\ 0.5 \\ 0.5 \end{bmatrix}.
	\]}	

It can be seen that the states of~$\mA$ form four strongly connected components, which are mutually disconnected. Each component consists of a pair of states. In addition, the states~$2$, $3$, $4$, and~$5$ are unreachable and inproductive, and therefore redundant.

Executing the \SoftStateReduction algorithm for~$\mA$ with $\varepsilon = 0$, $k=\infty$, and either the G\"odel, \L{}ukasiewicz, or nilpotent structure as the underlying residuated lattice, the number of states is quickly reduced to two. The same result is obtained using the product structure. In the latter case, at $\varepsilon = 10^{-6}$, the algorithm terminates after 608 executions of the statement~\ref{step: funcRBRI 6} in the \ReductionByRightInvariance function. When $\varepsilon$ is set to~0, it terminates solely due to possible numerical errors and the precision used in the statement~\ref{step: funcRBRI 11} in that function, after 34,368 executions of the mentioned statement~6.

Using the Hamacher structure together with $\varepsilon$ equal to $10^{-3}$ (respectively, $10^{-4}$ or $10^{-5}$) and $k = \infty$, the algorithm reduces~$\mA$ to two states, after 31,936 (respectively, 319,968 or 3,199,936) executions of the statement~\ref{step: funcRBRI 6} in the \ReductionByRightInvariance function. 

Executing the \SoftStateReduction algorithm for~$\mA$ with $\varepsilon = 0$, $k=1000$, and the product or Hamacher structure as the underlying residuated lattice, the number of states is reduced to two, after 31,984 executions of the statement~\ref{step: funcRBRI 6} in the \ReductionByRightInvariance function.

An interesting aspect of this example is that if the algorithm is modified by removing the statement~\ref{step: algSSR 0} (which eliminates unreachable or unproductive states), then, when executed for~$\mA$ using $\varepsilon = 0$, $k = \infty$, and the G\"odel structure, it reduces the number of states only to three (instead of two).
\myend
\end{example}

%=====================================================

The next two examples -- the final ones -- consider \FfAs used in the literature~\cite{StanimirovicMC.22,BASAK2002223}. They demonstrate that, even without approximation (i.e., using $\varepsilon = 0$ and $k = \infty$), our \SoftStateReduction algorithm achieves better state reduction than the approaches proposed in those works for these \FfAs.

\begin{example}\label{example-StanimirovicMC.22}
	Consider the \FfA $\mA$ over~$\Sigma = \{x\}$ given in~\cite[Example~V.4]{StanimirovicMC.22}, which has 6 states. Its full details are also provided in the file~{\em in8.txt}, available in~\cite{SRFA-prog}. 
	In~\cite{StanimirovicMC.22}, the authors considered reducing this \FfA over the G\"odel structure. They showed that this \FfA cannot be reduced using forward and backward simulations, but that applying 0.7-approximate forward or backward simulation yields a reduction to three states.
	
	Similar to the previous examples, we executed the \SoftStateReduction algorithm for $\mA$ under various parameter settings. Specifically, we considered the five residuated lattices mentioned earlier, three values of~$\varepsilon$ (0, 0.1, and 0.2), and a range of values for~$k$ (including~$\infty$). Details of the experimental results are summarized below:
	
	\begin{itemize}
		\item Using the G\"odel or nilpotent structure results in a significant reduction, consistently reducing the \FfA from six to two states, regardless of the values of~$k$ and~$\varepsilon$. 
		\item Using the \L{}ukasiewicz, Hamacher, or product structure yields no reduction at all, with the number of states remaining six across all configurations.
	\end{itemize}
	
	Remarkably, for the considered \FfA, the combination of $\varepsilon = 0.2$ and $k = 1$ yields the same number of remaining states as the configuration of $\varepsilon = 0$ and $k = \infty$, regardless of the chosen residuated lattice.
	\myend
\end{example}

\begin{example}
		Consider the \FfA $\mA$ over the alphabet~$\Sigma = \{\sigma_0, \sigma_1\}$ given in~\cite[Section~7]{BASAK2002223}, which consists of nine states. Its complete specification is also provided in the file~\emph{in9.txt}, available in~\cite{SRFA-prog}. 
		In~\cite{BASAK2002223}, the authors aimed to reduce the size of this~\FfA by applying the substitution property partition method (SP-partition). During the reduction process, states with similar behavior are gradually grouped together. This results in three final blocks: one corresponding to the accepting states and two corresponding to the non-accepting states. These blocks form the states of the reduced \FfA, which therefore has only three states. The reduced \FfA preserves the original fuzzy behavior while significantly simplifying the structure. This example shows that SP-partitioning is a practical and effective approach for minimizing \FfAs without altering their semantics.
		
		%Using the same experimental settings as in Example~\ref{example-StanimirovicMC.22}, we executed our \SoftStateReduction algorithm for the~\FfA~$\mA$. We observed a clear and consistent result: the fuzzy automaton was always reduced to just two states. This was a notable finding because the number of states did not change regardless of how we adjusted the values of~$\varepsilon$ or~$k$, or how we chose the underlying residuated lattice among the five aforementioned ones. This shows that, for this example, the reduction process is stable and does not depend on the specific settings of these parameters.
        
        Using the same experimental settings as in Example~\ref{example-StanimirovicMC.22}, we executed our~\SoftStateReduction algorithm for the~\FfA~$\mA$. We observed a clear and consistent outcome: the fuzzy automaton was always reduced to just two states. Notably, this result remained unchanged regardless of the values of~$\varepsilon$ and~$k$, or the choice of the underlying residuated lattice among the five considered. This indicates that, for this example, the reduction process is stable and insensitive to these parameter settings.
		\myend
\end{example}
	
%=====================================================

\end{document}